\documentclass[12pt]{article}
\usepackage[english]{babel}
\usepackage{verbatim} 
\usepackage{slashed} 
\usepackage{amsthm}
\usepackage{amssymb}
\usepackage{amsmath}
\usepackage[title]{appendix}
\usepackage{caption}
\usepackage{enumitem}
\usepackage{float}  
\usepackage{graphicx} 
\usepackage{epstopdf} 
\usepackage{hyperref} 
\usepackage{listings} 
\usepackage{xcolor} 
\usepackage[utf8]{inputenc} 
\usepackage[T1]{fontenc}
\usepackage[all]{xy} 

\hypersetup{
pdfpagemode=UseOutlines,
pdftitle={Categorical Approach to S-duality}, 
pdfauthor={M. Caorsi}, 
pdfcreator={M. Caorsi},
pdfkeywords={quantum, categories, representations, quivers, SUSY, duality, Ginzburg, cluster}, 
bookmarks=true,
bookmarksopen=true,
bookmarksnumbered=true
}
\newtheorem{theorem}{Theorem}[section] 
\newtheorem{lemma}[theorem]{Lemma} 
\newtheorem{proposition}[theorem]{Proposition}

\newtheorem{fact}{Fact}
\newtheorem{corollary}[theorem]{Corollary}
\newtheorem*{problem}{Problem}
\newtheorem*{principle}{RG principle}
\theoremstyle{definition}
\newtheorem{definition}{Definition}[section]
\newtheorem{example}{Example}[section]
\theoremstyle{remark}
\newtheorem{remark}{Remark}

\usepackage{array}
\usepackage[bottom=3cm,top=3cm,right=2.5cm, left=2.5cm]{geometry}
\usepackage{booktabs} 
\usepackage{multirow} 
\usepackage{lipsum} 
\newcommand{\vev}[1]{\left< #1 \right>} 
\newcommand{\ket}[1]{\left| #1 \right>} 
\newcommand{\bra}[1]{\left< #1 \right|} 

\newcommand{\ca}{\mathcal A}
\newcommand{\cp}{\mathcal{P}}
\newcommand{\D}{\mathcal D}
\newcommand{\A}{\mathcal A}
\newcommand{\Z}{\mathbb Z}
\newcommand{\Q}{\mathbb Q}
\newcommand{\C}{\mathbb C}
\newcommand{\R}{\mathbb R}

\def\CircleArrowleft{\ensuremath{%
  \reflectbox{\rotatebox[origin=c]{270}{$\circlearrowleft$}}}}
\def\CircleArrowright{\ensuremath{%
  \reflectbox{\rotatebox[origin=c]{90}{$\circlearrowright$}}}}

\renewenvironment{abstract}
 {\small
  \begin{center}
  \bfseries \abstractname\vspace{-.5em}\vspace{0pt}
  \end{center}
  \list{}{
    \setlength{\leftmargin}{.0cm}%
    \setlength{\rightmargin}{\leftmargin}%
  }%
  \item\relax}
 {\endlist}

\title{\Huge Categorical Webs and $S$-duality\\ in 4d $\mathcal{N}=2$ QFT}

\author{\\ \\ Matteo Caorsi\footnote{e-mail: {\tt matteocao@gmail.com}} \ and Sergio Cecotti\footnote{e-mail: {\tt cecotti@sissa.it}}\\
\centerline{SISSA, via Bonomea 265, I-34100 Trieste, ITALY}}
\date{}

\begin{document}
\begin{titlepage}
\clearpage\maketitle
\thispagestyle{empty}
\begin{abstract}
We  review the categorical approach to the BPS sector of a 4d $\mathcal{N}=2$ QFT, clarifying many tricky issues and presenting a few  novel results.\\
 To a given $\mathcal{N}=2$ QFT one associates  several triangle categories: they describe various kinds of BPS objects from different physical viewpoints (e.g.\! IR versus UV). These diverse categories are related by a web of exact functors expressing physical relations between the various objects/pictures. A basic theme of this review is the emphasis on the full web of categories, rather than on  what we can learn from a single description. A second general theme is viewing the cluster category as a sort of `categorification' of 't Hooft's theory of quantum phases for a 4d non-Abelian gauge theory.\\  
The $S$-duality group is best described as the auto-equivalences of the full web of categories. This viewpoint leads to a combinatorial algorithm to search for $S$-dualities of the  given $\mathcal{N}=2$ theory. If the ranks of the gauge and flavor groups are not too big, the algorithm may be effectively run on a laptop. This viewpoint also leads to a clearer view of $3d$ mirror symmetry.\\
For class $\mathcal{S}$ theories, all the relevant triangle categories may also be constructed in terms of geometric objects on the Gaiotto curve, and we present the dictionary between triangle categories and  the WKB approach of GMN. We also review how the VEV's of UV line operators are related to cluster characters.
\end{abstract}

\vfill
July, 2017
\newpage
\end{titlepage}

\tableofcontents
\newpage
\section{Introduction and Overview}

The BPS objects of a supersymmetric theory are naturally described in terms of ($\C$-linear) triangle categories \cite{aspinwall2009dirichlet} and their stability conditions \cite{bridgeland2007stability}.
The BPS sector of a given physical theory $\mathcal{T}$ is described by a plurality of different triangle categories $\mathfrak{T}_{(a)}$ depending on:\footnote{\ The index $a$ take values in some index set $I$.}
\begin{itemize}
\item[a)] the class of BPS objects (particles, branes, local or non-local operators,...)\! we are interested in;
\item[b)] the physical picture (fundamental UV theory, IR effective theory,...);
\item[c)] the particular engineering of $\mathcal{T}$ in QFT/string/M-/F-theory. 
\end{itemize} 
The diverse BPS categories $\mathfrak{T}_{(a)}$ are related by a web of exact functors, $\mathfrak{T}_{(a)}\xrightarrow{\mathsf{c}_{(a,b)}} \mathfrak{T}_{(b)}$, which express physical consistency conditions between the different physical pictures and BPS objects. The simplest instance is given by two different engineerings of the same theory: the duality $\mathcal{T}\leftrightarrow \mathcal{T}^\prime$ induces equivalences of triangle categories
$\mathfrak{T}_{(a)}\xrightarrow{\mathsf{d}_{(a)}} \mathfrak{T}^\prime_{(a)}$ for all objects and all physical descriptions $a\in I$. An example  is mirror symmetry between 
IIA and IIB string theories compactified on a pair of mirror Calabi-Yau 3-folds, $\mathcal{M}$, $\mathcal{M}^\vee$ which induces on the BPS branes \emph{homological mirror symmetry}, that is, the equivalences of triangle categories \cite{kapustin2008homological}
$$D^b(\mathsf{Coh}\,\mathcal{M})\cong D^b(\mathsf{Fuk}\,\mathcal{M}^\vee),\qquad
D^b(\mathsf{Coh}\,\mathcal{M}^\vee)\cong D^b(\mathsf{Fuk}\,\mathcal{M}).$$
In the same way, the functor relating the IR and UV descriptions of the BPS sector may be seen as \textit{homological Renormalization Group,} while the functor relating particles and branes may be seen as describing properties of the combined system. 

A duality induces a family of equivalences $\mathsf{d}_{(a)}$, one for each category $\mathfrak{T}_{(a)}$, and these equivalences should be compatible with the functors $\mathsf{c}_{(a,b)}$, that is, they should give an equivalence of the full web of categories and functors. Our philosophy is that the study of equivalences of the full functorial web is a very efficient tool  to detect  dualities.
We shall focus on the case of $4d$ $\mathcal{N}=2$ QFTs, but the strategy has general validity. We are particularly concerned with $S$-dualities, i.e.\! auto-dualities of the theory $\mathcal{T}$ which act non trivially on the UV degrees of freedom.

Building on previous work by several people\footnote{\ References to previous work are provided in the appropriate sections of the paper.}, we present our proposal for the triangle categories describing different BPS objects,  both from the UV and IR points of view,
and study the functors relating them. This leads, in particular, to a categorical understanding of the $S$-duality
 groups and of the \textsc{vev} of UV line operators. The categorical language unifies in a systematic way 
all aspects of the BPS physics, and leads to new powerful techniques to compute \textsc{susy} protected quantities in $\mathcal{N}=2$ $4d$
theories. We check in many explicit examples that the results obtained from this more abstract viewpoint reproduce the ones obtained by more traditional techniques.
However the categorical approach may also be used to tackle problems which look too hard for other techniques. 

\paragraph{Main triangle categories and functors.}
The basic example of a web of functors relating distinct BPS categories for $4d$ $\mathcal{N}=4$ QFT is the 
 following exact sequence of triangle categories (Theorem 5.6 of \cite{keller2011cluster}):
\begin{equation}\label{exactse}
0 \to D^b\Gamma \xrightarrow{\ \mathsf{s}\ } \mathfrak{Per}\,\Gamma \xrightarrow{\ \mathsf{r}\ } \mathcal{C}(\Gamma)\to 0,
\end{equation}
where (see \S.\,\ref{sec:math} for precise definitions and details):
\begin{itemize}
\item $\Gamma$ is the \emph{Ginzburg algebra} \cite{ginzburg2006calabi} of a  quiver with superpotential \cite{aspinwall2006superpotentials} 
associated to the $\mathcal{N}=2$ theory at hand;
\item 
$D^b\Gamma$ is the \emph{bounded derived} category of $\Gamma$. $D^b\Gamma$ may be seen as a ``universal envelope'' of 
the categories describing,
in the deep IR, the BPS
particle spectrum in the several BPS chambers. To discuss states in the IR  we need to fix a Coulomb vacuum $u$; this datum
defines a stability condition $Z_u$ on $D^b\Gamma$. The category which describes the BPS particles in the $u$ vacuum is 
the subcategory of $D^b\Gamma$ consisting of objects which are semi-stable for $Z_u$. The BPS particles
 arise from (the quantization of the moduli
 of) the simple objects in this subcategory. Its Grothendieck group $K_0(D^b\Gamma)$ is identified with 
 the Abelian group of the IR additive 
conserved quantum numbers (electric, magnetic, and flavor charges) which take value in the lattice $\Lambda\cong K_0(D^b\Gamma)$. $D^b\Gamma$ 
is a \textit{3-Calabi-Yau} (3-CY)\footnote{\ See \S.\,\ref{sec:math} for precise definitions. \textit{Informally,} a triangle category is $k$-CY iff it behaves as the derived category of coherent sheaves, $D^b\,\mathsf{coh}\,\mathcal{M}_k$, on a Calabi-Yau $k$-fold $\mathcal{M}_k$. } triangle category,
which implies that 
its Euler form 
$$
\chi(X,Y)\equiv \sum_{k\in \Z} (-1)^k\; \dim \mathrm{Hom}_{D^b\Gamma}(X,Y[k]),\qquad X,\;Y\in D^b\Gamma
$$
is a \emph{skew-symmetric} form $\Lambda\times \Lambda\to \Z$ whose physical meaning is the Dirac electro-magnetic pairing between the
 charges $[X]$, $[Y]\in \Lambda$ carried by the states associated to the stable objects $X$, $Y\in D^b\Gamma$;
\item $\mathcal{C}(\Gamma)$ is the \emph{cluster category} of $\Gamma$  which describes\footnote{\ This is slightly imprecise. Properly 
speaking, the line operators correspond to the generic objects on the
 irreducible components of
the moduli spaces of isoclasses of objects of $\mathcal{C}(\Gamma)$.} the BPS UV line operators. This identification 
is deeply related to the Kontsevitch-Soibelmann wall-crossing formula \cite{kontsevich2014wall}, see \cite{cecotti2010r,keller2011cluster}. The Grothendieck group $K_0(\mathcal{C}(\Gamma))$ 
then corresponds to the, additive as well as multiplicative,
 UV quantum numbers of the line operators. These quantum numbers, in particular the \textit{multiplicative} ones follow from the analysis by 't Hooft of the quantum phases of a $4d$ non-Abelian gauge theory being determined by the \emph{topology of the gauge group}
\cite{hooft1978phase,hooft1979property,hooft1980confinement,hooft1980topological}. 't Hooft arguments are briefly reviewed in \S.\,\ref{sec:UVcharge}: the UV line quantum numbers take value in a finitely generated Abelian group whose torsion part consists of two copies of the fundamental group of the gauge group while its free part describes flavor. 
 The fact that $K_0(\mathcal{C}(\Gamma))$ is automatically equal to the correct UV group, as predicted  by 't Hooft (detecting the precise topology of the gauge group!), yields convincing evidence for the proposed identification, see \S.\,\ref{kkkas12}.  $\mathcal{C}(\Gamma)$ is a 2-CY category, and hence its Euler form induces
a \emph{symmetric} form on the additive UV charges, which roughly speaking has the form
$$
 K_0(\mathcal{C}(\Gamma))\big/K_0(\mathcal{C}(\Gamma))_\text{torsion}\; \bigotimes\, 
 K_0(\mathcal{C}(\Gamma))\big/K_0(\mathcal{C}(\Gamma))_\text{torsion}\to
\Z,
$$
but whose precise definition is slightly more involved\footnote{\ The subtleties in the definition are immaterial when the QFT is UV superconformal (as contrasted to  asymptotically-free) and all chiral operators have integral dimensions.} (see \S.\,\ref{rrem}). We call this pairing the \emph{Tits form} of $\mathcal{C}(\Gamma)$. Its physical meaning is simple: while in the IR the masses break (generically) the flavor group to its maximal torus $U(1)^f$, in the deep UV the masses become irrelevant and the flavor group gets  enhanced to its maximal non-Abelian form $F$. Then the UV category should see the full $F$ and not just its Cartan torus. The datum of the group $F$ may be given as its weight lattice together with its Tits form; the cluster Tits form is equal to the Tits form of the non-Abelian flavor group $F$, and we may read $F$ directly from the cluster category. In fact the cluster category also detects the global topology of the flavor group, distinguishing (say) $SO(N)$ and $\mathrm{Spin}(N)$ flavor groups.\footnote{\ In facts, the cluster Grothendieck group $K_0(\mathcal{C}(\Gamma))$ should contain even more detailed informations on the flavor. For instance, in $SU(2)$ gauge theory with $N_f$ flavors the states of even magnetic charge are in tensor representations of the flavor $SO(2N_f)$ while states of odd magnetic charge are in spinorial representation of $\mathrm{Spin}(8)$; $K_0(\mathcal{C}(\Gamma))$ should know the correlation between the parity of the magnetic charge and $SO(2N_f)$ vs.\! $\mathrm{Spin}(2N_f)$ flavor symmetries (and it does). } 
For objects of $\mathcal{C}(\Gamma)$ there is also a weaker notion of `charge', taking value in the lattice $\Lambda$ of electric/magnetic/flavor charges, namely the \emph{index}, which is the quantity referred to as `charge' in many treatments.
Since $\mathcal{C}(\Gamma)$ yields an UV description of the theory, there must exist relations between its mathematical properties and the physical conditions assuring UV completeness of the associated QFT. We shall point of some of them in \S.\,\ref{lllaz231};
\item 
$\mathfrak{Per}\,\Gamma$ is the \emph{perfect derived category} of $\Gamma$. From eqn.\eqref{exactse} we see that, morally speaking,
the triangle category $\mathfrak{Per}\,\Gamma$ describes all 
possible BPS IR object generated by the insertion of UV line operators, dressed (screened)
by particles, in all possible vacua. This rough idea is basically correct. Perhaps the most convincing argument comes from consideration of  class $\mathcal{S}$ theories,
where we have a geometric construction of the perfect category $\mathfrak{Per}\,\Gamma$ \cite{qiu2016decorated}
as well as a detailed  understanding of the BPS physics \cite{gaiotto2013wall,gaiotto2013framed}.
In agreement with this identification, the Grothendieck group $K_0(\mathfrak{Per}\,\Gamma)$ is isomorphic to the IR group $\Lambda$.
$\mathfrak{Per}\,\Gamma$ is not CY, instead the Euler form defines a \emph{perfect} pairing
$$
K_0(D^b\Gamma)\bigotimes K_0(\mathfrak{Per}\,\Gamma)\to \Z;
$$
\item the exact functor $\mathsf{r}$ in eqn.\eqref{exactse} may be seen as the homological (inverse) RG flow.
\end{itemize} 
\medskip

\paragraph{Dualities.} The (self)-dualities of an $\mathcal{N}=2$ theory should relate BPS objects to BPS objects of the same kind, and
hence should be (triangle) auto-equivalences of the above categories which are consistent with the functors relating them (e.g.\! $\mathsf{s}$, $\mathsf{r}$ in eqn.\eqref{exactse}). We may
describe the physical situation from different viewpoints.
In the IR picture one would have the putative `duality' group
$\mathrm{Aut}\,D^b\Gamma$; however a subgroup acts trivially on all observables \cite{caorsi2016homological},
and the physical IR `duality' group is\footnote{\ For the precise definition of $\mathrm{Auteq}\,D^b\Gamma$, see \S.\,5. $\mathrm{Aut}(Q)$ is the group of automorphisms of the quiver $Q$ modulo the subgroup which fixes all nodes.}
\begin{equation}\label{kkkkzz10c}
\mathcal{S}_\text{IR}\equiv \mathrm{Aut}\,D^b\Gamma\Big/\big\{\text{physically trivial autoequivalences}\big\}
= \mathrm{Auteq}\,D^b\Gamma \rtimes \mathrm{Aut}(Q).  
\end{equation}
In the UV (that is, at the operator level) the natural candidate `duality' group is
$$\mathcal{S}_\text{UV}\equiv \mathrm{Aut}\,\mathcal{C}(\Gamma)\Big/\big\{\text{physically trivial}\big\}$$

From the explicit description of $\mathrm{Aut}\, D^b\Gamma$ (see \S.5) we learn that $\mathcal{S}_\text{IR}$ extends to a group
of autoequivalences of $\mathfrak{Per}\,\Gamma$ which preserve
$D^b\Gamma$ (by definition). Hence the exact functor $\mathsf{r}\colon \mathfrak{Per}\,\Gamma\to \mathcal{C}(\Gamma)$ 
in eqn.\eqref{exactse} induces a group homomorphism
$$
\mathcal{S}_\text{IR} \xrightarrow{\;r\;} \mathcal{S}_\text{UV},
$$
whose image is
\begin{equation}\label{sduality}
 \mathbb{S}= r\Big(\mathrm{Auteq}\,D^b\Gamma\Big)\rtimes \mathrm{Aut}\,Q.
\end{equation}
$\mathbb{S}$ is a group of auto-equivalences whose action is defined at the operator level, that is, independently of a choice of vacuum.
They are equivalences of the full web of BPS categories in eqn.\eqref{exactse}. Thus $\mathbb{S}$ is the natural candidate for the role of
the (extended) \emph{$S$-duality group} of our $\mathcal{N}=2$ model. Indeed, in the examples where we know the  $S$-duality group from more 
conventional considerations,
it coincides with our categorical group $\mathbb{S}$. In this survey we take equation \eqref{sduality} as the definition of the 
$S$-duality group.

Clearly, the essential part of $\mathbb{S}$ is the group $r(\mathrm{Auteq}\,D^b\Gamma)$. It turns out that
precisely this group is an object of central interest in the mathematical literature which provides an 
explicit combinatorial
description of it \cite{fock2009cluster}. This combinatorial description is the basis of an algorithm for computer search of $S$-dualities, see \S.\,\ref{duacccom}. If our $\mathcal{N}=2$ theory is not too complicated (that is, the ranks of the gauge and flavor group are not too big) the algorithm may be effectively implemented on a laptop, see \S.\,\ref{duacccom} for explicit examples. 

For class $\mathcal{S}$ theories, the above combinatorial description of $S$-duality has a nice geometric
intepretation as the (tagged) mapping class group of the Gaiotto surface, in agreement with the predictions of \cite{gaiotto2012n}(see also \cite{assem2012cluster}), see \S.\,5.2. More generally,
for class $\mathcal{S}$ theories all categorical constructions have a simple geometric realization which makes manifest their physical meaning. 

The IR group $\mathcal{S}_\text{IR}$ may be understood in terms of duality walls, see \S.\ref{sec:dualitywalls}. 

\paragraph{Cluster characters and vev of line operators.} The datum of a Coulomb vacuum $u$ defines a map
$$
\langle\, -\, \rangle_u \colon \mathrm{GenOb}(\mathcal{C}\big(\Gamma)\big) \to \C, 
$$
given by taking the \textsc{vev} in the vacuum $u$ of the UV line operator associated to a given \textit{generic} object of the cluster category
$\mathcal{C}(\Gamma)$.
Physically, the renormalization group implies that the
 map $\langle\,-\,\rangle_u$ factors through the (Laurent) ring $\Z[\mathcal{L}]$ of line operators 
in the effective (Abelian) IR theory. The associated map
$$
\mathrm{GenOb}(\mathcal{C}\big(\Gamma)\big) \to \Z[\mathcal{L}]
$$ 
 is called a \emph{cluster character} and is well understood in the mathematical literature. Thus the theory of cluster character solves 
(in principle) the problem of computing the \textsc{vev} of arbitrary BPS line operators (see \S.\ref{sec:linop}).

\paragraph{Organization of the paper.} The rest of this paper is organized as follows. In section \ref{sec:math} we review the mathematics of derived DG categories,
cluster categories, and related topics in order to provide the reader with a
language which allows to unify and generalize several previous analysis of the BPS sector of a $\mathcal{N}=2$ $4d$ theory.
In section \ref{sec:physprel} we recall some general physical properties that our categories should enjoy in order to be valid descriptions of the various BPS objects. Here we stress the various notions of charge, with particular reference to the 't Hooft charges in the UV description.
Section \ref{sec:physinterpr} is the core of the paper, where we present our physical interpretation of the categories described mathematically in section \ref{sec:math} and show that they satisfy all physical requirements listed in section \ref{sec:physprel}. In section \ref{sec:sdualmap} we discuss $S$-duality from the point of view of triangle categories, present first examples and describe the relation to $3d$ mirror symmetry. In section \ref{duacccom} we introduce our combinatorial algorithm to find $S$-dualities and give a number of examples. In section \ref{sec:surfaces} we consider class $\mathcal{S}[A_1]$ theories and describe all the relevant triangle categories in geometric terms \emph{\`a la} Gaiotto. In particular, this shows that our categorical definition of $S$-duality is indeed equivalent to more conventional physical definitions.
In section \ref{sec:linop} we discuss \textsc{vev}'s of line operators from the point of view of cluster characters. Explicit computer codes are presented in the appendices.

\section{Mathematical background}
\label{sec:math}
In this section we recall the basic definitions of DG categories \cite{keller2006differential}, 
cluster categories \cite{amiot2011generalized}, stability conditions for Abelian 
and triangulated categories \cite{bridgeland2007stability} and show some concrete examples. 
We then specialize these definitions to the Ginzburg algebra \cite{ginzburg2006calabi} $\Gamma$ associated
 to a BPS quiver with (super)potential $(Q,W)$ \cite{cecotti2012quiver}. 
 
 Some readers may prefer to skip this section  and refer back to it when looking for definitions and/or details on some mathematical tool used in the main body of this survey.
 
\subsection{Differential graded categories}
\label{sec:basicdef}
The main reference for this section is \cite{keller2006differential}.
Let $k$ be a commutative ring,\footnote{\ In all our physical applications $k$ will be the (algebraically closed) field of complex numbers $\mathbb{C}$.} for example a field or the ring of
integers $\Z$. We will write $\otimes$ for the tensor product over $k$. 
\begin{definition}A \textit{$k$-algebra}
is a $k$-module $A$ endowed with a $k$-linear associative multiplication $A \otimes_k A \to A$
admitting a two-sided unit $1 \in A$. 
\end{definition}
For example, a $\Z$-algebra is just a (possibly non-commutative)
ring. A $k$-category $\A$ is a ``$k$-algebra with several objects''. Thus, it is the datum of a class of objects obj($\A$), of a $k$-module
$\A(X, Y )$ for all objects $X, Y$ of $\A$, and of $k$-linear associative composition maps
$$\A(Y, Z) \otimes \A(X, Y ) \to \A(X, Z),\qquad (f, g) \mapsto fg$$
admitting units $1_X \in \A(X, X)$. For example, we can interpret $k$-algebras as $k$-categories
with only one object. The category $\mathsf{mod}\,A$ of finitely generated right $A$-modules over a $k$-algebra $A$ is
an example of a $k$-category with many objects. It is also an example of a $k$-linear
category (i.e. a $k$-category which admits all finite direct sums).
\begin{definition}A \textit{graded $k$-module} is a $k$-module $V$ together with a decomposition indexed by
the positive and the negative integers:
$$V =\bigoplus_{p\in \Z}V^p.$$
The shifted module $V [1]$ is defined by $V [1]^p = V^{ p+1}, p \in \Z$. A morphism $f : V \to
V $ of graded $k$-modules of degree $n$ is a $k$-linear morphism such that $f (V^p) \subset V^{p+n}$
for all $p \in \Z$. 
\end{definition}
\begin{definition}
The tensor product $V \otimes W$ of two graded $k$-modules $V$ and $W$ is the
graded $k$-module with components
$$(V \otimes W )^n =\bigoplus_{p+q=n}V^p \otimes W^q , \ n \in \Z.$$
The tensor product $f \otimes g$ of two maps $f : V \to V $ and $g : W \to W$ of graded
$k$-modules is defined using the Koszul sign rule: we have
$$(f \otimes g)(v \otimes w) = (-1)^{pq}\,f (v) \otimes g(w)$$
if $g$ is of degree $p$ and $v$ belongs to $V^q$.
\end{definition}
\begin{definition}
A \textit{graded $k$-algebra} is a graded $k$-module $A$
endowed with a multiplication morphism $A \otimes A \to A$ which is graded of degree 0,
associative and admits a unit $1 \in A^0$.
\end{definition}
\noindent An ``ordinary'' $k$-algebra may be identified with a graded $k$-algebra concentrated in degree 0. 
\begin{definition}
A \textit{differential graded} (=DG) \textit{$k$-module} is a $\Z$-graded $k$-module $V$ endowed with a
differential $d_V$, i.e.\! a map $d_V\colon V \to V$ of degree 1 such that $d^2_V = 0$. Equivalently, $V$
is a complex of $k$-modules. The shifted DG module $V [1]$ is the shifted graded module
endowed with the differential $-d_V$ . 
\end{definition}
The tensor product of two DG $k$-modules is the
graded module $V \otimes W$ endowed with the differential $d_V \otimes 1_W + 1_V \otimes d_W $.
\begin{definition}
A \textit{differential graded $k$-algebra $A$} is a DG $k$-module endowed with a multiplication morphism $A \otimes A \to A$ graded 
of degree 0 and associative. Moreover, the differential satisfies the graded Leibnitz rule:
$$d(ab)=(da) b+(-1)^{\deg(a)}\,a(db), \quad \forall\, a,b\in A \ \text{and } a \text{ homogeneous}.$$
\end{definition}
The cohomology of a DG algebra is defined as $H^*(A):=\mathrm{ker}\, d/\mathrm{im}\,d$.
Let $\mathsf{mod}\,A$ denote the category of finitely generated  DG modules over the DG algebra $A$.
\begin{definition}
The \textit{derived category} $D(A):=D(\mathsf{mod}\,A)$ is the localization
of the category $\mathsf{mod}\,A$ with respect to the class of quasi-isomorphisms. 
\end{definition}
Thus, the objects of $D(A)$ are the DG modules and its morphisms are obtained from morphisms of DG
modules by formally inverting all quasi-isomorphisms. The bounded derived category of $\mathsf{mod}\,A$, denoted $D^bA$, 
is the triangulated subcategory of $D(A)$ whose objects are quasi-isomorphic to objects with bounded cohomology.
\begin{definition}
The \textit{perfect derived category} of a DG algebra $A$, $\mathfrak{Per}\,A$, is the smallest full triangulated subcategory of $D(A)$ containing
$A$ which is stable under taking shifts, extensions and direct summands.
\end{definition}
\subsection{Quivers and mutations}
\label{sec:quiv}
In this section we follow \cite{keller2011derived}. Let $k$ be an algebraically closed field. 
\begin{definition}
A (finite) \emph{quiver} $Q$ is a (finite) oriented graph (possibly with loops and 2-cycles). We denote its set of vertices by $Q_0$ and its
set of arrows by $Q_1$. For an arrow $a$ of $Q$, let $s(a)$ denote its source node and $t(a)$ denote
its target node. The lazy path corresponding to a vertex $i$ will be denoted by $e_i$.
\end{definition}
\begin{definition}
The \emph{path algebra} $k\hat Q$ is the associative unital algebra whose 
elements are finite compositions of arrows of $Q$, where the composition of $a,b \in Q_1$ is denoted $ab$ and 
it is nonzero iff $s(b)=t(a).$ The complete path algebra $kQ$ is the completion of the path algebra with 
respect to the ideal $I$ generated by the arrows of $Q$. 
\end{definition}
Let $I$ be the ideal of $kQ$ generated by
the arrows of $Q$. A potential $W$ on $Q$ is an element of the closure of the space
generated by all non trivial cyclic paths of $Q$. We say two potentials are \emph{cyclically
equivalent} if their difference is in the closure of the space generated by all differences
$a_1 ... a_s - a_2 ... a_sa_1$, where $a_1 ... a_s$ is a cycle.
\begin{definition}
Let $u,p$ and $v$ be nontrivial paths of $Q$ such that $c=upv$ is a nontrivial cycle. For the path $p$ of $Q$, we define $$\partial_p : kQ \to kQ$$
as the unique continuous linear
map which takes a cycle $c$ to the sum $\sum_{
c=upv} vu$ taken over all decompositions of the
cycle $c$ (where $u$ and $v$ are possibly lazy paths).
\end{definition}
Obviously two cyclically equivalent
potentials have the same image under $\partial_p$. If $p = a$ is an arrow of $Q$, we call $\partial_a$ the
cyclic derivative with respect to $a$.
Let $W$ be a potential on $Q$ such that $W$ is in $I^2$ and no two cyclically equivalent
cyclic paths appear in the decomposition of $W$. Then the pair $(Q, W)$
is called a \textit{quiver with potential.} 
\begin{definition}
Two quivers with potential $(Q, W)$ and $(Q', W')$
are \textit{right-equivalent} if $Q$ and $Q'$ have the same set of vertices and there exists an
algebra isomorphism $\phi : kQ \to kQ'$ whose restriction on vertices is the identity
map and $\phi(W) $ and $W'$
are cyclically equivalent. Such an isomorphism $\phi$ is called a
right-equivalence.
\end{definition}
\begin{definition}
The \emph{Jacobian algebra} of a quiver with potential $(Q, W)$, denoted by $J(Q, W)$, is
the quotient of the complete path algebra $kQ$ by the closure of the ideal generated by
$\partial_a W$, where $a$ runs over all arrows of Q:
$$J(Q,W):=kQ/\vev{\partial_aW}.$$
We say that the quiver with potential $(Q,W)$ is \emph{Jacobi-finite} if the Jacobian algebra 
$J(Q,W)$ is finite-dimensional over $k$.
\end{definition}
It is clear that two right-equivalent quivers
with potential have isomorphic Jacobian algebras. A quiver with potential is called
\textit{trivial} if its potential is a linear combination of cycles of length 2 and its Jacobian
algebra is the product of copies of the base field $k$.

\subsubsection{Quiver mutations}
Let $(Q, W)$ be a quiver with potential. Let $i\in Q_0$ a vertex. Assume the following conditions:
\begin{itemize}
\item the quiver $Q $ has no loops;
\item the quiver $Q$ does not have 2-cycles at $i$;
\end{itemize}
We define a new quiver with potential $\tilde \mu_i(Q, W) = (Q', W')$ as
follows. The new quiver $Q'$ is obtained from $Q$ by
\begin{enumerate}
\item  For each arrow $\beta$ with target $i$ and each arrow $\alpha$ with source $i$, add a new
arrow $[\alpha \beta]$ from the source of $\beta$ to the target of $\alpha$.
\item Replace each arrow $\alpha$ with source or target $i$ with an arrow $\alpha^*$ in the opposite
direction.
\end{enumerate}
If we represent the quiver with its exchange matrix $B_{ij}$, i.e.
 the matrix such that 
 \begin{equation}\label{exchanger}
 B_{ij}=\#\{\text{ arrows from } i \text{ to } j\}-\#\{\text{ arrows from } j \text{ to } i\}
 \end{equation} 
then the transformation that $B_{ij}$ undergoes is
$$
B'_{ij}=\begin{cases}
- B_{ij}, & i=k \ \text{or }j=k\\
B_{ij}+\max[-B_{ik},0]\,B_{kj}+B_{ik}\max[B_{kj},0] & \text{otherwise.}
\end{cases}
$$
The new potential $W'$ is the sum of two potentials $W'_1$ and $W'_2$. The potential $W'_1$
is obtained from $W$ by replacing each composition $\alpha \beta$ by $[\alpha \beta]$, where $\beta$ is an arrow
with target $i$. The potential $W'_2$ is given by \cite{derksen2008quivers}
$$
W'_2 =\sum_{\alpha,\beta}[\alpha \beta]\beta^*\alpha^*,
$$
the sum ranging over all pairs of arrows $\alpha$ and $\beta$ such that $\beta$ ends at $i$ and $\alpha$
starts at $i$.
\begin{definition}
Let $I$ be the ideal in $kQ$ generated by all arrows. Then, a quiver
with potential is called \emph{reduced} if $\partial_aW$ is contained in $I^2$ for all arrows $a$ of $Q$.
\end{definition}
\noindent One shows that all quivers with potential $(Q,W)$ are right-equivalent to the direct sum of a 
reduced quiver with potential and a trivial one.\footnote{\ In terms of the corresponding SQM system, the process
of replacing the pair $(Q,W)$ by its reduced part $(Q_\text{red.},W_\text{red.})$ corresponds to integrate away
the massive Higgs bifundamentals.}

We can now give the definition of the mutated quiver: we define $\mu_i(Q, W)$ as the reduced part 
of $\tilde \mu_i(Q, W),$ and call $\mu_i$ the mutation at the vertex $i$. An example will clarify all these concepts.
\begin{example}[$A_3$ quiver]
Consider the quiver $A_3$ given by $Q:\ \bullet_1 \overset{\alpha}{\leftarrow} \bullet_2 \overset{\beta}{\to} \bullet_3$ with $W=0$. 
Let us consider the quiver $ \mu_1(Q): \ \bullet_1 \overset{\alpha^*}{\to} \bullet_2 \overset{\beta}{\to} \bullet_3$ with $W=0$. 
Now apply the mutation at vertex $2$: we get
\begin{displaymath}
    \xymatrix{  &  \bullet_2  \ar[dl]_{\alpha} &  \\
            \bullet_1  \ar[rr]^{[\alpha\beta]}&   & \bullet_3 \ar[ul]_{\beta^*}  \\
 & \mu_2(\mu_1(Q))& }
\end{displaymath}
with potential $W=\alpha\beta^*[\alpha\beta]$.
\end{example}
We conclude this subsection with the following 
\begin{theorem}[\cite{keller2011derived}]\hskip-12pt
\begin{enumerate}
\item The right-equivalence class of $\tilde \mu_i(Q, W)$ is
determined by the right-equivalence class of $(Q, W)$.
\item The quiver with potential $\tilde \mu^2_i(Q, W)$ is right-equivalent to the direct sum
of $(Q, W)$ with a trivial quiver with potential.
\item The correspondence $\mu_i$ acts as an involution on the right equivalence
classes of reduced quivers with potential.
\end{enumerate}
\end{theorem}

\subsection{Cluster algebras}
We follow \cite{reiten2010tilting}. Let $Q$ be a 2-acyclic quiver
with vertices $1, 2, ..., n$, and let $F = \Q(x_1, ... , x_n)$ be the function field in $n$ indeterminates
over $\Q$. Consider the pair $(\vec x, Q)$, where $\vec x = \{x_1,... , x_n\}$. The cluster
algebra $C(\vec x, Q)$ will be defined to be a subring of $F$.

The pair $(\vec x, Q)$ consisting of a transcendence basis $\vec x$ for $F$ over the rational numbers
$\Q$, together with a quiver with $n$ vertices, is called a \emph{seed}. For $i = 1, ... , n$ we
define a mutation $\mu_i$ taking the seed $(\vec x, Q)$ to a new seed $(\vec x', Q')$, where $Q' = \mu_i(Q)$
as discussed in \ref{sec:quiv}, and $\vec x'$ is obtained from $\vec x$ by replacing $x_i$ by a new element $x'_i$ in
$F$. Here $x'_i$ is defined by $$x_ix'_i = m_1 + m_2,$$
where $m_1$ is a monomial in the variables
$x_1, ... , x_n$, where the power of $x_j$
is the number of arrows from $j$ to $i$ in $Q$, and $m_2$ is the
monomial where the power of $x_j$
is the number of arrows from $i$ to $j$. (If there is no
arrow from $j$ to $i$, then $m_1 = 1$, and if there is no arrow from $i$ to $j$, then $m_2 = 1$.)
Note that while in the new seed the quiver $Q'$ only depends on the quiver $Q$, then
$x'$ depends on both $x$ and $Q$. We have $$\mu^2_i(\vec x, Q) = (\vec x, Q).$$
The procedure to get the full cluster algebra is iterative. We perform this mutation operation for all $i = 1, ... , n,$ then we perform it on the new seeds and so on. Either we get new seeds or we get back one of the seeds already computed. The $n$-element subsets
$\vec x, \vec x', \vec x'', ...$ occurring are by definition the \emph{clusters}, the elements in the clusters are
the \emph{cluster variables}, and the \emph{seeds} are all pairs $(\vec x', Q')$ occurring in the iterative procedure. The corresponding
cluster algebra $C(\vec x, Q)$, which as an algebra only depends on $Q$, is the subring of $F$
generated by the cluster variables.
\begin{example}\label{kkkazm987} Let $Q$ be the quiver $1 \to 2 \to 3$ and $\vec x = \{x_1, x_2, x_3\}$, where
$x_1, x_2, x_3$ are indeterminates, and $F = \Q(x_1, x_2, x_3)$. We have $\mu_1(\vec x, Q) = (x', Q')$,
where $Q' = \mu_1(Q)$ is the quiver $1 \leftarrow 2 \to 3$ and $\vec x' = \{x'_1, x_2, x_3\}$, where $x_1x'_1 = 1+x_2,$
so that $x'_1 =\frac{1+x_2}{x_1}.$ And so on.
The clusters are: $$\{x_1, x_2, x_3\},\{\frac{1+x_2}{x_1}, x_2, x_3\},\{ x_1, \frac{x_1+x_3}{x_2}, x_3\},\{x_1, x_2,\frac{1+x_2}{x_3}\},$$ 
$$\{\frac{1+x_2}{x_1},\frac{x_1+(1+x_2)x_3}{x_1x_2}, x_3\},\{\frac{1+x_2}{x_1}, x_2,\frac{1+x_2}{x_3}\},\{\frac{x_1+(1+x_2)x_3}{x_1x_2},\frac{x_1+x_3}{x_2}, x_3\},$$
$$ \{x_1,\frac{x_1+x_3}{x_2}, \frac{(1+x_2)x_1+x_3}{x_2x_3}\}, \{x_1,\frac{(1+x_2)x_1+x_3}{x_2x_3}, \frac{1+x_2}{x_3}\},$$
$$ \{\frac{1+x_2}{x_1}, \frac{x_1+(1+x_2)x_3}{x_1x_2}, \frac{(1+x_2)x_1+(1+x_2)x_3}{x_1x_2x_3}\},\{\frac{1+x_2}{x_1}, \frac{(1+x_2)x_1+(1+x_2)x_3}{x_1x_2x_3},\frac{1+x_2}{x_3}\},$$ $$\{\frac{x_1+(1+x_2)x_3}{x_1x_2}, \frac{x_1+x_3}{x_2}, \frac{(1+x_2)x_1+(1+x_2)x_3}{x_1x_2x_3}\},$$ $$ \{\frac{(1+x_2)x_1+(1+x_2)x_3}{x_1x_2x_3},\frac{x_1+x_3}{x_2},\frac{(1+x_2)x_1+x_3}{x_2x_3}\}, $$
$$ \{\frac{(1+x_2)x_1+(1+x_2)x_3}{x_1x_2x_3},\frac{(1+x_2)x_1+x_3}{x_2x_3}, \frac{1+x_2}{x_3}\},$$ 
and the cluster variables are: $$x_1,x_2, x_3,\frac{1+x_2}{x_1},\frac{x_1+x_3}{x_2},\frac{1+x_2}{x_3},\frac{x_1+(1+x_2)x_3}{x_1x_2},\frac{(1+x_2)x_1+x_3}{x_2x_3},\frac{(1+x_2)x_1+(1+x_2)x_3}{x_1x_2x_3}.$$

\subsubsection{The cluster exchange graph (CEG)}\label{CEG}
If $Q'$
is a quiver mutation equivalent to $Q$, then the cluster algebras $C(Q')$ and
$C(Q)$ are isomorphic. The $n$-regular connected graph whose vertices are the seeds of $C(\vec x,Q)$ (up to simultaneous renumbering of rows, columns and variables) and whose edges connect the seeds related by a single mutation is called \emph{cluster exchange graph} (=CEG). The CEG for \textbf{Example \ref{kkkazm987}} is represented in figure \ref{exceg}.
\begin{figure}
\centering
\includegraphics[width=1.1\textwidth]{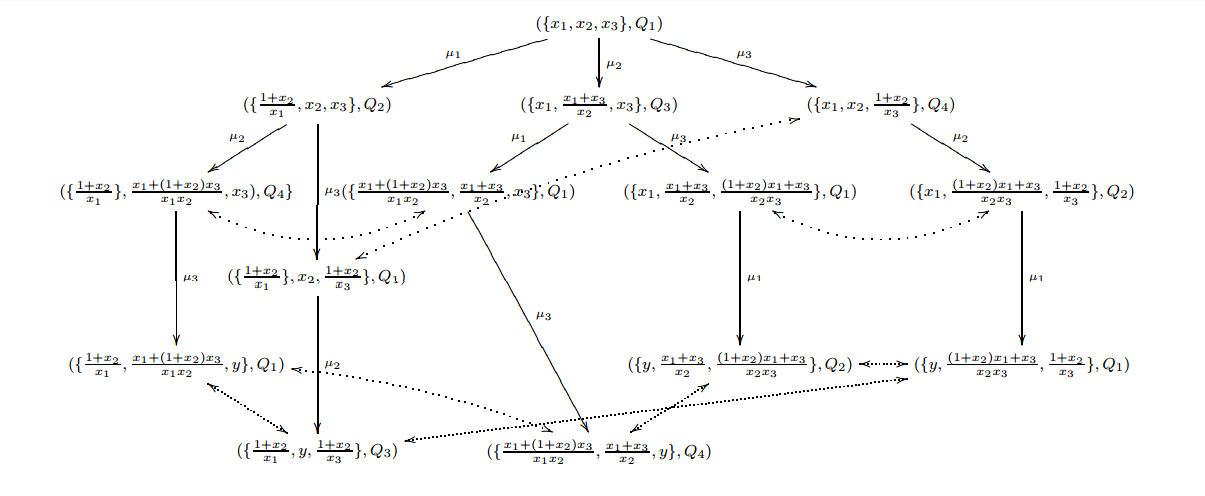}
\caption{\label{exceg}This figure represent the CEG of the $A_3$ cluster algebra. The dotted arrows are identifications up to permutations of the variables. The plain arrows represent mutations.}
\end{figure}
\end{example}

\subsection{Ginzburg DG algebras}
Given a quiver $Q$ with potential $W$, we can associate to it the Jacobian algebra $J(Q,W):=kQ/\vev{\partial W}$, where $kQ$ 
is the quiver path algebra (see section \ref{sec:quiv}). It is also possible to extend the path algebra $kQ$ to a DG algebra: the Ginzburg algebra.
\begin{definition} {(Ginzburg \cite{ginzburg2006calabi}).} Let $(Q, W)$ be a quiver with potential. Let $\hat Q$ be the graded quiver
with the same set of vertices as $Q$ and whose arrows are:
\begin{itemize}
\item the arrows of $Q$ (of degree 0);
\item an arrow $a^*: j \to i$ of degree $-1$ for each arrow $a : i \to j$ of $Q$;
\item a loop $t_i: i \to i$ of degree $-2$ for each vertex $i \in Q_0$.
\end{itemize}
The completed Ginzburg DG algebra $\Gamma(Q, W)$ is the DG algebra whose underlying
graded algebra is the completion\footnote{\ The completion is taken with respect to the $I$-adic topology, 
where $I$ is the ideal of the path algebra generated by all arrows of the quiver.} of the graded
path algebra $k\hat Q$. The differential of $\Gamma(  Q, W)$ is the unique continuous linear
endomorphism homogeneous of degree 1 which satisfies the Leibniz rule (i.e.
$d(uv) = (du)v + (-1)^p\,udv$ for all homogeneous $u$ of degree $p$ and all $v$), and takes
the following values on the arrows of $\hat Q$:
\begin{align*}d(a) &= 0\\ d(a^*) &= \partial_aW, &&\forall a \in Q_1;\\
 d(t_i) &= e_i\!\left(\sum_{a\in Q_1}[a, a^*]\right)\!e_i, &&\forall i \in Q_0.
\end{align*}
We shall write $\Gamma(Q,W)$ simply as $\Gamma$, unless we wish to stress its dependence on $(Q,W)$.
\end{definition}
From the definition of $\Gamma$ and $d$, one sees that $H^0\Gamma \cong J(Q,W)$.
\medskip

To the DG algebra $\Gamma$ we  associate three important triangle categories which we are now going to define and analyze in detail.

\subsection{The bounded and perfect derived categories}
\label{sec:cat}
The DG-category $\mathsf{mod}\,\Gamma$ is the category whose objects are finitely 
generated graded $\Gamma$-modules and the morphisms spaces have the structure of DG modules 
(cfr.\! section \ref{sec:basicdef}).
 The derived category $D\Gamma:=D(\mathsf{mod}\,\Gamma)$ 
\cite{keller2007derived,keller2006differential} is the localization of $\mathsf{mod}\,\Gamma$ 
at quasi-isomorphisms (the cohomology structure is given by the differential $d$ of the Ginzburg algebra). 
Thus, the objects of $D\Gamma$ are DG modules. There are two fundamental subcategories associated to $D\Gamma$:
\begin{itemize}
\item The bounded derived category $ D^b \Gamma$: it is the full subcategory of $D\Gamma$ such that its objects 
are graded modules $M$ for which, given a certain $N>0$, $H^n (M)=0 $ for all $|n| > N$.  This category is 3-CY (see below).
\item The perfect derived category $\mathfrak{Per}\,\Gamma$, i.e.\! the smallest full triangulated subcategory 
of $D\Gamma$ which contains $\Gamma$ and is closed under extensions, shifts in degree and taking direct summands. 
\end{itemize} 
Both $\mathfrak{Per}\,\Gamma$ and $D^b \Gamma$ are triangulated subcategories of $D\Gamma$ and in particular,
 $\mathfrak{Per}\,\Gamma \supset D^b \Gamma $ as a full subcategory (as explained in \cite{keller2011cluster}). 
Furthermore,\footnote{\ See \cite{keller2011cluster} for more details.} the category $D^b\Gamma$ 
has finite-dimensional morphism spaces (even its graded
morphism spaces are of finite total dimension) and is 3-Calabi-Yau (3-CY), by which we
mean that we have bifunctorial isomorphisms\footnote{\ More generally, we say that a triangle category is $\ell$-CY (for $\ell\in\mathbb{N}$) iff we have the bifunctorial isomorphism
$D\mathrm{Hom}(X,Y)\cong \mathrm{Hom}(Y,X[\ell])$.}
\begin{equation}\label{3cy}
D \mathrm{Hom}(X, Y ) \cong \mathrm{Hom}(Y, X[3]) ,
\end{equation}
where $D$ is the duality functor $\mathrm{Hom}_k(-, k)$ and $[1]$ the shift functor. The simple
$J(Q, W)$-modules $S_i$ can be viewed as $\Gamma$-modules via the canonical morphism
$$\Gamma \to H^0(\Gamma).$$ 
\begin{example}[$A_2$ quiver]
Consider the $A_2$ quiver $1\to 2$. The following is a graded indecomposable $\Gamma$-module:
$$t_1^*\CircleArrowright  (k[-1]\oplus k[-3]) \underset{a^*}{\overset{a}{\rightleftharpoons}} k \CircleArrowleft t_2^*, $$
where $a=0,\ a^*:k\overset{1}{\mapsto}k[-1],\ t_2^*=0, $ and $t_1^*:k[-1]\overset{1}{\mapsto} k[-3]$. 
This object can be generated from $S_1[-1],\ S_1[-3]$ and $S_2$ by successive extensions. Moreover, the modules $S_i, i =1,2$ 
and their shifts are enough to generate\footnote{\ In the triangulated category $\mathcal T$, a set of objects $S_i \in \mathcal T$ 
is a generating set if all objects of $\mathcal T$ can be obtained from the generating set via an iterated cone construction.} all
 (homologically finite) graded modules.
\end{example}

\subsubsection{Seidel-Thomas twists and braid group actions}\label{braidactions} Simple $\Gamma$-modules $S_i$ become 3-spherical objects in $D^b\Gamma$ (hence also in $D\Gamma$), that is,
 $$
\mathrm{Hom}(S,S[j])\cong k (\delta_{j,0}+\delta_{j,3}).
$$
 They yield the
Seidel-Thomas \cite{seidel2001braid,segal2016all} twist functors $T_{S_i}$. These are autoequivalences of $D\Gamma$ such that 
each object $X$ fits into a triangle
\begin{equation}\label{kkz12axl}\mathrm{Hom}^\bullet_D(S_i, X) \otimes_k S_i \to X \to T_{S_i}(X) \to .\end{equation}
By construction, $T_{S_i}$ restricts to an autoequivalence of the subcategory $D^b\Gamma\subset D\Gamma$. From the explicit realization of $T_{S_i}$ as a cone
in $D\Gamma$, eqn.\eqref{kkz12axl}, it is also clear that
it restricts to an auto-equivalence of $\mathfrak{Per}\,\Gamma$.

As shown in \cite{seidel2001braid}, the twist functors give rise to a (weak) action on $D\Gamma$ of the braid group
associated with $Q$, i.e.\! the group with generators $\sigma_i, i \in Q_0$, and relations $$\sigma_i\sigma_j =
\sigma_j\sigma_i$$
if $i$ and $j$ are not linked by an arrow in $Q$ and
$$\sigma_i\sigma_j\sigma_i = \sigma_j\sigma_i\sigma_j$$
if there is exactly one arrow between $i$ and $j$ (no relation if there are two or more
arrows). \medskip

\begin{definition}
We write $\mathrm{Sph}(D^b\Gamma)\subset \mathrm{Aut}(D^b\Gamma)$ for the subgroup of autoequivalences
 generated by the Seidel-Thomas twists associated to all simple objects $S_i\in D^b\Gamma$.
\end{definition}
 
\subsubsection{The natural $\boldsymbol{t}$-structure and the canonical heart} 
The category $D\Gamma$ admits a natural $\boldsymbol{t}$-structure whose truncation functors are
those of the natural $\boldsymbol{t}$-structure on the category of complexes of vector spaces
(because $\Gamma$ is concentrated in degrees $\leq 0$). Thus, we have an induced natural
$\boldsymbol{t}$-structure on $D^b\Gamma$. Its heart $\A$ is canonically equivalent to the category
$\mathsf{nil}\,J(Q, W)$ of nilpotent modules\footnote{\ If $(Q,W)$ is Jacobi-finite (as in our applications), $\mathsf{nil}\,J(Q, W)\equiv \mathsf{mod}\, J(Q, W)$.} \cite{keller2011cluster}. In particular, the inclusion of $\A$ into $D^b\Gamma$ induces an isomorphism
of Grothendieck groups
$$K_0(\A) \cong K_0(D^b\Gamma)\cong \bigoplus_i \Z [S_i].$$

\paragraph{The skew-symmetric form.}
Notice that the lattice $K_0(D^b\Gamma)$ carries the canonical Euler form defined by
\begin{equation}
\label{eq:euler}
\vev{X,Y}=\sum_{i=0}^3(-1)^i\dim \mathrm{Hom}_{D(\Gamma)}(X,Y[i]).
\end{equation}
It is skew-symmetric thanks to the 3-Calabi-Yau property \eqref{3cy}. Indeed it follows from the
Calabi-Yau property and from the fact that $\mathrm{Ext}^i_\A(L, M) = \mathrm{Hom}_{D^b\Gamma}(L, M[i])$ for $i = 0$
and $i = 1$ (but not $i > 1$ in general) that for two objects $L$ and $M$ of $\A\subset D^b\Gamma$, we have
$$\vev{L, M} = \dim \mathrm{Hom}(L, M) - \dim \mathrm{Ext}^1(L, M) + \dim \mathrm{Ext}^1(M, L) - \dim \mathrm{Hom}(M, L).$$
Since the dimension of $\mathrm{Ext}^1(S_i, S_j )$ equals the number of arrows in $Q$ from $j$ to
$i$ (Gabriel theorem \cite{assem2006elements}), we obtain that the matrix of $\vev{-,-}$ in the basis of the simples of $\A$ has its 
$(i, j)$-coefficient equal to the number of arrows from $i$ to $j$ minus the number of arrows
from $j$ to $i$ in $Q$, that is,
 (cfr.\! eqn.\eqref{exchanger})
 \begin{equation}\label{uuuia}\langle S_i,S_j\rangle= B_{ij}.\end{equation}

\subsubsection{Mutations at category level}
The main reference for this subsection is \cite{keller2011cluster}. Let $k$ be a vertex of the quiver $Q$ not 
lying on a 2-cycle and let $(Q', W')$ be the
mutation of $(Q, W)$ at $k$. Let $\Gamma'$ be the Ginzburg algebra
associated with $(Q', W').$ Let $\A'$ be the canonical heart in $D^b\Gamma'$. There are
two canonical equivalences 
$$D\Gamma' \to D\Gamma$$
given by functors $\Phi^\pm$ related by
$$T_{S_k} \circ \Phi^- \to \Phi^+.$$
where, again, $T_{S_k}$ is the Seidel-Thomas twist generated by the spherical object $S_k$.
If we put $P_i = \Gamma e_i, i \in Q_0$, and similarly for $\Gamma'$, then both $\Phi^+$ and $\Phi^-$ send $P'_i$ to
$P_i$ for $i \neq k$; the images of $P'_k$ under the two functors fit into triangles
\begin{equation}
P_k \to \bigoplus_{k \to i} P_i \to \Phi^-(P'_k) \to
\label{eq:m-}
\end{equation}
and
\begin{equation}
\Phi^+(P'_k) \to \bigoplus_{j\to k} P_j \to P_k .
\label{eq:m+}
\end{equation}
The functors $\Phi^\pm$ send $\A'$ onto the hearts $\mu_k^\pm(\A)$ of two new $\boldsymbol{t}$-structures. These
can be described in terms of $\A$ and the subcategory\footnote{\ Here and below, given a (collection of) object(s) $\mathcal{O}$ of a linear category $\mathfrak{L}$, by $\mathsf{add}\,\mathcal{O}$ we mean the \textit{additive closure of $\mathcal{O}$ in $\mathfrak{L}$,} that is, the full subcategory over the direct summands of finite direct sums of copies of $\mathcal{O}$.} 
$\mathsf{add}\,S_k$ as follows:
Let $S^\perp_k$ be the right orthogonal subcategory of $S_k$ in $\A$\footnote{\ Its objects are those $M$'s
with $\mathrm{Hom}(S_k, M) = 0$. It is a full subcategory of $\A$.}. Then $\mu^+_k(\A)$ is formed by the objects $X$ of $D^b\Gamma$ such that
the object $H^0(X)$ belongs to $S^\perp_k$, the object $H^1(X)$ belongs to $\mathsf{add}\,S_k$ and $H^p(X)$
vanishes for all $p \neq  0, 1$. Similarly, the subcategory $\mu^-_k(\A)$ is formed by the objects
$X$ such that the object $H^0(X)$ belongs to the left orthogonal subcategory $^\perp S_k$, the
object $H^{-1}(X)$ belongs to $\mathsf{add}\,S_k$ and $H^p(X)$ vanishes for all $p \neq -1, 0$. The
subcategory $\mu^+_k(\A)$ is the right mutation of $\A$ and $\mu^-_k(\A)$ is its left mutation.
By construction, we have
$$T_{S_k}(\mu^-_k(\A)) = \mu^+_k(\A).$$
Since the categories $\A$ and $\mu^\pm(\A)$ are hearts of bounded, non degenerate $\boldsymbol{t}$-structures
on $D^b\Gamma$, their Grothendieck groups identify canonically with that of $D^b\Gamma$. They
are endowed with canonical basis given by the simples. Those of $\A$ identify with
the simples $S_i, i \in Q_0,$ of $\mathsf{nil}\,J(Q, W)$. The simples of $\mu^+_k(\A)$ are $S_k[-1]$, the
simples $S_i$ of $\A$ such that $\mathrm{Ext}^1(S_k, S_i)$ vanishes and the objects $T_{S_k}(S_i)$ where
$\mathrm{Ext}^1(S_k, S_i)$ is of dimension $\geq 1$. By applying $T^{-1}_{S_k}$ to these objects we obtain the
simples of $\mu^-_k(\A).$

We saw that $D^b\Gamma \subset\mathfrak{Per}\, \Gamma$ as a full subcategory \cite{keller2011derived}: 
what is then the meaning of the Verdier quotient \cite{neeman2014triangulated} of these two triangulated categories?

\subsection{The cluster category}
\label{sec:cluster}
The next result is the main step in the construction
of new 2-CY categories with cluster-tilting object which generalize the acyclic
cluster categories introduced by Buan-Marsh-Reineke-Todorov to categorify the cluster agebras of Fomin and Zelevinski.
\begin{theorem}[Thm 2.1 of \cite{amiot2009cluster})]\label{amiotthm} Let $A$ be a DG-algebra with the following properties:
\begin{enumerate}
\item $A$ is homologically smooth (i.e. $A\in \mathfrak{Per}(A\otimes A^{op})$),
\item $H^p(A) = 0$ for all $p \geq 1$,
\item $H^0(A)$ is finite dimensional as a $k$-vector space,
\item $A$ is bimodule 3-CY, i.e. 
\begin{equation}\label{uuuu}
\mathrm{Hom}_{D(A)}(X, Y ) \cong D\mathrm{Hom}_{D(A)}(Y, X[3]),
\end{equation}
for any $X \in D(A)$ and $Y \in D^b(A)$.
\end{enumerate}
Then the triangulated category 
$$\mathcal{C}(A) = \mathfrak{Per}\,A\big/D^b A$$
 is Hom-finite, 2-CY, i.e.
 $$
\mathrm{Hom}_{\mathcal{C}(A)}(X, Y ) \cong D\mathrm{Hom}_{\mathcal{C}(A)}(Y, X[2]),\qquad X,Y\in \mathcal{C}(A).
$$
 and the
object $A$ is a \emph{cluster-tilting object}\footnote{\ See \textbf{Definition \ref{def:clustertilt}}.} with 
\begin{equation}\label{whatend}
\mathrm{End}_{\mathcal{C}(A)}(A) \cong H^0(A).
\end{equation}
\end{theorem}


The category $\mathcal{C}(A)$ is called the \emph{generalized cluster category} and 
it reduces to the standard cluster category \cite{keller2005triangulated} in the acyclic case. 
It is triangulated since it is the Verdier quotient of triangulated 
categories.\footnote{\ The main references for these categorical facts are \cite{drinfeld2004dg,neeman2014triangulated}.
We recall the definition of Verdier quotient of triangle categories:\smallskip

\noindent\textbf{Lemma.}
\textit{Let $D$ be a triangulated category. Let $D' \subset D$ be a full triangulated subcategory. Let $S\subset \mathsf{Mor}(D)$
be the subset of morphisms such that there exists a distinguished triangle $(X,Y,Z,f,g,h)\in D$ with $Z$
isomorphic to an object of $D^\prime$. 
Then $S$ is a multiplicative system compatible with the triangulated structure on $D$.}\smallskip

\noindent\textbf{Definition.}
 Let $D$ be a triangulated category. Let $B$ be a full triangulated subcategory.
 We define the (Verdier) quotient category $D/B$ by the formula $D/B=S^{-1}D$, where $S$ is the multiplicative 
system of $D$ associated to $B$ via the previous lemma.}

\subsubsection{The case of the Ginzburg algebra of $(Q,W)$} In particular, we may specialize to the case where $A=\Gamma$, i.e. the Ginzburg algebra of a quiver with 
potential $(Q,W)$, and write the following sequence:
\begin{equation}\label{treexa}
0 \to D^b\Gamma \xrightarrow{\,\mathsf{s}\,} \mathfrak{Per}\, \Gamma \xrightarrow{\,\mathsf{r}\,} \mathcal{C}(\Gamma) \to 0
\end{equation}
the above theorem states that this sequence is exact and $\mathsf{r}(\Gamma)=T$, where $T$ is the \emph{canonical} cluster-tilting 
object\footnote{\ See \textbf{Definition \ref{def:clustertilt}}.} of $\mathcal{C}(\Gamma)$. The first map in eqn.\eqref{treexa} is the
 inclusion map: see \cite{keller2011derived} for details. 
\begin{remark}
Moreover, an object $M\in\mathfrak{Per}\,\Gamma$ belongs to the subcategory $D^b\Gamma$ if and only if
the space $\mathrm{Hom}_{\mathfrak{Per}\,\Gamma}(P, M)$ is finite-dimensional for each $P\in \mathfrak{Per}\,\Gamma$. 
In particular, this implies that there is a duality between the simple objects $S_i \in D^b\Gamma$ 
and the projective objects $\Gamma e_i \in\mathfrak{Per}\,\Gamma$
\begin{equation}\label{duadua}
\langle \Gamma e_i, S_j\rangle=\delta_{ij}.
\end{equation}
\end{remark}

\begin{theorem}[Keller \cite{keller2011deformed}] The completed Ginzburg DG algebra $\Gamma(  Q, W)$ is
homologically smooth and bimodule 3-Calabi-Yau.
\end{theorem}
We have already shown that $\Gamma(Q, W)$ is non zero only in negative degrees, and
that $H^0(\Gamma(  Q, W)) \cong  J(Q, W)$. Therefore by the theorem above we get the following
\begin{corollary}\label{corco}
Let $(Q, W)$ be a Jacobi-finite quiver with potential. Then the category
$$\mathcal{C}(\Gamma(Q,W)):= \mathfrak{Per}\,\Gamma(Q,W)\big/D^b(\Gamma(  Q, W))$$
is Hom-finite, 2-Calabi-Yau, and has a canonical cluster-tilting object\footnote{\ See \textbf{Definition \ref{def:clustertilt}}.} whose endomorphism
algebra is isomorphic to $J(Q, W)$.
\end{corollary}

We shall write $\mathcal{C}(\Gamma(Q,W))$ simply as $\mathcal{C}(\Gamma)$ leaving $(Q,W)$ implicit.

\subsubsection{The cluster category of a hereditary category}\label{herditarycase}

The above structure simplifies in the case of a cluster category arising from a hereditary (Abelian) category $\mathcal{H}$ (with a Serre functor and a tilting object) \cite{lenzing1}. Physically this happens for the following list of complete $\mathcal{N}=2$ QFTs \cite{cecotti2011classification}: \textit{i)} Argyres-Douglas of type $ADE$, \textit{ii)} asymptotically-free $SU(2)$ gauge theories coupled to fundamental quarks and/or Argyes-Douglas models of type $D$, and \text{iii)} SCFT $SU(2)$ theories with the same kind of matter. In terms of quiver mutations classes, they correspond (respectively) to $ADE$ Dynkin quivers of the finite, affine, and elliptic type\footnote{\ In the elliptic type we are restricted to the four types $D_4$, $E_6$, $E_7$ and $E_8$, corresponding to the four tubular weighted projective lines \cite{geigle,lenzing1,kussin2013triangle}.
Elliptic $D_4$ is $SU(2)$ with $N_f=4$ \cite{cecotti2011classification}.} \cite{cecotti2011classification}. In all these case we have
an hereditary (Abelian) category $\mathcal{H}$, with the Serre functor $S=\tau [1]$ where $\tau$ is the Auslander-Reiten translation. That is, in their derived category we have
$$
\mathrm{Hom}_{D^b(\mathcal{H})}(X,Y)\cong D\mathrm{Hom}_{D^b(\mathcal{H})}(Y,\tau X[1])$$
$\tau$ is an auto-equivalence of
$D^b(\mathcal{H})$. The cluster category can be shown to be equivalent to the orbit category
\cite{keller2005triangulated,barot}
\begin{equation}\label{mzxa}
\mathcal{C}(\mathcal{H})\cong D^b(\mathcal{H})/\langle \tau^{-1}[1]\rangle^\Z.
\end{equation} 
For future reference, we list
the relevant categories $\mathcal{H}$
(further details may be found in \cite{cecotti2013categorical}):
\begin{itemize}
\item For Argyres-Douglas of type $ADE$, we have $\mathcal{H}\cong \mathsf{mod}\,k \vec{\mathfrak{g}}$, where $\vec{\mathfrak{g}}$ is a quiver obtained by choosing an orientation to the Dynkin graph of type $\mathfrak{g}\in ADE$ (all orientations being derived-equivalent). $\tau$ satisfies the equation (for more refined results see \cite{miyaki,caorsi2016homological})
\begin{equation}\label{gabeq}
\tau^h=[-2],
\end{equation}
where $h$ is the Coxeter number of the associated Lie algebra $\mathfrak{g}$;
\item for $SU(2)$ gauge theories coupled to Argyres-Douglas systems of types\footnote{\ In our conventions, $p_i=1$ means the empty matter system, while $p_i=2$ is a free quark doublet.} $D_{p_1},\cdots,D_{p_s}$,
we have $\mathcal{H}=\mathsf{coh}\,\mathbb{X}(p_1,\dots,p_s)$, the coherent sheaves over a weighted projective line of weights $(p_1,\dots, p_s)$ \cite{geigle,lenzing1}\footnote{\ For a review of the category of coherent sheaves on weighted projective lines and corresponding cluster categories  from a physicist prospective, see \cite{cecotti2015higher}.}. $\tau$ acts by multiplication by the canonical sheaf $\omega$, and hence is periodic iff $\deg \omega=0$; in general, $\deg\omega$ is minus the Euler characteristic of $\mathbb{X}(p_1,\dots,p_s)$, $\chi=2-\sum_i(1-1/p_i)$.
However, $\tau$ is always periodic of period $\text{lcm}(p_i)$ when restricted to the zero rank sheaves (`skyskrapers' sheaves).
\end{itemize}


\subsection{Mutation invariance}
We have already stated that mutations correspond to Seiberg-like dualities. Therefore, our categorical construction makes sense only if it is invariant by mutations: indeed, we do not want the categories representing the physics to change when we change the mathematical description of the same dynamics.

The following two results give a connection between the DG categories we just analyzed and quivers with potentials linked by mutations.
\begin{theorem}\label{ttthma} Let $(Q, W)$ be a quiver with potential without loops and $i \in Q_0$ not on a 2-cycle in
$Q$. Denote by $\Gamma := \Gamma(  Q, W)$ and $\Gamma':= \Gamma(  \mu_i(Q, W))$ the completed Ginzburg DG algebras.
\begin{enumerate}
\item {\rm \cite{keller2011deformed} }There are triangle equivalences
\begin{displaymath}
    \xymatrix{\mathfrak{Per}\, \Gamma \ar[r]^\sim & \mathfrak{Per}\,\Gamma^\prime \\
                     D^b\Gamma \ar[r]^{\sim}\ar@{^{(}->}[u]& D^b\Gamma^\prime \ar@{^{(}->}[u] }
\end{displaymath}
Hence we have a triangle equivalence $\mathcal{C}(\Gamma) \cong \mathcal{C}(\Gamma')$.
\item {\rm\cite{plamondon2011cluster} } We have a diagram
\begin{displaymath}
    \xymatrix{\mathfrak{Per}\, \Gamma \ar[rrr]^\sim\ar[d]^{H^0} & && \mathfrak{Per}\, \Gamma^\prime \ar[d]^{H^0}\\
                    \mathsf{mod}\, J(Q,W) \ar@{<.>}[rrr]^{mutation}& && \mathsf{mod}\, J(\mu_i(Q,W)) }
\end{displaymath}
\end{enumerate}
\end{theorem}

\begin{definition}
\label{def:clustertilt}
 Let $\mathcal{C}$ be a Hom-finite triangulated category. An object $T \in \mathcal{C}$ is
called \emph{cluster-tilting} (or \emph{2-cluster-tilting}) if $T$ is basic (i.e.\!\! with pairwise non-isomorphic direct summands) and if we have
$$\mathsf{add}\,T = \big\{X \in \mathcal{C}\; \big|\; \mathrm{Hom}_{\mathcal{C}}(X, T [1]) = 0\big\} = \big\{X \in \mathcal{C}\; \big|\; \mathrm{Hom}_{\mathcal{C}}(T, X[1]) = 0\big\}.$$
Note that a cluster-tilting object is maximal rigid (the converse is not always true, see
\cite{burban2008cluster}), and that the second equality in the definition always holds when $\mathcal{C}$ is
2-Calabi-Yau.
\end{definition}
If there exists a cluster-tilting object in a 2-CY category $\mathcal{C}$, then it is possible
to construct others by a recursive process resumed in the following:
\begin{theorem}[Iyama-Yoshino \cite{iyama2008mutation}] Let $\mathcal{C}$ be a Hom-finite 2-CY triangulated category
with a cluster-tilting object $T$. Let $T_i$ be an indecomposable direct summand
of $T \cong T_i \oplus T_0$. Then there exists a unique indecomposable $T^*_i$ non isomorphic
to $T_i$ such that $T_0 \oplus T^*_i$ is cluster-tilting. Moreover $T_i$ and $T^*_i$ are linked by the
existence of triangles
$$T_i \overset{u}{\to}B\overset{v}{\to}T^*_i \overset{w}{\to} T_i[1] \ \ and \ \ T^*_i\overset{u'}{\to}B'\overset{v'}{\to}T_i\overset{w'}{\to}T^*_i[1]$$
where $u$ and $u'$ are minimal left $\mathsf{add}\,T_0$-approximations and $v$ and $v'$ are minimal
right $\mathsf{add}\,T_0$-approximations.
\end{theorem}
These triangles allow to make a mutation of the cluster-tilting object: they are called IY-mutations.
\begin{proposition}[Keller-Reiten \cite{keller2007cluster}]\label{kellerrein} Let $\mathcal{C}$ be a 2-CY triangulated category with
a cluster-tilting object $T$. Then the functor
\begin{equation}\label{krfunctor}
F_T = \mathrm{Hom}_{\mathcal{C}}(T, -)\colon  \mathcal{C} \to \mathsf{mod}\, \mathrm{End}_{\mathcal{C}}(T)
\end{equation}
induces an equivalence 
$$
\mathcal{C}\big/\mathsf{add}\,T [1] \cong \mathsf{mod}\, \mathrm{End}_{\mathcal{C}}(T ).
$$
If the objects $T$ and $T^\prime$ are linked by an IY-mutation, then the categories
$\mathsf{mod}\, \mathrm{End}_{\mathcal{C}}(T )$ and $\mathsf{mod}\,\mathrm{End}_{\mathcal{C}}(T^\prime)$ are \emph{nearly Morita equivalent,} that is, there exists
a simple $\mathrm{End}_{\mathcal{C}}(T)$-module $S$, and a simple $\mathrm{End}_{\mathcal{C}}(T^\prime)$-module $S^\prime$, and an equivalence
of categories 
$$\mathsf{mod}\, \mathrm{End}_{\mathcal{C}}(T)\big/\mathsf{add}\,S \cong \mathsf{mod}\, \mathrm{End}_{\mathcal{C}}(T^\prime)\big/\mathsf{add}\, S^\prime.$$
Moreover, if $X$ has no direct summands in $\mathsf{add}\,T[1]$, then $F_TX$ is \emph{projective} (resp.\! \emph{injective}) if and only if $X$ lies in $\mathsf{add}\,T$ (resp.\! in $\mathsf{add}\,T[2]$).
\end{proposition}
Thus, from \textbf{Theorem \ref{amiotthm}} and the above \textbf{Proposition}, we get that in the Jacobi-finite
case, for any cluster-tilting object $T \in \mathcal{C}(\Gamma)$ which is IY-mutation equivalent to
the canonical one, we have:
\begin{displaymath}
    \xymatrix{ & \mathcal{C}(\Gamma) \ar[dl]_{F_{T^\prime}}\ar[dr]^{F_T}& \\
                    \mathsf{mod}\, \mathrm{End}_{\mathcal{C}(\Gamma)}(T) \ar@{<.>}[rr]^{mutation}& & \mathsf{mod}\,\mathrm{End}_{\mathcal{C}(\Gamma)}(T^\prime)}
\end{displaymath}

\subsection{Grothendieck groups,
skew-symmetric pairing, and the index}

\subsubsection{Motivations from physics}\label{kkkzzz21mmc}

In a quantum theory there are two distinct notions of `quantum numbers': the  quantities which are conserved in all physical processes and, on the other hand, the numbers which are used to label (i.e.\! to distinguish) states and operators.  
If a class of BPS objects is described (in a certain physical set-up) by the triangle category $\mathfrak{T}$, 
these two notions of `quantum numbers' get identified as follows:
\begin{itemize}
\item \textbf{conserved quantities:}  numerical invariants of objects $X\in\mathfrak{T}$ which only depend on their Grothendieck class $[X]\in K_0(\mathfrak{T})$.\footnote{\ In general, the conserved quantum \emph{numbers} take value in the \emph{numeric} Grothendieck group $K_0(\mathfrak{T})_\text{num}$. For the categories we consider in this paper, the Grothendieck group is a finitely generated Abelian group and the two groups coincide.} This is the free Abelian group over the isoclasses of objects of $\mathfrak{T}$ modulo the relations given by distinguished triangles of $\mathfrak{T}$;
\item \textbf{labeling  numbers:} correspond to  numerical invariants of the objects $X\in\mathfrak{T}$ which are well-defined, that is, depend only on its isoclass (technically, on their class in the split-Grothendieck group).
\end{itemize}

Of course, conserved quantities are in particular labeling  numbers.
Depending on the category $\mathfrak{T}$ there may be or not be   enough conserved quantities $K_0(\mathfrak{T})$ to label all the relevant BPS objects.
\medskip

In the categorical approach to the BPS sector of a supersymmetric  theory, the basic problem takes the form:
\begin{problem}
Given a class of BPS objects $A$ in a specified physical set-up, determine the corresponding triangulated category $\mathfrak{T}_A$.
\end{problem} 

The Grothendieck group is a very handy tool to solve this \textbf{Problem}. Indeed, the BPS objects of $A$ carry certain conserved quantum numbers which satisfy a number of physical consistency requirements. The allowed quantum numbers take value in an Abelian group $\mathsf{Ab}_A$, and the consistency requirements endow the group with some extra mathematical structures. Both the group $\mathsf{Ab}_A$ and the extra structures on it are known from physical considerations (we shall review the ones of interest in \S.\,\ref{physpre}). Then suppose we have a putative solution $\mathfrak{T}_A$ of the above problem. We compute its Grothendieck group; if $K_0(\mathfrak{T}_A)\not\cong \mathsf{Ab}_A$, we can rule out
$\mathfrak{T}_A$ as a solution of the above \textbf{Problem}. Even if 
$K_0(\mathfrak{T}_A)\cong \mathsf{Ab}_A$, but $K_0(\mathfrak{T}_A)$ is not naturally endowed with the required extra structures, we may rule out $\mathfrak{T}_A$. On the other hand, if we find that $K_0(\mathfrak{T}_A)\cong \mathsf{Ab}_A$ and the Grothendieck group is canonically equipped with the physically expected structures, we gain confidence on the proposed solution, especially if the requirements on $K_0(\mathfrak{T}_A)$ are quite restrictive.  

Therefore, as a preparation for the discussion of their physical interpretation in section 4, we need to analyze in detail the Grothendieck groups of the three triangle categories $D^b\Gamma$, $\mathfrak{Per}\,\Gamma$, or $\mathcal{C}(\Gamma)$. These  categories are related by the functors $\mathsf{s}$, $\mathsf{r}$ which, being exact, induce group homomorphisms between the corresponding Grothendieck groups.  In all three cases $K_0(\mathfrak{T})$ is a finitely generated Abelian group carrying additional structures;
later in the paper we shall compare this structures with the one required by quantum physics. 

\subsubsection{The lattice $K_0(D^b\Gamma)$ and the skew-symmetric form}
The group $K_0(D^b\Gamma)$ is easy to compute using the following
\begin{proposition}
[Keller \cite{keller2011cluster}] The Abelian category $\mathsf{nil}\,J(Q,W)$ is the heart of a bounded $\boldsymbol{t}$-structure in $D^b\Gamma$.
\end{proposition}
Hence, since we assume $(Q,W)$ to be Jacobi-finite, $\mathsf{nil}\,J(Q,W)\cong \mathsf{mod}\,J(Q,W)$ and
$$K_0(D^b\Gamma)\simeq K_0(\textsf{mod}\,J(Q,W))$$
 is isomorphic to the free Abelian group over the isoclasses $[S_i]$ of the simple Jacobian modules $S_i$, that is,  $K_0(D^b\Gamma)\cong\Z^n$
($n$ being the number of nodes of $Q$).

 $D^b\Gamma$ is 3-CY, and then the lattice $K_0(D^b\Gamma)$ is equipped with an intrinsic skew-symmetric pairing given by the Euler characteristics, see discussion around eqn.\eqref{eq:euler}. This pairing has an intepretation in terms of modules of the Jacobian algebra $B\equiv J(Q,W)\cong \mathrm{End}_{\mathcal{C}(\Gamma)}(\Gamma)$.

\begin{proposition}[Palu \cite{palu2009grothendieck}]\label{ppaq} Let  $X,Y\in\mathsf{mod}\,B$. Then the  form
$$\langle X, Y\rangle_a= 
\dim \mathrm{Hom}(X,Y)-\dim \mathrm{Ext}^1(X,Y)-\dim \mathrm{Hom}(Y,X)+\dim \mathrm{Ext}^1(Y,X)$$
descends to an antisymmetric form on $K_0(\mathsf{mod}\,B)$. Its matrix in the basis of simples $\{S_i\}$ is the exchange matrix $B$ of the quiver $Q$ {\rm (cfr.\! eqn.\eqref{uuuia}).}
\end{proposition}

In conclusion, for the 3-CY category $D^b\Gamma$, the Grothendieck group is a rank $n$ lattice equipped with a skew-symmetric bilinear form $\langle-,-\rangle$.
We shall refer to the radical of this form as the \emph{flavor lattice}
$\Lambda_\text{flav}=\mathrm{rad}\,\langle-,-\rangle$.

\subsubsection{$K_0(\mathfrak{Per}\,\Gamma)\cong K_0(D^b\Gamma)^\vee$}\label{k0per}

More or less by definition, $K_0(\mathfrak{Per}\,\Gamma)$ is the free Abelian group over the classes $[\Gamma_i]$ of indecomposable summands of $\Gamma$. Since the general perfect object has infinite homology, there is no well-defined Euler bilinear form. However,
eqn.\eqref{uuuu} implies that for $X\in\mathfrak{Per}\,\Gamma$, $Y\in D^b\Gamma$,
$$\mathrm{Hom}_{\mathfrak{Per}}(X,Y[k])=\mathrm{Hom}_{\mathfrak{Per}}(Y,X[k])=0\quad \text{for }k<0\ \text{or }k>3$$
and hence we have a Euler pairing
$$
K_0(\mathfrak{Per}\,\Gamma)\times K_0(D^b\Gamma)\to \Z,
$$
under which
$$
\langle \Gamma_i, S_j\rangle=-\langle S_j, \Gamma_i\rangle=\delta_{ij}.
$$
Thus $[S_i]$ and $[\Gamma_i]$ are dual basis and both Grothendieck groups are free (i.e.\! lattices) of rank $n$. Then we have two group isomorphisms
\begin{align}
&\Z^n\to K_0(D^b\Gamma) &&(m_1,m_2,\dots,m_n) \longmapsto \bigoplus_{i=1}^n m_i [S_i]\\
&K_0(\mathfrak{Per}\,\Gamma)\to\Z^n
&& [X]\longmapsto \Big(\langle X, S_1\rangle,\langle X, S_2\rangle,\dots,
\langle X, S_n\rangle\Big).
\end{align} 
The image of $K_0(D^b\Gamma)$ inside $K_0(\mathfrak{Per}\,\Gamma)\cong Z^n$ is isomorphic to the image of $B\colon \Z^n\to\Z^n$ where $B$ is the exchange matrix of the quiver $Q$.\footnote{\ Note that this image is invariant under quiver mutation.}
We have the obvious isomorphism
$$
K_0(\mathfrak{Per}\,\Gamma)\cong K_0(\mathsf{add}\,\Gamma).$$

\subsubsection{The structure of $K_0(\mathcal{C}(\Gamma))$}\label{k0clust}

From the basic exact sequence of categories \eqref{treexa} we get
$$\xymatrix{0 \ar[r] & K_0(D^b \Gamma) \ar[r]^s & K_0(\mathfrak{Per}\,\Gamma) \ar[r]^r & K_0(\mathcal C(\Gamma))\ar[r] & 0}$$
hence 
\begin{equation}\label{kkkaq12}
K_0(\mathcal{C}(\Gamma))\cong \Z^n\big/B\cdot \Z^n.
\end{equation}
$K_0(\mathcal{C}(\Gamma))$ is not a free Abelian group (in general) but has a torsion part which we denote as $\mathsf{tH}$ (and call the \emph{'t Hooft group})
\begin{equation}\label{uvgroup}
K_0(\mathcal{C}(\Gamma))= K_0(\mathcal{C}(\Gamma))_\text{free}\oplus\mathsf{tH}\cong \Z^f\oplus \mathsf{A}\oplus \mathsf{A}
\end{equation}
where $f=\mathrm{corank}\,B$ and $\mathsf{A}$ is the torsion group
$$
\mathsf{A}= \bigoplus_s \Z/d_s \Z,\qquad d_s\mid d_{s+1}
$$
where the $d_s$ are the positive integers in the normal form of $B$ \cite{intematrix,barot2008grothendieck}
\begin{equation}\label{normaaaal}
B \xrightarrow{\;\text{normal form}\;}
\overbrace{\,0\oplus 0\oplus \cdots \oplus 0\,}^{f\ \text{summands}}\;\oplus \begin{bmatrix}0 & d_1\\
-d_1 &0\end{bmatrix}\oplus
\begin{bmatrix}0 & d_2\\
-d_2 &0\end{bmatrix}\oplus \cdots\oplus \begin{bmatrix}0 & d_\ell\\
-d_\ell &0\end{bmatrix}
\end{equation}

\subsubsection{The index of a cluster object}\label{kkazqw}

Since the rank of the Abelian group $K_0(\mathcal{C}(\Gamma))$ is (in general) smaller than $n$, the Grothendieck class is not sufficient to label different objects (modulo deformation). We need to introduce other `labeling quantum numbers' which do the job. This corresponds to the math concept of \emph{index} (or, dually, \emph{coindex}).

\begin{lemma}[Keller-Reiten \cite{keller2007cluster}] For each object $L\in\mathcal{C}(\Gamma)$ there is a triangle
$$T_1\to T_0\to L\to\qquad\text{with } T_1,T_0\in\mathsf{add}\,\Gamma.$$
The difference
$$[T_0]-[T_1]\in K_0(\mathsf{add}\,\Gamma)$$
does not depend on the choice of this triangle.
\end{lemma}

\begin{definition}\label{def:index}The quantity
$$\mathsf{ind}(L)\equiv [T_0]-[T_1]\in
K_0(\mathsf{add}\,\Gamma)\equiv
K_0(\mathsf{proj}\,J(Q,W))\cong K_0(\mathfrak{Per}\,\Gamma)\cong\Lambda^\vee$$
is called the \emph{index} of the object $L\in\mathcal{C}(\Gamma)$.
\end{definition}

It is clear from the \textbf{Lemma} that the class $[L]\in K_0(\mathcal{C}(\Gamma))$ is the image of $\mathsf{ind}(L)$ under the projection
$$\Lambda^\vee\to \Lambda^\vee/B\cdot \Lambda.$$

As always, we use the canonical cluster-tilting object $\Gamma$; the modules $F_\Gamma \Gamma_i\in \mathsf{mod}\,J(Q,W)$;
 are the indecomposable projective modules (cfr.\! \textbf{Proposition \ref{kellerrein}}). We write $S_i\equiv \mathsf{Top}\,F_\Gamma \Gamma_i\in\mathsf{mod}\,J(Q,W)$  for the simple with support at the $i$-th node.

\begin{lemma}[Palu \cite{palu2008cluster}]  Let $X\in\mathcal{C}(\Gamma)$ be indecomposable. Then
$$
\mathsf{ind}\, X=\begin{cases}-[\Gamma_i] &X\cong \Gamma_i[1]\\
\sum_{i=1}^n \langle F_\Gamma X,S_i\rangle\,[\Gamma_i] &\text{otherwise,}
\end{cases}$$
where $\langle -,-\rangle$ is the Euler form in $\mathsf{mod}\,J(Q,W)$.
\end{lemma}

\begin{remark} The dual notion to the index is the \emph{coindex} \cite{palu2008cluster}. For $X\in\mathcal{C}(\Gamma)$ one has
\begin{gather}
\mathsf{ind}\,X=-\mathsf{coind}\,X[1]\\
\mathsf{coind}\,X-\mathsf{ind}\,X=\sum_{i=1}^n \langle S_i, F_\Gamma X\rangle_\text{a}\,[\Gamma_i]\label{oooq2}\\
\mathsf{coind}\,X-\mathsf{ind}\,X\ \ \text{depends only on }F_\Gamma X\in \mathsf{mod}\,J(Q,W).
\end{gather}
From \eqref{oooq2} it is clear that the projections in $K_0(\mathcal{C}(\Gamma))$ of the index and coindex agree.
\end{remark}

The precise mathematical statement corresponding to the rough idea that the `index yields enough quantum numbers to distinguish operator'
is the following
\begin{theorem}[Dehy-Keller \cite{dehy2008combinatorics}] Two rigid objects of $\mathcal{C}(\Gamma)$ are isomorphic if and only if their indices are equal.
\end{theorem}

\begin{remark}
We shall show in \S.\ref{sec:UVcharge} how this is related to UV completeness of the corresponding QFT.
\end{remark}

\subsection{Periodic subcategories,  the normalized Euler and Tits  forms}\label{rrem}

We have seen that the group $K_0(D^b\Gamma)$ has an extra structure namely a skew-symmetric pairing. It is natural to look   
 for additional structures on the group $K_0(\mathcal{C}(\Gamma))$.
The argument around \eqref{eq:euler} implies that the Euler form of the 2-CY category $\mathcal{C}(\Gamma)$ \emph{\underline{if defined}} is symmetric:
\begin{equation*}
\begin{split}\langle X,Y\rangle_{\mathcal{C}(\Gamma)}&\equiv\sum_{k\in \Z} (-1)^k\,\dim \mathrm{Hom}_{\mathcal{C}(\Gamma)}(X,Y[k])=\\
&=\sum_{k\in \Z} (-1)^{2-k}\,\dim \mathrm{Hom}_{\mathcal{C}(\Gamma)}(Y, X[2-k])=\langle Y, X\rangle_{\mathcal{C}(\Gamma)}.
\end{split}
\end{equation*}
However the sum in the \textsc{rhs}  is typically not defined, since it is \emph{not true} (in general) that $\mathrm{Hom}_{\mathcal{C}(\Gamma)}(X,Y[k])=0$ for  $k\ll 0$. In order to remediate this, we  introduce an alternative concept.

\begin{definition}\label{def:periodi} We say that a full subcategory $\mathcal{F}(p)\subset\mathcal{C}(\Gamma)$, closed under shifts, direct sums and summands, is \emph{$p$-periodic} ($p\in\mathbb{N}$) iff
the functor $[p]$ restricts to an equivalence in $\mathcal{F}(p)$, and $\mathcal{F}(p)$ is \emph{maximal} with respect to these properties.
Note that we do not require $p$ to be the minimal period.
\end{definition}

\begin{lemma} A $p$-periodic sub-category, $\mathcal{F}(p)\subset \mathcal{C}(\Gamma)$, is triangulated and 2-CY\footnote{\ $\mathcal{F}(p)$ is linear, Hom-finite, and 2-CY. However, it is not necessarily a generalized cluster category since it may or may not have a tilting object.  The prime examples of such a category without a tilting object are the \emph{cluster tubes}, see \cite{barot2008grothendieck,barot}. Sometimes the term `cluster categories' is extended also to such categories.} and the inclusion functor
$\mathcal{F}(p)\xrightarrow{\mathsf{p}} \mathcal{C}(\Gamma)$ is exact.
\end{lemma}

\begin{proof} Since $\mathcal{F}(p)$ is closed under shifts, direct sums, and summands in $\mathcal{C}(\Gamma)$, it suffices to verify that $X,Y\in \mathcal{F}(p)$ implies $Z\in\mathcal{F}(p)$ for all triangles $X\to Y\to Z\to$ in $\mathcal{C}(\Gamma)$. Applying $[p]$ to the triangle, one gets $Z[p]\simeq Z$.
\end{proof}

\begin{definition}
Let $\mathcal{F}(p)\subset\mathcal{C}(\Gamma)$ be $p$-periodic. We define the \emph{normalized Euler form as}
\begin{equation}\label{euflavor}
\langle\!\langle X,Y\rangle\!\rangle=\langle\!\langle Y,X\rangle\!\rangle =\frac{1}{p}\sum_{k=0}^{p-1}(-1)^k\dim \mathrm{Hom}_{\mathcal{C}(\Gamma)}(X,Y[k]),\quad\  X,Y\in\mathcal{F}(p).\end{equation}
Note that it is independent of the chosen $p$ as long as $Y[p]\cong Y$.
\end{definition}

\begin{remark} If $p$ is odd, $\langle\!\langle-,-\rangle\!\rangle\equiv0$.
\end{remark}

\begin{proposition}
The normalized Euler form $\langle\!\langle -,-\rangle\!\rangle$ induces a symmetric form on the group
$$
K_0(\mathcal{F}(p))/K_0(\mathcal{F}(p))_\text{torsion},$$
which we call the \emph{Tits form} of $\mathcal{F}(p)$.
\end{proposition}

\begin{remark}
We shall see in \S.\,\ref{kkkaq12} the physical meaning of the periodic sub-categories and their Tits form.
\end{remark}

\subsubsection{Example: cluster category of the projective line of weights (2,2,2,2)}\label{spi8}

As an example of Tits form in the sense of the above \textbf{Proposition}, we consider the cluster category (see \S.\ref{herditarycase})
$$
\mathcal{C}= D^b(\mathcal{H})\big/\langle\tau^{-1}[1]\rangle^\Z,\qquad \text{where }\mathcal{H}=\mathsf{coh}\,\mathbb{X}(2,2,2,2)
$$
which corresponds to $SU(2)$ SQCD with $N_f=4$ \cite{cecotti2013categorical,cecotti2015higher}. We may think of this cluster category as having the same objects as $\mathsf{coh}\,\mathbb{X}(2,2,2,2)$ and extra arrows \cite{barot}. In this case $\deg\omega=0$, and hence the category $\mathcal{C}$ is triangulated and periodic of period $p=2$ in the sense of \textbf{Definition \ref{def:periodi}}, so $\mathcal{F}(2)$ is the full cluster category $\mathcal{C}$.
We write $\mathcal{O}$ for the structure sheaf and $\mathcal{S}_{i,0}$ for the unique simple sheaf with support at the $i$-th special point such that $\mathrm{Hom}_{\mathsf{coh}\,\mathbb{X}}(\mathcal{O},\mathcal{S}_{i,0})\cong k$. The cluster Grothedieck group
$K_0(\mathcal{C})$ is generated by
$[\mathcal{O}]$ and $[\mathcal{S}_{i,0}]$ ($i=1,2,3,4$) subjected to the relation \cite{barot2008grothendieck}
\begin{equation}\label{kkaqzxcc}
2[\mathcal{O}]=\sum_{i=1}^4[\mathcal{S}_{i,0}].
\end{equation}
Thus we may identify
$$
K_0(\mathcal{C})\cong \Big\{ (w_1,w_2,w_3,w_4)\in \left(\tfrac{1}{2}\Z\right)^2\;\Big|\; w_i=w_j\mod 1\Big\}\equiv\Gamma_{\text{weight},\,\mathfrak{spin}(8)}.
$$
by writing a class as $\sum_iw_i[\mathcal{S}_{i,0}]$.
The Tits pairing is
$$
\langle\!\langle [\mathcal{S}_{i,0}],
[\mathcal{S}_{j,0}]\rangle\!\rangle=\delta_{i,j},
$$
Then $K_0(\mathcal{C})$ equipped with this pairing is isomorphic to the $\mathfrak{spin}(8)$ weight lattice equipped with its standard inner product 
(valued in $\tfrac{1}{2}\Z$) dual to the even one given on the root lattice by the Cartan matrix. We remark that a class in $K_0(\mathcal{C})$ is a \emph{spinorial} $\mathfrak{spin}(8)$ weight iff it is of the form
$k[\mathcal{O}]+\sum_i m_i[\mathcal{S}_{i,0}]$ ($m_i\in\Z$)
with $k$ \emph{odd}. The physical meaning of this statement and eqn.\eqref{kkaqzxcc} will be clear in \S.\,\ref{nonabeenh}.

\subsection{Stability conditions for Abelian and triangulated categories}
\label{sec:stab}
We start with the Abelian category case, since it all boils down to it. The main reference for this part is \cite{bridgeland2007stability}.
Let $ \ca $ be an Abelian category and $K_0( \ca ) $ its Grothendieck group.
\begin{definition} 
A \emph{Bridgeland stability condition} on an Abelian category $\mathcal A$ is a group homomorphism
$$Z : K_0( \ca) \to \C,$$
satisfying certain properties:\footnote{\ If $[X]\in K_0(\ca)$ is the class of $X\in\ca$, we write simply $Z(X)$ for $Z([X])$. }
\begin{enumerate}
\item $Z(\ca) \subset \overline{\mathbb{H}} \setminus \mathbb{R}_{>0}$, the closed upper half plane minus the positive reals;
\item If $Z(E) = 0$, then $E = 0$. This allows to define the map $$\arg Z(-)\colon K_0(\ca)\setminus \{0\}\to (0,\pi];$$
\item The Harder-Narasimhan (HN) property. Every object $E \in  \ca $ admits a filtration
$$0 = E_0 \subset E_1 \subset E_2 \subset \cdot \cdot \cdot \subset E_n = E,$$
such that, for each $i$:
\begin{itemize}
\item $E_{i+1}/E_i$ is $Z$-semistable;\footnote{\ See below \textbf{Definition \ref{semi-sta}} of  semistability of objects in an abelian category.}
\item $\arg Z(E_{i+1}/E_i) > \arg Z(E_{i+2}/E_{i+1})$.
\end{itemize}
\end{enumerate}
\end{definition}
We also have the following
\begin{definition}\label{semi-sta} An object $E\in\ca$ is called \emph{$Z$-stable} if for all nonzero proper subobjects
$E_0 \subset E$,
$$\arg Z(E_0) < \arg Z(E).$$
If $\leq$ replaces $<$, then we get the definition of $Z$-semistable.
\end{definition}
We are now going to give the corresponding definitions for the triangulated categories. The definition is more involved since there is no concept of subobject.

\begin{definition}\label{defslicing} A \emph{slicing} $ \cp$ of a triangulated category $\D$ is a collection of full
additive subcategories $ \cp(\phi)$ for each $\phi \in \R$ satisfying
\begin{enumerate}
\item $ \cp(\phi + 1) =  \cp(\phi)[1]$;
\item For all $\phi_1 > \phi_2$ we have $\mathrm{Hom}(\cp(\phi_1), \cp(\phi_2)) = 0$;
\item For each $0 \neq E \in \D$ there is a sequence $\phi_1 > \phi_2 > \cdot\cdot\cdot > \phi_n$ of real
numbers and a sequence of exact triangles
\begin{displaymath}
    \xymatrix{0=E_0 \ar[rr]& & E_1 \ar[r]\ar[dl]&  \cdots\cdots\ar[r]& E_{n-1} \ar[rr]&&  E_n=E \ar[dl] \\
               & A_1 \ar@{.>}[ul] & & \cdots && A_n \ar@{.>}[ul]& }
\end{displaymath}
with $A_i\in \cp(\phi_i)$ (which we call the Harder-Narasimhan filtration of $E$).
\end{enumerate}
\end{definition}
\begin{remark}
 We call the objects in $ \cp(\phi)$ semistable of phase $\phi$.
\end{remark}
And finally, the definition of stability conditions in a triangulated category.
\begin{definition} A stability condition on a triangulated category $\D$ is a pair $(Z, \cp)$
where $Z : K_0(\D) \to  \C$ is a group homomorphism (called \emph{central charge}) and $ \cp$ is
a slicing, so that for every $0 \neq E \in  \cp(\phi)$ we have
$$Z(E) = m(E)\,  e^{i\pi \phi}$$
for some $m(E) \in \R>0$.
\end{definition}
Indeed, the following proposition shows that to some extent (once we identify a
$\boldsymbol{t}$-structure), stability is intrinsically defined. It also describes how stability conditions
are actually constructed:
\begin{proposition}[\cite{bridgeland2007stability}]\label{prop:bridge} Giving a stability condition $(Z, \cp)$
on a triangulated category $\D$ is equivalent to giving a heart $ \ca$ of a bounded $\boldsymbol{t}$-structure with a stability
function $Z_ \ca : K_0( \ca) \to \C$ such that $( Z_ \ca, \ca)$ have the
Harder-Narasimhan property, i.e.\! any object in $\ca$ has a HN-filtration by $Z_ \ca$-stable
objects.
\end{proposition}
We will focus on how to obtain a stability condition from the datum $( Z_ \ca, \ca)$, as
this is how stability conditions are actually constructed:
\begin{proof} If $\ca$ is the heart of a bounded $\boldsymbol{t}$-structure on $\D$, then we have $K_0(\D) =
K_0( \ca)$, so clearly $Z$ and $Z_ \ca$ determine each other.
Given $( Z_ \ca, \ca)$, we define $ \cp(\phi)$ for $\phi \in (0, 1]$ to be the $Z_ \ca$-semistable objects
in $ \ca$ of phase $\phi(E) = \phi$. This is extended to all real numbers by $ \cp(\phi + n) =
 \cp(\phi)[n] \subset  \ca[n]$ for $\phi \in (0, 1]$ and $0 \neq n \in \mathbb Z$. The compatibility condition 
$$\frac{1}{\pi}\,\arg Z(E)=\phi$$
is satisfied by construction,
so we just need show that $ \cp$ satisfies the remaining properties in our definition of slicing.
The Hom-vanishing condition in definition \ref{defslicing} follows from the definition of heart of a bounded $\boldsymbol{t}$-structure. Finally, given
$E \in \D$, its filtration by cohomology objects $A_i \in  \ca[k_i]$, and the
HN-filtrations $0 \to A_{i1} \to A_{i2} \to \dots \to A_{im_i} = A_i$ given by the HN-property
inside $ \ca$ can be combined into a HN-filtration of $E$: it begins with
$$0 \to F_1 = A_{11}[k_1]\to F_2 = A_{12}[k_1]\to \cdots \to F_{m_1} = A_1[k_1] = E_1,$$
i.e.\! with the HN-filtration of $A_1$. Then the following filtration steps $F_{m_{1+i}}$ are an
extensions of $A_{2i}[k_2]$ by $E_1$ that can be constructed as the cone of the composition
$A_{2i}[k_2]\to A_2[k_2] \to ^{[1]} E_1$ (the octahedral axiom shows that these have the same
filtration quotients as $0 \to A_{21}[k_2] \to A_{22}[k_2]\cdots$ ); continuing this way we obtain a
filtration of $E$ as desired.
Conversely, given the stability condition, we set $ \ca =  \cp((0, 1])$ as before; by
the compatibility condition, the central charge $Z(E)$ of any $ \cp$-semistable
object $E$ lies in $\overline{\mathbb{H}}\setminus \R_{>0}$; since any object in $ \ca$ is an extension of semistable
ones, this follows for all objects in $ \ca$ by the additivity. Finally, it is fairly straightforward
to show that $Z$-semistable objects in $ \ca$ are exactly the semistable objects
with respect to $ \cp$.
\end{proof}

\section{Some physical preliminaries}\label{physpre}
\label{sec:physprel}
In the next section we shall relate the various triangle categories introduced in the previous section to the BPS objects of a $4d$
$\mathcal{N}=2$ QFT as described  from two different points of view: \textit{i)} the microscopic UV description (i.e.\! in terms of a UV complete Lagrangian description or a UV  fixed-point SCFT), and \textit{ii)} the effective Seiberg-Witten IR description. Before doing that, we discuss some general properties of these physical systems. As discussed in \S.\,\ref{kkkzzz21mmc},
the categories $\mathfrak{T}_A$ which describe the BPS objects  should enjoy the categorical versions of these physical properties in order to be valid solutions to the \textbf{Problem} in \S.\,\ref{kkkzzz21mmc}.

\subsection{IR viewpoint}

\subsubsection{IR conserved charges}\label{IRpicture}

The Seiberg-Witten theory \cite{seiberg1994monopoles} describes, in a quantum exact way, the low-energy physics of our $4d$ $\mathcal{N}=2$ model in a given vacuum $u$ along its Coulomb branch. Assuming $u$ and the mass deformations to be generic, the effective theory is an Abelian gauge theory $U(1)^r$ coupled to states carrying both electric and magnetic charges. The flavor group is also Abelian $U(1)^f$, so that the IR conserved charges consist of $r$ electric, $r$ magnetic, and $f$ flavor charges. In a non-trivial theory the gauge group is compact, and the flavor group is always compact, so these charges are quantized. Then the conserved charges take value in a lattice $\Lambda$ (a free Abelian group) of rank
$$
n=2r+f.
$$
The lattice $\Lambda$ is equipped with an extra structure, namely a skew-symmetric quadratic form
$$
\langle-,-\rangle\colon \Lambda\times\Lambda\to \Z,$$
given by the Dirac electro-magnetic pairing. The radical of this form,
$$
\Lambda_\text{flav}=\mathrm{rad}\,\langle-,-\rangle\equiv \Big\{\lambda\in\Lambda\;\Big|\; \langle \mu,\lambda\rangle=0\ \ \forall\;\mu\in\Lambda\Big\}\subset\Lambda,
$$ 
is the lattice of flavor charges and has rank $f$. The effective theory has another bosonic complex-valued conserved charge, namely the central charge of the $4d$
$\mathcal{N}= 2$ superalgebra $Z := \epsilon^{\alpha\beta}\epsilon_{AB}\{Q^A_\alpha, Q^B_\beta\}$. $Z$ is not an independent charge but a linear combination of the charges in $\Lambda$ with complex coefficients which depend on all IR data, and in particular on the vacuum $u$. Hence,
for a given $u$, the susy central charge is a linear map
(group homomorphism)
$$
Z_u\colon \Lambda\to\C.$$
Any given state of charge $\lambda \in \Lambda$ has mass greater than or equal to $|Z_u(\lambda)|$. BPS states are the ones which saturate this bound.

In the case of a $4d$ $\mathcal{N}=2$ with a UV Lagrangian formulation, $r$ and $f$ are the ranks of the (non-Abelian) gauge $G$ and flavor $F$ groups, respectively. At extreme weak coupling, the IR electric and flavor charges are the weights under the respective maximal tori.

\subsubsection{The IR landscape vs.\! the swampland}\label{kkkazx32}

The IR $\mathcal{N}=2$ theories we consider are not generic Abelian gauge theories with electric and magnetic charged matter.
They belong to the \emph{landscape} (as opposed to the \emph{swampland}), that is, they have a well defined UV completion. Such theories have special properties.

One property which seems to be true in the landscape, is that there are ``enough'' conserved IR charges to label all BPS states, so we don't need extra quantum numbers to distinguish the BPS objects in the IR description.
This condition is certainly not sufficient to distinguish the landscape from the swampland, but it plays a special role in our discussion.

To support the suggestion that being UV complete is related to $\Lambda$ being large enough to label IR objects, we mention a simple fact.

\begin{fact} Let the UV theory consists of a $\mathcal{N}=2$ gauge theory with semi-simple gauge group $G$ and quark half-hypermultiplets in the (reducible) quaternionic representation $\boldsymbol{H}$. Assume that the beta-functions of all simple factor of $G$ are non-positive.
In the IR theory along the Coulomb branch, consider the BPS hypermultiplets $h_i$ with zero magnetic charge and write $[h_i]$ for their IR charges in $\Lambda$. Then 
$$
[h_i]=[h_j]\ \text{and }h_i\neq h_j\quad\Rightarrow\quad [h_i]\in\Lambda_\text{flav}.
$$
That is, the charges in $\Lambda$ are enough to distinguish (zero magnetic charge) hypermultiplets unless they carry only flavor charge (i.e.\! are electrically neutral).  
\end{fact}  

\begin{remark} The hypermultiplets with purely flavor charge (called ``everywhere light'' since their mass is independent of the Coulomb branch parameters) just decouple in the IR, so in a sense they are no part of the IR picture.
\end{remark}

To show the above fact, just list for all possible gauge group all representations compatible with non-positivity of the beta-function. Check, using Weyl formula, that the multiplicities of all weights for these representations is 1 except for the zero weight.

\subsection{UV line operators and the 't Hooft group} 
\label{sec:UVcharge}

\subsubsection{'t Hooft theory of quantum phases of gauge theories}\label{extendedthof}
We start by recalling the classical arguments by 't Hooft on the quantum phases of a $4d$ gauge theory \cite{hooft1978phase,hooft1979property,hooft1980confinement,hooft1980topological}. 
The basic order operator in a gauge theory is the Wilson line associated to a  (real) curve $C$
in space time and a representation $\boldsymbol{R}$ of the gauge group $G$,
\begin{equation}\label{wilsonline}
W_{\boldsymbol{R}}(C)= \mathrm{tr}_{\boldsymbol{R}}\,e^{-\int_C A}.
\end{equation}
Here $C$ is either a closed loop or is stretched out to infinity. In the second case we don't take the trace and hence the operator
depends on a choice of a weight $w$ of the representation $\boldsymbol{R}$
modulo che action of Weyl group. In the $\mathcal{N}=2$ case, the Wilson line \eqref{wilsonline} is replaced by its half-BPS counterpart \cite{gaiotto2013framed} which, to preserve half supersymmetries should be stretched along a straight line $L$;
we still denote this operator as $W_w(L)$.\footnote{\ The half-BPS lines are also parametrized by an angle $\vartheta$ which specifies which susy subalgebra leaves them invariant. We suppress $\vartheta$ from the notation.}  

What are the quantum numbers carried by $W_w(L)$?
This class of UV line operators is \emph{labelled} by (the Weyl orbit of) the weight $w$, so gauge weights  are useful quantum numbers. However, these numbers do not correspond to conserved quantities in a general gauge theory. For instance, consider pure (super-)Yang-Mills theory and let $\boldsymbol{R}$ be the adjoint representation. Since an adjoint Wilson line may terminate at the location of a colored particle transforming in the adjoint representation, a gluon (gluino) particle-antiparticle pair may be dynamically created out of the vacuum, breaking the line, see figure \ref{fig:break}.
\begin{figure}[H]
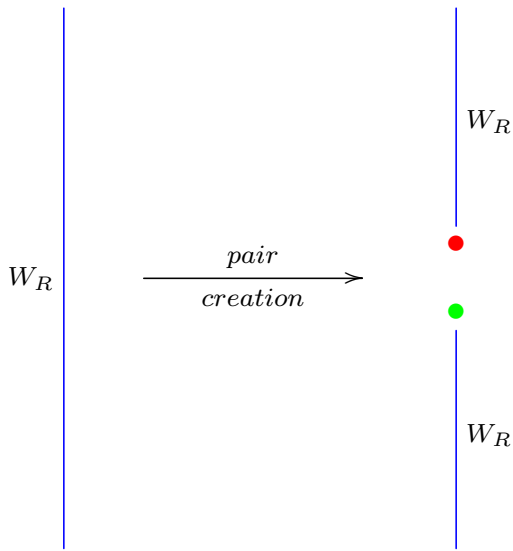

\centering
\resizebox{0.42\textwidth}{!}{\xygraph{
!{<0cm,0cm>;<1.5cm,0cm>:<0cm,1.2cm>::}
!{(0,0) }*+{}="a"
!{(0,-5) }*+{}="b"
!{(2.8,0) }*+{}="c"
!{(2.8,-5)}*+{}="d"
!{(2.8,-2.2)}*+[red]{\bullet}="au"
!{(2.8,-2.8)}*+[green]{\bullet}="ad"
!{(2.2,-2.5)}*+{}="ar1"
!{(0.5,-2.5)}*+{}="ar2"
"a"-@[blue]"b"_{W_R} "c"-@[blue]"au"^{W_{R}} "ad"-@[blue]"d"^{W_{R}} "ar2":"ar1"^{pair}_{creation}
}}
\caption{\textsc{Left}: an electric flux tube line created by an adjoint Wilson line. \textsc{Right}: the adjoint flux line is broken by the creation of a gluon-antigluon pair out of the vacuum.}
\label{fig:break}
\end{figure}

If our gauge theory is in the confined phase,  breaking the line $L$ is energetically favorable, so the line label $w$ does not correspond to a conserved quantity. 
On the contrary, a Wilson line in the \emph{fundamental} representation cannot break in pure $SU(N)$
(S)YM, since there is no dynamical particle which can be created out of the vacuum where it can terminate.
The obstruction to breaking the line is the center $\boldsymbol{Z}(SU(N))\cong \Z_N$ of the gauge group under which all \emph{local} degrees of freedom are inert while the fundamental Wilson line is charged. Stated differently, the gluons may screen all color degrees of freedom of a physical state  but the center of the gauge group. The conclusion is that the conserved quantum numbers of the line operators $W_{\boldsymbol{R}}(L)$ consist of  the representation $\boldsymbol{R}$ seen as a
representation of the center of the gauge group, $\boldsymbol{Z}(G)$,
which take value in the dual group $\boldsymbol{Z}(G)^\vee\cong\boldsymbol{Z}(G)$.
On the other hand, in $SU(N)$
(S)QCD we have quarks transforming in the fundamental representation; hence  
a quark-antiquark pair may be created to break a fundamental Wilson line. Then, in presence of fundamental matter, Wilson lines do not carry any conserved quantum number. In general, the conserved quantum numbers of the Wilson lines of a gauge theory with gauge group $G$ take value in the finite Abelian group
$\pi_1(G_\text{eff})^\vee\cong \pi_1(G_\text{eff})$,
where  $G_\text{eff}$ is 
the quotient group of $G$ which acts \emph{effectively}
on the microscopic UV degrees of freedom.

For clarity of presentation, the above discussion was in the confined phase. This is not the case of the $\mathcal{N}=2$ theory which we assume to be realized in its Coulomb phase.
In the physically realized phase the line $W_w(L)$ may be stable; then  its labeling quantum number $w$ becomes an \emph{emergent} conserved quantity of the IR description (see \S.\ref{llzaq10}). However, from the UV perspective, the only strictly conserved quantum numbers are still the (multiplicative) characters of $\pi_1(G_\text{eff})$ which take value in the group
$$
\pi_1(G_\text{eff})^\vee \equiv \mathrm{Hom}(\pi_1(G), U(1))\cong \pi_1(G_\text{eff}).
$$

More generally, we may have Wilson-'t Hooft lines \cite{hooft1978phase,hooft1979property,hooft1980confinement,hooft1980topological} which carry both electric and magnetic weights. Their multiplicative conserved quantum numbers take value in the (Abelian) \emph{'t Hooft group}
$$
\mathsf{tH}= \pi_1(G_\text{eff})^\vee\oplus
\pi_1(G_\text{eff}),
$$
equipped with the canonical skew-symmetric bilinear pairing (the \emph{Weil pairing})\footnote{\ As always, $\boldsymbol{\mu}$ denotes the group of roots of unity. The name `Weil pairing' is due to its analogy with the Weil pairing in the torsion group of a polarized Abelian variety which arises in exactly the same way.}
$$
\mathsf{tH}\times \mathsf{tH}\to \boldsymbol{\mu},\qquad (x,y)\times (x^\prime,y^\prime)\mapsto x(y^\prime)\,x^\prime(y)^{-1}.
$$

The 't Hooft multiplicative quantum numbers of a line operator, written additively, are just its electric/magnetic weights
$(w_e,w_m)$ modulo the weight lattice of $G_\text{eff}$.
\medskip

The best way to understand the proper UV conserved quantum numbers of line operators is to consider the different sectors in which we may decompose the microscopic path integral of the theory which preserve the symmetries of a line operator stretched in the 3-direction in space (that is, rotations in the orthogonal plane and translations).
In a $4d$ gauge theory quantized on a periodic $3$-box of size $L$
we may defined the 't Hooft twisted
path integral \cite{hooft1979property} (see \cite{hooft1980topological, hooft1980confinement} for nice reviews)
$$
e^{-\beta\,F(\vec e, \vec m, \theta, \mu_s,\beta)} \equiv \mathrm{Tr}_{\vec e, \vec m}\Big[e^{-\beta H+i\theta\nu+\mu_s F_s}\Big],\qquad 
\vec e\in \big(\pi_1(G_\text{eff})^\vee\big)^3,\ \vec m\in \pi_1(G_\text{eff})^3,
$$    
where $\vec e$, $\vec m$ are 't Hooft (multiplicative) electric and magnetic fluxes, $\theta$ is the instanton angle, and $\mu_s$ are chemical potentials in the Cartan algebra of the flavor group $F$. Imposing rotational invariance in the $1-2$ plane and taking the Fourier transform with respect to the $\mu_s$ we remain (at fixed $\theta$)
with the quantum numbers
\begin{equation}\label{kkkaqw}
(e_3, m_3, w)\in \pi_1(G_\text{eff})^\vee\oplus \pi_1(G_\text{eff})\oplus \big(\text{weight lattice of $F$}\big).
\end{equation}
We shall call the vector $(e_3, m_3, w)$ the \textit{'t Hooft charge} and the group in the \textsc{rhs} the \emph{extended 't Hooft group.}

We stress that the structure of the Weil pairing is required in order to relate the Euclidean path integral in given topological sectors  to the free energy $F(\vec e, \vec m, \theta, \mu_s,\beta)$ with fixed non-abelian fluxes \cite{hooft1979property}.

\begin{remark} The boundary condition on the Euclidean box which corresponds to a given 't Hooft charge does no break any supercharges, that is, we do not need to specify a BPS angle $\vartheta$ to define it. 
\end{remark}

\subsubsection{Non-Abelian enhancement of the flavor group in the UV}\label{nonabenh}

Consider a UV complete $\mathcal{N}=2$ gauge theory. In the IR theory
the flavor group is (generically) Abelian of rank $f$. In the UV the masses are irrelevant and the flavor group enhances from the Abelian group $U(1)^f$ to some possibly non-Abelian rank $f$ Lie group $F$.
The free part of the 't Hooft group
\eqref{kkkaqw} is then the weight lattice of $F$. This weight lattice is equipped with a quadratic form dual to the Cartan form on the root lattice.
From the quadratic form we recover the non-Abelian Lie group $F$. 
In conclusion: 

\begin{fact}\label{factthof} The UV conserved quantities are encoded in the extended 't Hooft group, a finitely generated Abelian group  of the form
\begin{equation}\label{nnazq}
\pi_1(G_\text{eff})^\vee\oplus \pi_1(G_\text{eff}) \oplus \Gamma_\text{flav,weight},
\end{equation} 
whose free part has rank $f$. The extended 't Hooft group \eqref{nnazq} is equipped with two additional structures: \emph{i)} the Weil pairing on the torsion part, \emph{ii)} the dual Cartan symmetric form on the free part. Moreover, \emph{iii)} the UV lines carry an adjoint action of the half quantum monodromy $\mathbb{K}$ {\rm (see \S.\,\ref{monoaa})} which acts on the 't Hooft group as $-1$. 
\end{fact}

\paragraph{Finer structures on the 't Hooft group.} The 't Hooft group  \eqref{nnazq} detects the global topology of the gauge group $G_\text{eff}$; it also detects the topology of the flavor group $F$,
e.g.\! it distinguishes between the flavor groups $SO(N)$ and $\mathrm{Spin}(N)$, since they have different weight lattices
$$
\big[\,\Gamma_{\mathfrak{spin}(N)}: \Gamma_{\mathfrak{so}(N)}\,\big]=2.
$$
But there even finer informations on the flavor symmetry which we should be able to recover from the relevant categories. To illustrate the issue, consider $SU(2)$ SQCD with $N_f$ fundamental hypers. In the perturbative sector (states of zero magnetic charge) the flavor group is $SO(2N_f)$, but non-perturbatively it gets enhanced to $\mathrm{Spin}(2N_f)$. More precisely, states of \emph{odd} (resp.\! \textit{even}) magnetic charge are in spinorial (resp.\! tensorial) representations of the flavor group $\mathrm{Spin}(2N_f)$. This is due to the zero modes of the Fermi fields in the magnetic monopole background
\cite{seiberg1994monopoles}, which is turn are predicted by the Atiyah-Singer index theorem. The index theorem is an integrated version of the axial anomaly, so the correlation between magnetic charge and flavor representations should emerge from the same aspect of the category which expresses the $U(1)_R$ anomaly (and the $\beta$-function).

\subsubsection{The effective `charge' of a UV line operator}\label{llzaq10}
\label{sec:uvcharge}
We have two kinds of quantum numbers: conserved quantities and labeling numbers. In the IR we expect (see \S.\ref{kkkazx32}) that conserved quantities are (typically) sufficient to label BPS objects. However, the UV group of eqn.\eqref{kkkaqw} is too small to distinguish inequivalent BPS line operators.

We may introduce a different notion of `charge' for UV operators which takes value in a rank $n=2r+f$ lattice. This notion, albeit referred to UV objects, depends on a IR choice, e.g.\! the choice of a  vacuum $u$. Suppose that in this vacuum we have
$n$ species of stable lines $L_i$ ($i=1,\dots, n$) which are preserved by the the same susy sub-algebra preserving $L$ and  carry emergent IR quantum numbers $[L_i]$ which are $\mathbb{Q}$-linearly independent. We may consider the BPS state $|\{n_i\}\rangle$ in which
we have a configuration of parallel stable lines with $n_1$ of type $L_1$, $n_2$ of type $L_2$, and so on. Suppose that for our BPS line operator $L$
\begin{equation}\label{ooiq}
\big\langle \{n_i^\prime\}\big|\,L\,\big|\{n_i\}\big\rangle\neq 0
\end{equation}
It would be tempting to assign to the operator $L$ the `charge'
$$
\sum_i(n_i^\prime-n_i)[L_i]\in \bigoplus_i \Z[L_i].
$$
Such a charge would be well-defined on UV operators provided two conditions are satisfied: \textit{i)} for all $L$ we can find a pair of states $|\{n_i\}\rangle$, $|\{n_i^\prime\}\rangle$ such that eqn.\eqref{ooiq} holds, and moreover \textit{ii)} we can show that $n^\prime_i-n_i$ does not depend on the chosen $|\{n_i\}\rangle$, $|\{n_i^\prime\}\rangle$. The attentive reader may notice that this procedure is an exact parallel to the definition of the index of a cluster object (\textbf{Definition \ref{def:index}}). However the $i$-th `charge' $n^\prime_i-n_i$ is PCT-odd only if the lines $L$, $L_i$ carry `mutually local charge', that is, have trivial braiding; the projection of the `charge' so defined in the 't Hooft group \eqref{nnazq} is, of course, independent of all choices. This follows from the fact that the action of PCT on  the UV lines is given by the half quantum monodromy (see \S.\,\ref{monoaa}) which does not act as $-1$ on the present `effective' charges; of course, it acts as $-1$ on the 't Hooft charges as it should.

\subsection{The quantum monodromy}\label{monoaa}

There is one more crucial structure on the UV BPS operators, namely the quantum monodromy \cite{cecotti2010r,cecotti2014systems}. Let us consider first the case in which the UV fixed point is a good regular SCFT. At the UV fixed point the $U(1)_r$ $R$-symmetry is restored.
Let $e^{2\pi i r}$ be the operator implementing a $U(1)_r$ rotation by  
$2\pi$ (it acts on the supercharges as $-1$). $e^{2\pi i r}$ acts on a chiral primary operator of the UV SCFT as multiplication by $e^{2\pi i \Delta}$, where $\Delta$ is the scaling dimension of the chiral operator.
Suppose that for all chiral operators $\Delta\in \mathbb{N}$, then $e^{2\pi i r}=(-1)^F$
acts as 1 on all UV observables.
More generally, if all $\Delta\in m \mathbb{N}$ for some integer $m$,
the operator $(e^{2\pi i r})^m$ acts as 1 on observables \cite{cecotti2010r,cecotti2014systems}.

If the theory is asymptotically-free, meaning that the UV fixed point is approached with logarithmic deviations from scaling, the above relations get
also corrected, in a way that may be described rather explicitly, see  \cite{cecotti2014systems}. 

Now suppose we deform the SCFT by relevant operators to flow to the original $\mathcal{N}=2$ theory. We claim that, although the Abelian R-charge $r$ is no longer conserved,  $e^{2\pi ir}$ remain a symmetry in this set up\footnote{\ For the corresponding discussion in 2d, see \cite{Cecotti:2010qn}.} \cite{cecotti2010r}. This is obvious when the dimensions $\Delta$ are integral, since $e^{2\pi ir}$ commutes with the deforming operator. The quantum monodromy $\mathbb{M}$ is the 
operator induced in the massive theory from $e^{2\pi i r}$ in this way \cite{cecotti2010r,cecotti2014systems}.
It is well defined only up to conjugacy,\footnote{\ When the UV fixed point SCFT is non degenerated, the operator $\mathbb{M}$ is semisimple, and its conjugacy class is encoded in its spectrum, that is, the spectrum of dimensions of chiral operators $\Delta\mod1$.}
and may be written as a Kontsevitch-Soibelmann (KS) product of BPS factors ordered according to their phase\footnote{\ In eqn.\eqref{kkxza} we use the notations of \cite{cecotti2010r}: the product is over the BPS stable states of charge $\lambda\in\Lambda$ and spin $s_\lambda$ taken in the clockwise order in their phase $\arg Z_u(\lambda)$; $\psi(z;q)=\prod_{n\geq0}(1-q^{n+1/2}z)^{-1}$ is the quantum dilogarithm, and the $X_\lambda$ are quantum torus operators, i.e.\! they satisfy the algebra $X_\lambda X_{\lambda^\prime}=q^{\langle\lambda,\lambda^\prime\rangle/2} X_{\lambda+\lambda^\prime}$ with $\langle-,-\rangle$ the Dirac pairing. } \cite{cecotti2010r,cecotti2014systems}
\begin{equation}\label{kkxza}
\mathbb{M}=\prod^\circlearrowleft_{\lambda\in \text{BPS}} \Psi(q^{s_\lambda}X_\lambda;q)^{(-1)^{2s_\lambda}}.
\end{equation}
The KS wall-crossing formula \cite{Kontsevich:2008fj,kontsevich2014wall} is simply the statement that the conjugacy class of $\mathbb{M}$, being an UV datum, is independent of the particular massive deformation as well as of the particular BPS chamber we use to compute it (see \cite{cecotti2010r,cecotti2014systems}).

We may also define the half-monodromy $\mathbb{K}$, such that $\mathbb{K}^2=\mathbb{M}$ \cite{cecotti2010r}. The effect of the adjoint action of $\mathbb{K}$ on a line operator $L$ is to produce its PCT-conjugate. Then $\mathbb{K}$ inverts the 't Hooft charges.

We summarize this subsection in the following 

\begin{fact}\label{ppperdioc} If our $\mathcal{N}=2$ has a regular UV fixed-point SCFT and the dimension of all chiral operators satisfy $\Delta\in m\mathbb{N}$ for a certain integer $m$, then $\mathbb{K}^{2m}$ acts as the identity on the line operators. $\mathbb{K}$ acts as $-1$ on the 't Hooft charges.
\end{fact}

\section{Physical meaning of the categories $D^b\Gamma$, $\mathfrak{Per}\,\Gamma$, $\mathcal{C}(\Gamma)$}
\label{sec:physinterpr}
We start this section by reviewing as quivers with (super)potentials arise in the description of the BPS sector of a (large class of) $4d$ $\mathcal{N}=2$ theories, see \cite{cecotti2011classification,alim2013bps,alim2014mathcal,cecotti2012quiver,cecotti2013categorical,del2013four}.

\subsection{$\mathcal{N}=2$ BPS spectra and quivers}
\label{sec:physquiv}

We consider the IR physics of a $4d$ $\mathcal{N} = 2$ model at a generic vacuum $u$ along its Coulomb branch. 
We have the IR structures described in 
\S.\,\ref{IRpicture}: a charge lattice $\Lambda$ of rank $n=2r+f$,
 equipped with an integral skew-symmetric form given by the Dirac electro-magnetic pairing, and a complex linear form given by the $\mathcal{N}=2$ central charge: 
 $$\vev{-,-}\colon \Lambda \times \Lambda \to \Z,\qquad\quad Z_u\colon \Lambda\to\C.$$ 
 A $4d$ $\mathcal{N} = 2$ model has a BPS quiver at $u$
iff there exists a set of $n$ hypermultiplets, stable in the vacuum $u$, such that \cite{alim2014mathcal}: \textit{i)} their charges $e_i\in \Lambda$ generate $\Lambda$, i.e.\! $\Lambda\cong \oplus_i\Z e_i$, and \textit{ii)} the charge of each BPS states (stable in $u$), $\lambda \in \Lambda$, satisfies
$$
\lambda \in \Lambda_+ \ \ \text{or} \ \ -\lambda \in \Lambda_+,
$$
where $\Lambda_+ = \oplus_i \Z_+ e_i$ is the convex cone of `particles' \footnote{\ As contrasted with `antiparticles' whose charges belong to $-\Lambda_+$.}. 
The BPS quiver $Q$ is encoded in the
skew-symmetric $n\times n$ exchange matrix \begin{equation}\label{Bquiver}
B_{ij} :=\vev{e_i, e_j},\qquad i,j=1,\cdots,n.
\end{equation} The nodes of $Q$ are in one-to-one correspondence with the
generators $\{e_i\}$ of $\Lambda$. If $B_{ij} \geq 0$ then there are $|B_{ij}|$ arrows from node $i$ to node $j$; viceversa for $B_{ij}<0$.
\medskip

To find the spectrum of particles with given charge $\lambda = \sum_i m_i e_i \in \Lambda_+$ we may study the effective theory on their world-line. This is a SQM model with four supercharges \cite{Denef:2002ru,alim2014mathcal}, corresponding to the subalgebra of $4d$ susy which preserves the world-line. A particle is BPS in the $4d$ sense iff it
is invariant under 4 supersymmetries, that is, if it is a susy vacuum state of the world-line SQM. The 4-supercharge SQM is based on the quiver $Q$ defined in eqn.\eqref{Bquiver} \cite{Denef:2002ru,alim2014mathcal}: to the $i$--th node there correspond a $1d$ $U(m_i)$ gauge multiplet, while to an arrow $i\to j$ a $1d$ chiral multiplet  in the $(\boldsymbol{\overline{m}}_i,\boldsymbol{m}_j)$ bifundamental representation of the groups at its two ends.
To each oriented cycle in $Q$ there is associated a single-trace gauge invariant chiral operator, namely the trace of the product of the Higgs fields along the cycle.
The (gauge invariant) superpotential of the SQM is a complex linear combination of such operators associated to cycles of $Q$ \cite{alim2014mathcal}. Since we are interested only in the susy vacua, we are free to integrate out all fields entering quadratically in the superpotential.
We remain with a SQM system described by a \emph{reduced} quiver with (super)potential $(Q,W)$ in the sense of section 2.

Then the solutions of the SQM $F$-term
equations are exactly the modules $X$ of the Jacobian algebra\footnote{\ From now on the ground field $k$ is taken to be $\C$.} 
$J(Q, W)$ with dimension vector $\dim X = \lambda\in\Lambda$.

The $D$-term equation is traded for the stability condition \cite{alim2014mathcal}. Given the central charge $Z_u(-)$, we can 
choose a phase $\theta \in [0,2\pi)$ such that $Z_u(\Lambda_+)$ lies inside\footnote{\ $\mathbb{H}$ denotes the upper half-plane $\mathbb{H}:=\{z\in \C\,|\,\mathrm{Im}\,z>0\}$.} $\mathbb H_\theta:=e^{i\theta }\mathbb H$. 
Given a
module $X\in \mathsf{mod}\,J(Q, W)$, we define its stability function as $\zeta(X) := e^{-i\theta}Z_u(X) \in \mathbb H$. The module $X$ is stable iff
$$\arg\zeta(Y)<\arg\zeta(X), \quad \forall\, Y\subset X \text{ proper submodule.}$$
A stable module $X$ is always a \emph{brick,} i.e.\!
$\mathrm{End}_{\mathsf{mod}\,J(Q,W)}X\cong\C$ \cite{cecotti2013categorical}.

Keeping into account gauge equivalence,
the SQM classical vacuum space is the compact
K\"ahler variety \cite{alim2014mathcal}
\begin{equation}\label{whatMgamma}
M_\lambda:=\Big\{
X \in \textsf{mod}\,J(Q,W) \;\Big|\: X\  \text{stable},\ \dim X = \lambda
\Big\}
\Big/\prod_iGL(m_i
, \C),
\end{equation}
that is, the space of isoclasses of stable Jacobian modules of the given dimension $\lambda$.
The space of SQM quantum vacua is then $H^\ast(M_\lambda,\C)$ which carries a representation $\boldsymbol{R}$ of $SU(2)$ by hard Lefschetz \cite{GH,alim2014mathcal}, whose maximal spin is $\dim M_\lambda/2$; the space-time spin content of the charge $\lambda$ BPS particle is\footnote{\ The Cartan generator of $SU(2)_R$ acting on a BPS particle described by a $(p,q)$-harmonic form on $M_\lambda$ is $(p-q)$; however, it is conjectured that only trivial representations of $SU(2)_R$ appear \cite{gaiotto2013framed, del2014absence}.} 
$$\Big(\boldsymbol{0}\oplus \boldsymbol{2}\Big)\otimes \boldsymbol{R}.
$$
 For example, the charge $\lambda$ BPS states consist of a (half) hypermultiplet iff the corresponding moduli
space is a point, i.e.\! if the module $X$ is rigid.
\medskip

The splitting between particles
and antiparticles is conventional:
different choices lead to different pairs $(Q,W)$. However all these $(Q,W)$ should lead to \emph{equivalent} SQM quiver models. Indeed, distinct pairs are related by a chain of $1d$ Seiberg dualities \cite{seiberg1995electric}.
The Seiberg dualities act on $(Q,W)$ as the quiver mutations described in section 2. Indeed, the authors of \cite{derksen2008quivers} modeled their construction on  Seiberg's original work \cite{seiberg1995electric}.
\medskip 

The conclusion of this subsection
is that to a (continuous family of)
$4d$ $\mathcal{N}=2$ QFT (with the quiver property) there is associated a full \emph{mutation-class} of quivers with potentials $(Q,W)$. All $(Q,W)$ known to arise from consistent QFTs are Jacobi-finite, and we assume this condition throughout.
\medskip

Using the mathematical constructions reviewed in \S.\,2,
to such an $\mathcal{N}=2$ theory we naturally associate the three triangle categories
$D^b\Gamma$, $\mathfrak{Per}\,\Gamma$, and $\mathcal{C}(\Gamma)$, together with the functors $\mathsf{s}$, $\mathsf{r}$ relating them. We stress that the association is \emph{intrinsic,} in the sense that  
the  categories are independent of the choice of $(Q,W)$ in the  mutation-class  modulo triangle equivalence (cfr.\! \textbf{Theorem \ref{ttthma}}). 
Our next task is to give a physical interpretation to these three naturally defined categories. We start from the simpler one, $D^b\Gamma$. 

\subsection{Stable objects of $D^b\Gamma$ and BPS states}
Let $\Gamma$ be the Ginzburg algebra associated to the pair $(Q,W)$. Keller proved \cite{keller2011cluster} that the Abelian category $\mathsf{mod}\,J(Q,W)$ is the heart of a bounded $\boldsymbol{t}$-structure in $D^b\Gamma$. In particular, its Grothendieck group is
\begin{equation}\label{grogro}
K_0(D^b\Gamma)\cong K_0\big(\mathsf{mod}\,J(Q,W)\big)\equiv \Lambda,
\end{equation}
that is the lattice of the IR conserved charges (\S.\,\ref{IRpicture}).
Thus, given a stability condition on the Abelian category $\mathsf{mod}\,J(Q,W)$, we can  extend it to the entire triangular category $D^b\Gamma$. In particular, since the semi-stable objects of $D^b\Gamma$ are  the elements of $\cp(\phi)$ (cfr.\! the proof of \textbf{Proposition \ref{prop:bridge}}), we have two possibilities:
\begin{itemize}
\item $\phi \in (0,1],$ then the only semistable objects are the semistable objects of $\mathsf{mod}\,J(Q,W)$ in the sense of ``Abelian category stability'' plus the zero object of $D^b\Gamma$;
\item $\phi \not\in (0,1],$ then the only semistable objects are the shifts of the semistable objects of $\mathsf{mod}\,J(Q,W)$  in the sense of ``Abelian category stability''.
\end{itemize}
In other words, a generic object $E \in D^b\Gamma$ is \emph{unstable if it has a nontrivial HN filtration}. Thus, up to shift $[n]$, the only possible semistable objects in $D^b\Gamma$ are those objects belonging to the heart $\mathsf{mod}\,J(Q,W)$ that are ``Abelian''-stable in it. 
We have already seen that the category $\mathsf{mod}\,J(Q,W)$ describes the BPS spectrum of our $4d$ $\mathcal{N}=2$ QFT: by what we just concluded, the isoclasses of stable objects $X$ of $D^b\Gamma$ with Grothendieck class $[X]=\lambda\in\Lambda$ are parametrized, up to even shifts\footnote{\ Since the shift by $[1]$ acts on the BPS states as PCT, it is quite natural to identify the BPS states associated to stable objects differing by even shifts.}, by the K\"ahler manifolds $M_\lambda\cong M_{-\lambda}$ in eqn.\eqref{whatMgamma} whose cohomology yields the BPS states.

The category $\mathcal{P}(\phi)$ is an Abelian category in its own right.
The stable objects with BPS phase
$e^{i\pi\phi}$ are the simple objects in this category; in particular they are bricks in $\mathcal{P}(\phi)$ hence bricks in $\mathsf{mod}\,J(Q,W)$,
that is,
$$
X\ \text{stable}\quad\Rightarrow\quad\mathrm{End}_{\mathsf{mod}\,J(Q,W)}(X)\cong \C.
$$

\subsection{Grothendieck groups vs.\! physical charges}\label{34zaq}

When the triangle category $\mathcal{T}$ describes a class of BPS objects, the Abelian group
$K_0(\mathcal{T})$ is identified with the conserved quantum numbers carried by those objects. In particular, the group $K_0(\mathcal{T})$ should carry all  the additional structures required by the physics of the corresponding BPS objects, as described in \S.\,\ref{physpre}.

Let us pause a while to discuss the Grothendieck groups of the three triangle categories $K_0(\mathcal{T})$, where $\mathcal{T}=D^b\Gamma$, $\mathfrak{Per}\,\Gamma$, or $\mathcal{C}(\Gamma)$, and check that they indeed possess all properties and additional structures as required by their proposed physical interpretation. 

\subsubsection{$K_0(D^b\Gamma)$} 

Since $D^b\Gamma$ describes BPS particles, $K_0(D^b\Gamma)$ is just the IR charge lattice $\Lambda$, see eqn.\eqref{grogro}.
Physically, the charge lattice carries the structure of a skew-symmetric integral bilinear form, namely the Dirac electromagnetic pairing.
This matches with the fact that, 
since $D^b\Gamma$ is 3-CY,  its Euler form \eqref{eq:euler} is skew-symmetric and is identified with the Dirac pairing (compare eqn.\eqref{Bquiver} and the last part of \textbf{Proposition \ref{ppaq}}). We stress that the pairing is intrinsic (independent of all choices) as it should be on physical grounds.

\subsubsection{$K_0(\mathcal{C}(\Gamma))$: structure}\label{kkkas12}
The structure of the group $K_0(\mathcal{C}(\Gamma))$ was described in \S.\,\ref{k0clust}.
We have
\begin{equation}\label{uvgroup}
K_0(\mathcal{C}(\Gamma))= \Z^f\oplus \mathsf{A}^\vee\oplus \mathsf{A}
\end{equation}
where $\mathsf{A}$ is the torsion group\footnote{\ Of course, $\mathsf{A}^\vee\cong \mathsf{A}$; however it is natural to distinguish the group and its dual.}
$$
\mathsf{A}= \bigoplus_s \Z/d_s \Z,\qquad d_s\mid d_{s+1}
$$
where the $d_s$ are the positive integers appearing in the normal form of $B$, see eqn.\eqref{normaaaal}.

The physical meaning of the Grothendieck group \eqref{uvgroup} is easily understood by considering the case of pure $\mathcal{N}=2$ super-Yang-Mills with gauge group $G$. 
Then one shows \cite{toappear}
$$\mathsf{A}=\boldsymbol{Z}(G)\equiv \text{the center of the (simply-connected) gauge group }G$$
that is
$$K_0(\mathcal{C}(\Gamma_{\text{SYM,}G}))\cong
\boldsymbol{Z}(G)^\vee\oplus 
\boldsymbol{Z}(G).$$  
This is exactly the group of multiplicative quantum numbers labeling the UV Wilson-'t Hooft line operators in the pure SYM case \cite{hooft1978phase},
as reviewed in \S.\,\ref{extendedthof}.
This strongly suggests the identification of the cluster Grothendieck group $K_0(\mathcal{C}(\Gamma))$ with
the group of additive and multiplicative quantum numbers carried by the UV line operators. 

This is confirmed by of the example of $\mathcal{N}=2$ SQCD with (semi-simple) gauge group $G$ and quark hypermultiplets in a (generally reducible) representation $\boldsymbol{R}$. One finds \cite{toappear}
$$
K_0(\mathcal{C}(\Gamma_\text{SQCD}))\cong \Z^{\mathrm{rank}\, F} \oplus \pi_1(G_\text{eff})^\vee\oplus \pi_1(G_\text{eff}),
$$
where $F$ is the flavor group and $G_\text{eff}$ is the quotient of $G$
acting effectively on the UV degrees of freedom. Again, this corresponds to the UV extended 't Hooft group as defined in \S.\,\ref{extendedthof}.
More generally one has:

\begin{fact} In all $\mathcal{N}=2$ theories with a Lagrangian formulation (and a BPS quiver) we have
$$
K_0(\mathcal{C}(\Gamma))\cong \big(\text{the extended 't Hooft group of \S.\,\ref{extendedthof}}\big).$$
\end{fact}
This is already strong evidence that
the cluster category $\mathcal{C}(\Gamma)$ describes UV line operators. For $\mathcal{N}=2$ theories without a Lagrangian, we adopt the above \textbf{Fact} as the definition of the extended 't Hooft group.

From \textbf{Fact \ref{factthof}} we know that the physical 't Hooft group has three additional structures. Let us show that all three structures are naturally present in $K_0(\mathcal{C}(\Gamma))$.

\subsubsection{$K_0(\mathcal{C}(\Gamma))$: action of half-monodromy and periodic subcategories}\label{kkkaq12}

There is a natural candidate for the half-monodromy: on $X\in\mathcal{C}(\Gamma)$, $\mathbb{K}$ acts as $X\mapsto X[1]$ and hence the full monodromy as $X\mapsto X[2]$. 
Then $\mathbb{K}$ acts on $K_0(\mathcal{C}(\Gamma))$ as $-1$, as required. Let us check that this action has the correct physical properties e.g.\! the right periodicity as described in \textbf{Fact \ref{ppperdioc}}. 

\begin{example}[Periodicity for Argyres-Douglas models] We use the notations of \S.\,\ref{herditarycase}. We know that the quantum monodromy $\mathbb{M}$ has a periodicity\footnote{\ For the relation of this fact with the $Y$-systems, see \cite{cecotti2014systems}.} equal to (a divisor of) $h+2$ \cite{cecotti2010r}, corresponding to the fact that the dimension of the chiral operators $\Delta\in \tfrac{1}{h+2}\mathbb{N}$. Indeed, from the explicit description of the cluster category, eqn.\eqref{mzxa},
we have $\mathcal{C}(\mathfrak{g})= D^b(\mathsf{mod}\,\C \mathfrak{g})/\langle \tau^{-1}[1]\rangle^\Z$, so that $\tau\cong [1]$ in $\mathcal{C}(\mathfrak{g})$. Hence,
$$
\mathbb{M}^{h+2} \equiv [h+2] \cong\tau^h[2]=\mathrm{Id},
$$
where we used eqn.\eqref{gabeq}. 
\end{example}

Under the identification $\mathbb{K}\leftrightarrow [1]$, we may rephrase \textbf{Fact \ref{ppperdioc}}
in the form:

\begin{fact} Let $\mathcal{C}(\Gamma)$ be the cluster category associated to a $\mathcal{N}=2$ theory with a \emph{regular} UV fixed-point SCFT such that all chiral operators have dimensions $\Delta \in m\mathbb{N}$. Then $\mathcal{C}(\Gamma)$ is periodic with minimal period $p\mid 2m$. 
If the theory has flavor charges, $p$ is even (more in general: $p$ is even unless the 't Hooft group is a vector space over $\mathbb{F}_2$). In particular, for $\mathcal{N}=2$ theories with a regular UV fixed-point the cluster Tits form $\langle\!\langle [X],[Y]\rangle\!\rangle$ is well-defined.
\end{fact}

\paragraph{Asymptotically-free theories.}
It remain to discuss the  asymptotically-free theories. The associate cluster categories $\mathcal{C}(\Gamma)$ are not periodic. However, from the properties of the 't Hooft group, we expect that, whenever our theory has a non-trivial flavor symmetry, $\mathcal{C}(\Gamma)$ still contains a periodic sub-cluster category of even period. We give an informal argument corroborating this idea which may be checked in several explicit examples.

Sending all non-exactly marginal couplings to zero, our asymptotically-free theory reduces to a decoupled system of free glue and UV regular matter SCFTs. Categorically, this means that cluster category of each matter SCFT, $\mathcal{C}_\text{mat}$ embeds as an additive sub-category in $\mathcal{C}(\Gamma)$ closed under shifts (by PCT). The embedding functor $\iota$ is not exact (in general), so we take the triangular hull of the full subcategory over the objects in its image
$\mathsf{Hu}_\triangle\big((\iota\,\mathcal{C}_\text{mat})_\text{full}\big)\subset \mathcal{C}(\Gamma)$. If the model has non-trivial flavor, at least one matter subsector has non trivial flavor, and the corresponding category $\mathcal{C}_\text{matter}$ is periodic of even period $p$. Its objects satisfy $X[p]\cong X$ and this property is preserved by $\iota$.
The triangle category
$\mathsf{Hu}_\triangle\big((\iota\,\mathcal{C}_\text{mat})_\text{full}\big)$ is generated by these periodic objects and hence is again periodic of period $p$. Then set
$$
\mathcal{F}(p)=\mathsf{Hu}_\triangle\big((\iota\,\mathcal{C}_\text{mat})_\text{full}\big).
$$
Again, the flavor Tits form is well defined.
\medskip

The above discussion shows that the presence of a $p$-periodic subcategory $\mathcal{F}(p)\subset\mathcal{C}$ is related to the presence of a sector in the $\mathcal{N}=2$ theory described  by susy protected operators of dimension $$\Delta= \tfrac{2}{p}\mathbb{N}.$$ Let us present some simple examples.

\begin{example}[Pure $SU(2)$ SYM]\label{kkz19ab} The cluster category $\mathcal{C}_{SU(2)}$ is not periodic; this is a manifestation of the fact that the $\beta$-function of the theory is non zero \cite{cecotti2013categorical}.
However, let us focus on the perturbative ($\equiv$ zero magnetic charge) sector in the $g_\text{YM}\to0$ limit. The chiral algebra is generated by a single operator of dimension $\Delta=2$,
namely $\mathrm{tr}(\phi^2)$. Hence we expect that the zero-magnetic charge sector is described by a subcategory of 
$\mathcal{C}_{SU(2)}$ which is 1-periodic. Indeed, this is correct, $\mathcal{F}(1)$ being a $\mathbb{P}^1$ family of homogenous cluster tubes.
\end{example}

\begin{example}[$SU(2)$ SYM coupled to $D_p$ Argyres-Douglas]
In this case the matter is an Argyres-Douglas theory of type $D_p$; the matter half quantum monodromy $\mathbb{K}_\text{matter}$ has order $(h(D_p)+2)/\gcd(2,h(D_p))=p$ as we may read from the spectrum of chiral ring dimensions of the Argyres-Douglas model \cite{Cecotti:2010qn}. Thus the matter corresponds to a periodic subcategory $\mathcal{F}(p)\subset \mathcal{C}$. This category is a cluster tube of period $p$
\cite{barot2008grothendieck,barot}.
See also \cite{cecotti2013categorical}.
\end{example}

\begin{remark} Equivalently, we may understand that the presence of a non-trivial flavor group implies the existence of a $2$-periodic subcategory $\mathcal{F}(p)\subset \mathcal{C}$ by the fact that the corresponding conserved super-currents have canonical dimension $1$ which cannot be corrected by RG. 
\end{remark}

\subsubsection{$K_0(\mathcal{C}(\Gamma))$: non-Abelian enhancement of flavor}\label{nonabeenh}
As discussed in \S.\,\ref{nonabenh},
the IR flavor symmetry $U(1)^f$ gets enhanced in the $UV$ to a non-Abelian group $F$. The identification of $K_0(\mathcal{C}(\Gamma))$ with the extended 't Hooft group requires, in particular, that its free part is equipped with the correct dual Cartan form for the flavor group $F$.

In \S.\,\ref{rrem} we defined a Tits form associated to (a periodic subcategory of) $\mathcal{C}(\Gamma)$. This is a symmetric form on the free part of the Grothendieck group, and is the natural candidate for the dual Cartan form of the physical flavor group $F$. 
Let us check in a couple of examples that this identification yields the correct flavor group:
the cluster category knows the actual non-Abelian group.

\begin{example}[$SU(2)$ with $N_f\geq1$ fundamentals]\label{Nf3} We use the same notations\footnote{\ However we often write simply $\mathbb{X}$ instead of $\mathbb{X}(p_1,\dots,p_s)$ leaving the weights implicit.} as in \S.\,\ref{spi8}. The cluster category is
$$
\mathcal{C}_{N_f}=D^b\!\big(\mathsf{coh}\,\mathbb{X}(\,\overbrace{2,\cdots,2}^{N_f\ 2's}\,)\big)\big/\langle \tau^{-1}[1]\rangle^\Z.
$$
For $N_f\neq4$ this category is not periodic since the canonical sheaf has non-zero degree (in the physical language: the $\beta$-function is non-zero). We are in the situation discussed at the end of \S.\,\ref{kkkaq12}, and the present example is also an illustration of that issue.

The 2-periodic triangle 2-CY subcategory $\mathcal{F}(2)\xrightarrow{\mathsf{j}}\mathcal{C}(N_f)$ is given by the orbit category of the derived category of finite-length sheaves. 
It consists of a $\mathbb{P}^1$ family of cluster tubes; in $\mathbb{P}^1$ there are $N_f$ special points whose cluster tubes have period $2$.
Let $\mathcal{S}_{i,k}$, $k\in\Z/2\Z$, be the simples in the $i$-th special cluster tube, satisfying
$$\mathcal{S}_{i,k}[1]\cong \tau\mathcal{S}_{i,k}\cong \mathcal{S}_{i,k+1}$$
and let $\mathcal{S}_z$ be the simple over the regular point  $z\in\mathbb{P}^1$, $\tau \mathcal{S}_z\cong \mathcal{S}_z$. Thus $[\mathcal{S}_z]=0$ 
and $K_0(\mathcal{F}(2))$
 is generated by the $[\mathcal{S}_{i,0}]$ ($i=1,\dots, N_f$).
 The image of $K_0(\mathcal{F}(2))$
 in $K_0(\mathcal{C}(N_f))$ has index 2; indeed in $K_0(\mathcal{C}(N_f))$ we have an extra generator $[\mathcal{O}]$ and a relation \cite{barot2008grothendieck}
\begin{equation}\label{kkazqw5}
 2[\mathcal{O}]=\sum_{i=1}^{N_f}[\mathcal{S}_{i,0}]
\end{equation}
 Then as in \S.\,\ref{spi8} (for the special case $N_f=4$) we have
 $$
 K_0(\mathcal{C}(N_f))\cong\Big\{(w_1,\cdots,w_{N_f})\in \left(\tfrac{1}{2}\Z\right)^{N_f}\;\Big|\; w_i=w_j\mod1\Big\}\equiv \Gamma_{\text{weight},\,\mathfrak{spin}(2N_f)}
 $$
 with
 $$
 \langle\!\langle [\mathcal{S}_{i,0}],[\mathcal{S}_{j,0}]\rangle\!\rangle=\delta_{i,j},
 $$
 that is, $K_0(\mathcal{C}(N_f))$ is the $\mathfrak{spin}(2N_f)$ weight lattice equipped with the dual Cartan pairing which is the correct physical extended ' t Hooft group for this model  which has $\pi_1(G_\text{eff})=1$ and $F=\mathrm{Spin}(2N_f)$, as expected. \end{example}

\begin{remark}[$\mathrm{Spin}(8)$ triality] The case of $N_f=4$ was already presented in \S.\,\ref{spi8}. In that case $\deg\mathcal{K}=0$ (i.e.\! $\beta=0$), the theory is UV superconformal, and the cluster category is periodic.
The correlation between magnetic charge and $\mathrm{Spin}(8)$ representation becomes the fact that the modular group $PSL(2,\Z)$ acts on the flavor by triality \cite{seiberg1994monopoles}, see \cite{cecotti2015higher} for details from the cluster category viewpoint. 
\end{remark}

 \subsubsection{Example \ref{Nf3}: Finer flavor structures,  $U(1)_r$ anomaly, Witten effect}\label{pppq12x} The cluster category contains even more detailed information on the UV flavor physics of the corresponding $\mathcal{N}=2$ QFT. Let us illustrate the finer flavor structures in the case  of $SU(2)$ SYM coupled to $N_f$ flavors\footnote{\ Or, more generally, to several Argyres-Douglas systems of type $D$.} (\textbf{Example \ref{Nf3}}).

 Note that the sublattice
 $K_0(\mathcal{F}(2))\subset K_0(\mathcal{C}(N_f))$ is the weight lattice of $SO(2N_f)$; since 
 $\mathcal{F}(2)$ is the cluster sub-category of the `perturbative' (zero magnetic charge) sector, we recover the finer flavor structures mentioned at the end of \S.\,\ref{nonabenh}. In facts, eqn.\eqref{kkazqw5} is the image in the Grothendieck group of the equation which is the categorical expression of the $U(1)_r$ anomaly \cite{cecotti2013categorical}. Indeed, in the language of coherent sheaves, the $U(1)_r$ anomaly is measured by the non-triviality of the canonical sheaf $\mathcal{K}$  (think of a (1,1) $\sigma$-model: $\mathcal{K}$ trivial means the target space is Calabi-Yau, which is the condition of no anomaly). 
 The coefficient of the $\beta$-function, $b$, is (twice) its degree,\footnote{\ Notice that $\deg\mathcal{K}=0$ does not mean that $\mathcal{K}$ is trivial but only that it is a torsion sheaf in the sense that $\mathcal{K}^m\cong\mathcal{O}$ for some integer $m$.} $\deg\mathcal{K}=-\chi(\mathbb{X})$ \cite{cecotti2013categorical}. 
 As a preparation to the examples of \S.\,6, we briefly digress to recall how this comes about.

\paragraph{$\beta$-function and Witten effect.} The AR translation $\tau$ acts on $\mathsf{coh}\,\mathbb{X}$ as multiplication by the canonical sheaf \cite{geigle,lenzing1,cecotti2015higher}
 \begin{equation}\label{jzxce}\tau\colon \mathcal{A}\mapsto\mathcal{A}\otimes \mathcal{K}\equiv\mathcal{A}\otimes\mathcal{O}(\vec\omega).
 \end{equation} 
 Hence the $U(1)_R$ anomaly and $\beta$-function may be read from the action of $\tau$ on the derived category $D^b\mathsf{coh}\,\mathbb{X}$ which we may identify as the IR category of BPS particles.\footnote{\ Indeed, for $N_f\leq 3$, the triangle category $D^b\mathsf{coh}\,\mathbb{X}$ admits $\mathsf{mod}\,\C \hat{\mathfrak{g}}$ as the core of a $\boldsymbol{t}$-structure (here $\hat{\mathfrak{g}}$ is an acyclic affine quiver in the mutation class of the model \cite{cecotti2011classification}; see also \textbf{Example \ref{kzanncv}}.} Now, in the cluster category of a weighted projective line, $\mathcal{C}(\mathsf{coh}\,\mathbb{X})\equiv D^b(\mathsf{coh}\,\mathbb{X})/\langle \tau^{-1}[1]\rangle$, one has $\tau\cong[1]$, while $[1]$ acts in the UV as the half-monodromy, that is, as a UV $U(1)_r$ rotation by $\pi$.   
 In the normalization of ref.\cite{seiberg1994monopoles} (see their eqn.(4.3)), the complexified $SU(2)$ Yang-Mills coupling at weak coupling, $a\to\infty$, is
$$
\frac{\theta}{\pi}+\frac{8\pi i}{g^2}= -\frac{b}{\pi i}\log a+\cdots,
$$
Under a $U(1)_r$ rotation by $\pi$, $a\to e^{\pi i}a$, the vacuum angle shifts as $\theta\to\theta-b\pi$.
Since a dyon of magnetic charge $m$ carries an electric charge
$m\,\theta/2\pi\;\text{mod}\,1$ (the Witten effect \cite{Witten:1979ey}),
under the action of $\tau$ the IR electric/magnetic charges $(e,m)$ should undergo the flow 
\begin{equation}\label{nnna129c}
\tau\colon (e,m)\to (e-m b/2,m).
\end{equation}
For an object of $D^b(\mathsf{coh}\,\mathbb{X})$ the magnetic (electric) charge correspond to its rank (degree); then comparing eqns.\eqref{jzxce},\eqref{nnna129c} we get $b=-2\deg\mathcal{K}=2\,\chi(\mathbb{X}).$ 
 
\paragraph{Finer flavor structures (\S.\,\ref{nonabenh}).}   The Grothendieck group of
 $\mathsf{coh}\,\mathbb{X}(2,\dots,2)$ is generated by  $[\mathcal{O}]$, $[\mathcal{S}_0]$, $[\mathcal{S}_{i,j}]$ ($i=1,\dots,N_f$, $j\in\Z/2\Z$) subjected to the relations
 $[\mathcal{S}_0]=[\mathcal{S}_{i,0}]+[\mathcal{S}_{i,1}]$ $\forall\,i$, see \textbf{Proposition 2.1} of  \cite{barot2008grothendieck}. The action of $\tau$ in $K_0(\mathsf{coh}\,\mathbb{X})$ is
\begin{equation}\label{az129}
 [\tau\mathcal{S}_{i,j}]=[\mathcal{S}_{i,j+1}],\qquad\quad [\tau\mathcal{O}]-[\mathcal{O}]=(N_f-2)[\mathcal{S}_0]-\sum_{i=1}^{N_f}[\mathcal{S}_{i,0}].
\end{equation}
The difference $[\tau\mathcal{O}]-[\mathcal{O}]$ measures the non-triviality of the canonical sheaf, that is, the $\beta$-function/$U(1)_r$ anomaly. In the cluster category, for all sheaf $[\tau\mathcal{A}]=-[\mathcal{A}]$, so that $[\mathcal{S}_{i,0}]=0$ and
the second eqn.\eqref{az129} reduces to \eqref{kkazqw5}. Hence, as suggested by the physical arguments at the end of \S.\,\ref{nonabenh}, the non-perturbative flavor enhancement 
$SO(2N_f)\to\mathrm{Spin}(2N_f)$ follows from the counting of the Fermi zero-modes implied by the axial anomaly.

\subsubsection{$K_0(\mathcal{C}(\Gamma))_\text{torsion}$: the Weil pairing}

Let $X\in \mathcal{C}(\Gamma)$
The projection
$$
\langle S_i, F_\Gamma X\rangle \in \Z^n/B\Z^n,
$$
depends only on $[X]$. Rewrite the integral vector $\langle S_i, F_\Gamma X\rangle$ in the $\Z$-basis where $B$ takes the normal form \eqref{normaaaal}
$$
\big(\langle S_1, F_\Gamma X\rangle,\cdots, \langle S_n, F_\Gamma X\rangle\big)\xrightarrow{\text{normal form basis}} (w_1, w_2,\cdots, w_f, u_{1,1},u_{2,1},\cdots, u_{1,s}, u_{2,s},\cdots)
$$ and see its class  as an element of $(\mathbb{Q}^{2}/\Z^{2})^r$
$$
(w_1, w_2,\cdots, w_f, u_{1,1},u_{2,1},\cdots, u_{1,s}, u_{2,s},\cdots)\mapsto
\left(\frac{u_{1,1}}{d_1},\frac{u_{2,1}}{d_1},\cdots, \frac{u_{1,s}}{d_s},\frac{u_{2,s}}{d_s}\cdots\right) \in (\mathbb{Q}^{2}/\mathbb{Z}^{2})^r.
$$
The skew-symmetric matrix $B$ then defines a skew-symmetric pairing
$$
2\pi i\sum_{s=1}^r \frac{\epsilon^{ab}\,u_{a,s}\,u^\prime_{b,s}}{d_s}\in 2\pi i\,\mathbb{Q}/\mathbb{Z}.
$$
The exponential of this expression is the canonical Weil pairing. Let us check one example.

\begin{example}[Pure $SU(2)$] The basis $[P_1]$, $[P_2]$ is canonical. Then the Weil pairing is
$$
(Z/2\Z)^2\times (Z/2\Z)^2\ni (e,m) \times (e^\prime,m^\prime)\mapsto (-1)^{em^\prime -me^\prime}.
$$
\end{example}

\subsection{The cluster category as the UV line operators}\label{lllaz231}

We have seen that for a $\mathcal{N}=2$ theory (with quiver property) the Grothendieck group $K_0(\mathcal{C}(\Gamma))$ is the extended 't Hooft group of additive and multiplicative conserved quantum numbers of the UV line operators and that this group is naturally endowed with all the structures required by physics, including the finer ones.

This amazing correspondence makes almost inevitable the identification of the the cluster category $\mathcal{C}(\Gamma)$ of the mutation-class of quivers with (super)potentials associated to a $4d$ $\mathcal{N}=2$ model with the triangle category describing its UV BPS line \emph{operators}. This identification has been pointed out by several authors working from different points of view \cite{gaiotto2013framed,Cecotti:2010qn, cordova2013line}. In particular, the structure of the mutations of the $Y$--seeds in the cluster algebras lead to the Kontsevich-Soibelman wall crossing formula \cite{Kontsevich:2008fj} (see \cite{Cecotti:2010qn,cecotti2014systems} for  details).  This is just the action of the shift $[1]$ on the cluster category which implements the quantum half monodromy $\mathbb{K}$ (cfr.\! \S.\,\ref{kkkaq12}).  

In section \ref{sec:surfaces} below we  check explicitly this identification by relating the geometrical description of the cluster category of a surface as given in the mathematical literature with the WKB analysis of line operators by GMN \cite{gaiotto2013wall,gaiotto2013framed}.
\medskip

For BPS line operators we also had a notion of `charge' which is useful to distinguish them, see \S.\,\ref{llzaq10}. We already mentioned there that both the definition and the properties of this `charge' have a precise correspondent in the mathematical notion of the \emph{index} of a cluster object.
Now we may identify these two quantities. Note that, while the 't Hooft charge is invariant under quantum monodromy (i.e.\! under the shift $[2]$), the
index is not. This is the effect of non-trivial wall-crossing and, essentially, measures it \cite{Cecotti:2010qn}.

In \S.\,\ref{kkazqw} we saw that the index is fine enough to distinguish rigid objects of the cluster category. This is reminiscent  of our discussion in \S.\ref{kkkazx32} about a (necessary) condition for UV completeness.  

In section \S.\,\ref{sec:lindef}, building over refs.\cite{gaiotto2013framed, cordova2013line},
we discuss how the interpretation of the cluster category $\mathcal{C}(\Gamma)$ as describing UV BPS line operators $L_{\mathrm{ind}\,X}(\zeta)$ (labeled by the index of the corresponding cluster object $X$ and the phase $\zeta$ of the preserved supersymmetry) leads to concrete expressions for their vacuum expectation values in the vacuum $u$
$$
\langle L_{\mathrm{ind}\,X}(\zeta)\rangle_u.
$$

\subsection{The perfect derived category $\mathfrak{Per}\,\Gamma$}

To complete the understanding of the web of categories and functors 
describing the BPS physics of 
a $4d$ $\mathcal{N}=2$ theory, it remains to discuss the physical meaning of the perfect category 
$\mathfrak{Per}\,\Gamma$.
To the best of our knowledge,
an interpretation of the perfect category of a Ginzburg DG algebra has not appeared before in the physics literature.

We may extract some properties of the BPS objected described by the perfect category already from its Grothendieck group $K_0(\mathfrak{Per}\,\Gamma)$ and the basic sequence of functors
\begin{equation}\label{thesequenceagain}
0\to D^b\Gamma\xrightarrow{\;\mathsf{s}\;} \mathfrak{Per}\,\Gamma
\xrightarrow{\;\mathsf{r}\;}\mathcal{C}(\Gamma)\to 0.
\end{equation}
The Grothendieck group $K_0(\mathfrak{Per}\,\Gamma)$ is isomorphic to the IR charge lattice $\Lambda$, so $\mathfrak{Per}\,\Gamma$ is a category of IR BPS objects whose existence (i.e.\! ``stability'') depends on the particular vacuum $u$. $\mathfrak{Per}\,\Gamma$ yields the description of these physical objects from the viewpoint of the Seiberg-Witten low-energy effective Abelian theory. This is already clear from the fact that $\mathfrak{Per}\,\Gamma$ contains the category describing the IR BPS particles i.e.\!
$D^b\Gamma$; BPS particles then form part of the physics described by $\mathfrak{Per}\,\Gamma$. A general object in $\mathfrak{Per}\,\Gamma\setminus D^b\Gamma$ differs from an object in the category $D^b\Gamma$ in one crucial aspect: its total homology has infinite dimension, so (typically) infinite \textsc{susy} central charge and hence infinite energy. Then $\mathfrak{Per}\, \Gamma$ is naturally interpreted as the category yielding the IR description of
 half-BPS \emph{branes} of some kind. They may have infinite energy just because their volume may be  infinite.
 Although their central charge is not well defined, its phase is: it is just the angle $\theta$ corresponding to the subalgebra of supersymmetries under which the brane is invariant.
 
 On the other hand, the RG functor $\mathsf{r}$ in \eqref{thesequenceagain} associates to each IR object in $\mathcal{O}\in\mathfrak{Per}\,\Gamma\setminus D^b\Gamma$ a non-trivial UV line operator $\mathsf{r}(\mathcal{O})$.
 This suggests a heuristic physical picture: let  $\mathcal{O}\in\mathfrak{Per}\,\Gamma\setminus D^b\Gamma$ describe a BPS brane which is stable in the Coulomb vacuum $u$; this brane should be identified with the ``state'' obtained by acting with the UV line operator $\mathsf{r}(\mathcal{O})$ on the vacuum $u$ as seen in the low-energy Seiberg-Witten effective Abelian theory.  
 
 In order to make this proposal explicit, in the next section we shall consider a particular class of examples, namely the class $\mathcal{S}[A_1]$ theories \cite{gaiotto2012n,gaiotto2013wall}. In this case all three categories $D^b\Gamma$, $\mathfrak{Per}\,\Gamma$ and $\mathcal{C}(\Gamma)$ are explicitly understood both from the Representation-Theoretical side (in terms of string/band modules \cite{assem2010gentle}) as well as in terms of the geometry of curves on the Gaiotto surface $C$.
In this setting BPS objects are also well understood from the physical side since WKB is exact in the BPS sector.

Comparing the mathematical definition of the various triangle categories associated to a class $\mathcal{S}[A_1]$ model, and the physical description of the BPS objects, we shall check that the above interpretation of $\mathfrak{Per}\,\Gamma$ is correct. 

\subsubsection{``Calibrations'' of perfect categories} To complete the story we need to introduce a notion of ``calibration'' on the objects of $\mathfrak{Per}\,\Gamma$ which restricts in the full subcategory $D^b\Gamma$ to the Bridgeland notion of stability.
The specification of a ``calibration''
requires the datum of the Coulomb vacuum $u$ and a phase $\theta=\pi \phi\in\R$. Given an $u$ (corresponding to specifying a central charge $Z$), the
$\phi$-calibrated objects form
a full additive subcategory of $\mathfrak{Per}\,\Gamma$, $\mathcal{K}(\phi)$, such that
$$
\mathcal{P}(\phi)\subset \mathcal{K}(\phi)\subset \mathfrak{Per}\,\Gamma,\qquad \forall\,\phi\in\R.
$$ 
We use the term ``calibration'' instead of ``stability'' since it is quite a different notion with respect to Bridgeland stability (in a sense, it has ``opposite'' properties), and it does not correspond to the physical idea of stability.
These aspects are already clear from the fact that the central charge $Z$ is not defined for general objects in  
$\mathfrak{Per}\,\Gamma$. 

In the special case of the perfect categories arising from class $\mathcal{S}[A_1]$ QFTs, where everything is explicit and geometric, the calibration condition may be expressed in terms of flows of quadratic differentials, see \S.\,\ref{sec:surfaces}. 

We leave a more precise discussion of calibrations for perfect categories to future work. Here we limit ourselves to make some observations we learn from the class $\mathcal{S}[A_1]$ example.

\begin{definition}
 A phase $\pi\phi\in\R$ is called a \emph{BPS phase} if the slice $\mathcal{P}(\phi)\subset D^b\Gamma$ contains non-zero objects. A phase $\pi\phi$ is \emph{generic} if it is not a BPS phase nor an accumulation point of BPS phases.
\end{definition}

\begin{fact} In a class $\mathcal{S}[A_1]$ theory, assume there is no
BPS phase in the range $[\pi\phi,\pi \phi^\prime]$. Then
$$
\mathcal{K}(\phi)\cong \mathcal{K}(\phi^\prime).
$$
Moreover,
let $\pi\phi$ be a \emph{generic} phase. Then the $\phi$-calibrated category  $\mathcal{K}(\phi)\subset \mathfrak{Per}\,\Gamma$ has the form 
$$\mathcal{K}(\phi)\cong \mathsf{add}\,\mathcal{T}_\phi$$
for an object $\mathcal{T}_\phi\in \mathfrak{Per}\,\Gamma$ such that
$$\mathsf{r}(\mathcal{T}_\phi)\in\mathcal{C}(\Gamma)\ \text{is cluster-tilting.}$$
In other words, the generic $\mathcal{T}_\phi$ is a silting object of $\mathfrak{Per}\,\Gamma$.
\end{fact}

We conjecture that something like the above \textbf{Fact} holds for general $4d$ $\mathcal{N}=2$ theories.

\section{Cluster automorphisms and $S$-duality}
\label{sec:sdualmap}

\subsection{Generalities}

A duality between two supersymmetric theories  induces a (triangle) equivalence between
the triangle categories describing its 
BPS objects. The celebrate example is mirror symmetry between 
IIA and IIB string theories compactified on a pair of mirror Calabi-Yau 3-folds, $\mathcal{M}$, $\mathcal{M}^\vee$.
At the level of the corresponding categories of BPS branes, mirror symmetry duality induces \emph{homological mirror symmetry}, that is the equivalences of triangle categories \cite{kontsevich1994homological, kapustin2008homological}
$$D^b(\mathsf{Coh}\,\mathcal{M})\cong D^b(\mathsf{Fuk}\,\mathcal{M}^\vee),\qquad
D^b(\mathsf{Coh}\,\mathcal{M}^\vee)\cong D^b(\mathsf{Fuk}\,\mathcal{M}).$$

In fact, since to
a supersymmetric theory $\mathcal{T}$ we associate a family of triangle categories, $\{\mathfrak{T}_{(a)}\}_{a\in I}$, depending on the class of BPS objects and the physical picture (e.g.\! IR versus UV), a dual pair of theories $\mathcal{T}$, $\mathcal{T}^\vee$, yields
a \emph{family} of equivalences of categories labeled by the index set $I$ 
$$
\mathfrak{T}_{(a)}\xrightarrow{\;\mathsf{d}_{(a)}\;}\mathfrak{T}_{(a)}^\vee\qquad a\in I.
$$
The several categories associated to the theory, $\{\mathfrak{T}_{(a)}\}_{a\in I}$, are related by
physical compatibility functors having the schematic form
$$\mathfrak{T}_{(a)}\xrightarrow{\;\mathsf{c}_{(a,b)}\;}\mathfrak{T}_{(b)}\qquad a,b\in I$$
(e.g.\! the `inverse RG flow' functor $\mathsf{r}$
in eqn.\eqref{exactse}). Physical consistency of the duality then require that we have commutative diagrams of functors of the form
$$
\xymatrix{\mathfrak{T}_{(a)}\ar[r]^{\mathsf{d}_{(a)}}\ar[d]_{\mathsf{c}_{(a,b)}} &
\mathfrak{T}_{(a)}^\vee\ar[d]^{\mathsf{c}_{(a,b)}^\vee}\\
\mathfrak{T}_{(b)}\ar[r]^{\mathsf{d}_{(b)}}&
\mathfrak{T}_{(b)}^\vee}
$$
The philosophy of the present review is that the dualities are better understood in terms of such   diagrams of exact functors  between the relevant triangle categories. This idea may be applied to all kinds of dualities; here we are particularly interested in \emph{auto-dualities}, that is, dualities of the theory with itself. The prime examples of auto-dualities is $S$-duality in $\mathcal{N}=2^*$ SYM and Gaiotto's $\mathcal{N}=2$ generalized $S$-dualities \cite{gaiotto2012n}. 
One of the motivation of this paper is to use categorical methods to compute the group  $\mathbb{S}$ of $S$-dualities which generalize the $PSL(2,\Z)$ group for  $\mathcal{N}=2^*$ as well as the results by Gaiotto.
\medskip 

An auto-duality induces a family of
exact functors $\mathsf{d}_{(a)}\colon
\mathfrak{T}_{(a)}\to \mathfrak{T}_{(a)}$, one for each BPS category $\mathfrak{T}_{(a)}$, such that:
\begin{itemize}
\item[\it a)] for all $a\in I$, $\mathsf{d}_{(a)}$ is an\emph{autoequivalence} of the triangle category $\mathfrak{T}_{(a)}$;
\item[\it b)] the $\{\mathsf{d}_{(a)}\}$ satisfy physical consistency conditions in the form of commutative diagrams 
\begin{equation}\label{conscon}
\begin{gathered}\xymatrix{\mathfrak{T}_{(a)}\ar[r]^{\mathsf{d}_{(a)}}\ar[d]_{\mathsf{c}_{(a,b)}} &
\mathfrak{T}_{(a)}\ar[d]^{\mathsf{c}_{(a,b)}}\\
\mathfrak{T}_{(b)}\ar[r]^{\mathsf{d}_{(b)}}&
\mathfrak{T}_{(b)}}\end{gathered}
\end{equation}
\end{itemize}

\begin{definition}\label{defSS2}
1) The group $\mathfrak{S}$ of \textit{generalized auto-dualities} is the group of families $\mathsf{d}_{(a)}$ of autoequivalences satisfying eqn.\eqref{conscon} \emph{modulo} its  subgroup acting trivially on the physical observables. 
2) The group $\mathbb{S}$ of \emph{(generalized) $S$-dualities} is the quotient group
of $\mathfrak{S}$ which acts effectively on the (UV) microscopic local degrees of freedom of the theory.
\end{definition}

\begin{remark}
With our definition of the $S$-duality group, the Weyl group of the flavor group is always part of the duality group $\mathbb{S}$. It action on the free part of the cluster Grothendieck group is the natural one on the weight lattice.
\end{remark}

\begin{example}\label{lllz120x} With this definition, the group $\mathbb{S}$ for $SU(2)$ SQCD with $N_f=4$ is \cite{cecotti2015higher}
 $$
 \mathbb{S}_{SU(2),\;N_f=4}=SL(2,\Z)\rtimes \mathrm{Weyl}(SO(8)).
 $$
 \end{example} 
 
 \begin{remark} We shall see in {\bf Example \ref{ex:su2n4}} that with this definition the $S$-duality group of a class $\mathcal{S}[A_1]$ theory is the tagged mapping class group of its Gaiotto surface, in agreement with the geometric picture in \cite{gaiotto2012n}, see also  \cite{drukker2009loop}.\end{remark}
 
 \subsection{Specializing to $\mathcal{N}=2$ in $4d$}
 
 We specialize the discussion to the case of a $4d$ $\mathcal{N}=2$ theory having the BPS quiver property. Such a theory is associated to a mutation-class of quivers with potential, hence to the three categories $D^b\Gamma$, $\mathfrak{Per}\,\Gamma$, $\mathcal{C}(\Gamma)$, discussed in the previous sections. They are related by the compatibility functors $\mathsf{s}$, $\mathsf{r}$ as in the exact sequence
 \eqref{exactse}.

Applying \textbf{Definition \ref{defSS2}} to the present set-up, we are lead to consider the  diagram of triangle functors 
$$
\begin{gathered}\xymatrix{0 \ar[r] & D^b\Gamma \ar[d]_{d_D}\ar[r]^{\mathsf{s}} & \mathfrak{Per}\,\Gamma\ar[d]_{d_\mathfrak{P}}\ar[r]^{\mathsf{r}} & \mathcal{C}(\Gamma)\ar[d]_{d_\mathcal{C}}\ar[r] &0\\
0 \ar[r] & D^b\Gamma \ar[r]^{\mathsf{s}} & \mathfrak{Per}\,\Gamma\ar[r]^{\mathsf{r}} & \mathcal{C}(\Gamma)\ar[r] &0}\end{gathered}
$$ 
having exact rows and commuting squares,
where
$$d_D\in \mathrm{Aut}\,D^b\Gamma,\quad d_\mathfrak{P}\in\mathrm{Aut}\,\mathfrak{Per}\,\Gamma,\quad
d_\mathcal{C}\in \mathrm{Aut}\,\mathcal{C}(\Gamma).$$
The group $\mathfrak{S}$ is the group of such triples $(d_D,d_\mathfrak{P},d_\mathcal{C})$ modulo the subgroup which acts trivially on the observables. The $S$-duality group $\mathbb{S}$ is the image of $\mathfrak{S}$ under the homomorphism
\begin{equation}\label{whatSs}
r\colon \mathfrak{S}\to \mathrm{Aut}\,\mathcal{C}(\Gamma)/\mathrm{Aut}\,\mathcal{C}(\Gamma)_\text{trivial},\qquad
(d_D,d_\mathfrak{P},d_\mathcal{C})\mapsto d_\mathcal{C}.
\end{equation}

\subsubsection{The trivial subgroup $(\mathrm{Aut}\,D^b\Gamma)_0$}

We start by characterizing the subgroup $(\mathrm{Aut}\,D^b\Gamma)_0\subset \mathrm{Aut}\,D^b\Gamma$ of `trivial' auto-equivalences, i.e.\! the ones which leave the physical observables invariant. Since the Grothendieck group is identified with the IR charge lattice $\Lambda$, and charge is an observable, $(\mathrm{Aut}\,D^b\Gamma)_0$ is a subgroup of the kernel $\mathrm{Aut}\,D^b\Gamma\to \mathrm{Aut}\,K_0(D^b\Gamma)$.
Next all $\varrho\in(\mathrm{Aut}\,D^b\Gamma)_0$ should leave invariant the stability condition, that is the slicing $\mathcal{P}(\phi)$, and hence the canonical heart $\mathsf{mod}\,J(Q,W)$ of $D^b\Gamma$. Since $\varrho$ acts trivially on the Grothendieck group, it should fix all simples $S_i$.
Hence the projection $\mathrm{Aut}\,D^b\Gamma\to \mathrm{Aut}\,D^b\Gamma/(\mathrm{Aut}\,D^b\Gamma)_0$ factors through the quotient group
$$
\mathrm{Autph}\, D^b\Gamma :=\mathrm{Aut}\, D^b\Gamma\bigg/
\Big\{\text{autoequivalences preserving the simples $S_i$ (element-wise)}
\Big\}.
$$
An equivalence in the kernel of the projection $\mathrm{Aut}\,D^b\Gamma\to \mathrm{Autph}\,D^b\Gamma$ preserves $(Q,W)$, the central charge $Z$, and the Grothendieck class $\lambda$. Hence it maps stable objects of charge $\lambda$ into stable objects of charge $\lambda$. Comparing with eqn.\eqref{whatMgamma}, we see that the net effect of an autoequivalence in the kernel is to produce an automorphism of projective varieties $M_\lambda\to M_{\lambda}$ for each $\lambda$. Since the BPS states are the \textsc{susy} vacua of the $1d$ sigma-model with target space $M_\lambda$, this is just a change of variables in the SQM path integral, which leave invariant all physical observables\footnote{\ 
The simplest example of such a negligible equivalence is the case of pure $SU(2)$ whose quiver is the Kronecker quiver, $\mathsf{Kr}= \bullet\rightrightarrows \bullet$.
The stable representations associated to the $W$ boson are the simples in the homogeneous tube which form a $\mathbb{P}^1$ family (i.e.\! $M_\text{W boson}\equiv\mathbb{P}^1$) since the $W$ boson belongs to a \emph{vector} superfield. Then a negligible auto-equivalence is just a projective automorphism of $\mathbb{P}^1$.}.
Since the auto-duality groups
are defined modulo transformations acting trivially on the observables,
$\mathrm{Auteq}\, D^b$ is the proper auto-duality group $\mathcal{S}_\text{IR}$ at the level of the BPS category $D^b\Gamma$. 

The automorphisms of the quiver
extend to automorphisms of $D^b\Gamma$; let $\mathrm{Aut}(Q)$ be the group of quiver automorphisms modulo the ones which fix the nodes. Clearly,
$$
\mathrm{Autph}\,D^b\Gamma= 
\mathrm{Auteq}\,D^b\Gamma \rtimes \mathrm{Aut}(Q),
$$
where
$$
\mathrm{Auteq}\, D^b\Gamma :=\mathrm{Aut}\, D^b\Gamma\bigg/
\Big\{\text{autoequivalences preserving the simples $S_i$ (as a set)}
\Big\}.
$$

\subsubsection{The duality groups $\mathfrak{S}$ and $\mathbb{S}$}

With the notation of section \ref{sec:cluster}, Bridgeland in \cite{bridgeland2015quadratic} and Goncharov in \cite{goncharov2016ideal} showed that the following sequence
\begin{equation}
0 \to \mathrm{Sph}\,D^b\Gamma\to \mathrm{Auteq}\, D^b\Gamma \to \mathrm{Aut}_Q(CEG)\to 0
\label{eq:seqsdual2}
\end{equation}
is exact. 
Here $CEG$ stands for the cluster exchange graph (cfr.\! \S.\,\ref{CEG}): the clusters of the cluster algebra $C_\Gamma$ are the vertices of the CEG and the edges are single mutations connecting two seeds; $\mathrm{Aut}_Q(CEG)$ is the graph automorphism group that sends the quiver to itself up to relabeling of the vertices. By construction this graph is connected.

\begin{theorem}[Goncharov \cite{goncharov2016ideal}, see also
\cite{keller2011cluster}] One has 
$$\mathrm{Aut}_Q(CEG)\subset \mathrm{Aut}\,\mathcal{C}(\Gamma),$$
i.e. the graph automorphisms (see {\rm \cite{assem2012cluster}}) are a subgroup of the autoequivalences of the cluster category.
\end{theorem}

Note that $\mathrm{Auteq}\,D^b\Gamma\equiv \mathrm{Auteq}\,\mathfrak{Per}\,\Gamma$,
the quotient group of $\mathrm{Aut}\,\mathfrak{Per}\,\Gamma$ by the subgroup fixing the $\Gamma_i$ (as a set).
Indeed, all autoequivalences of $\mathfrak{Per}\,\Gamma$ preserve
 the subcategory $D^b\Gamma$ and hence restrict to autoequivalences of the bounded category; an autoequivalence $\varrho\in\mathrm{Aut}\,\mathfrak{Per}\,\Gamma$ which does not preserve the $\Gamma_i$'s restricts to an element $\bar\varrho\in\mathrm{Aut}\,D^b\Gamma$ which does not preserve the $S_i$'s. Hence the restriction homomorphism
 $$
 \mathrm{Auteq}\,\mathfrak{Per}\,\Gamma\to \mathrm{Auteq}\,D^b\Gamma,
 $$
 is injective. On the other hand, from eqn.\eqref{eq:seqsdual2} we see that all autoequivalences in $\mathrm{Auteq}\,D^b\Gamma$ extend to autoequivalences in $\mathrm{Auteq}\,\mathfrak{Per}\,\Gamma$:  indeed, the objects which are spherical in the subcategory $D^b\Gamma$ remain spherical and 3-CY in the larger category $\mathfrak{Per}\,\Gamma$ (cfr.\!\! eqn.\eqref{uuuu}), so the auto-equivalences is $\mathrm{Sph}\,D^b\Gamma$ extend to $\mathfrak{Per}\,\Gamma$; the autoequivalences in $\mathrm{Aut}_Q(CEG)$ are induced by quiver mutations, and hence induce auto-equivalences of $\mathfrak{Per}\,\Gamma$.  

Comparing with our discussion around eqn.\eqref{whatSs} we conclude:

\begin{corollary} For a $4d$ $\mathcal{N}=2$ theory with the BPS quiver property
\begin{align}
\mathfrak{S}&\cong\mathrm{Auteq}\,D^b\Gamma \rtimes \mathrm{Aut}(Q),\\
\mathbb{S}&\cong\mathrm{Aut}_Q(CEG)\rtimes \mathrm{Aut}(Q).
\end{align}
\end{corollary}

\subsubsection{Example: the group $\mathbb{S}$ for $SU(2)$ $\mathcal{N}=2^*$}\label{ex2star}

The mutation class of $SU(2)$ $\mathcal{N}=2^*$ contains a single quiver, the Markoff one
\begin{displaymath}
   Q_\text{Mar}\equiv\begin{gathered} \xymatrix{\bullet_1 \ar@2{->}[r] & \bullet_2 \ar@2{->}[ld] \\
              \bullet_3 \ar@2{->}[u] &  }\end{gathered}
\end{displaymath}
which is the quiver associated to the once punctured torus \cite{fomin2008cluster,cecotti2011classification}. Clearly $\mathrm{Aut}(Q_\text{Mar})\cong\Z_3$, while all mutations leave $Q_\text{Mar}$ invariant up to a permutation of the nodes. Consider the covering graph
$\widetilde{CEG}$ of $CEG$ where we do not mod out the permutations of the nodes. Then $\widetilde{CEG}$ is the trivalent tree whose edges are decorated by $\{1,2,3\}$, the number attached to an edge corresponding to the nodes which gets mutated along that edge. One can check that there are no identifications between the nodes of this tree.

One may compare this ($\{1,2,3\}$-decorated) trivalent tree  with the ($\{1,2,3\}$-decorated) standard triangulation of the upper half-plane $\mathbb{H}$
given by the reflections of the  geodesic triangle of vertices $0,1,\infty$ (see ref.\cite{fock2009cluster}). One labels
the nodes of a triangle of the standard triangulation by elements of $\{1,2,3\}$, and then extends (uniquely) the numeration to all other vertices so that the vertices of each triangle get different labels.
The sides of a triangle are numbered as their opposite vertex.
The dual of this decorated triangulation is our decorated trivalent graph $\widetilde{CEG}$, see figure \ref{hhhxz}.
The arithmetic subgroup of the
hyperbolic isometry group,
$PGL(2,\Z)\subset PGL(2,\R)$ preserves the standard triangulation of $\mathbb{H}$ while permuting the decorations $\{1,2,3\}$. Since permutations are valid $S$-dualities, we get
$$
\mathbb{S}\cong PGL(2,\Z)\cong PSL(2,\Z)\rtimes \Z_2
$$
where the extra $\Z_2$ may be identified with the Weyl group of the flavor $SU(2)$.
Thus we recover as $S$-duality group in the usual sense ($\equiv$ the kernel of 
$\mathbb{S}\to \text{Weyl}(F)$)
the modular group $PSL(2,\Z)$ \cite{seiberg1994monopoles}. 
In the case of $SU(2)$ $\mathcal{N}=2^*$ we have
$$K_0(\mathcal{C}_\text{Mar})\equiv\mathrm{coker}\,B\cong\Z_2\oplus\Z_2\oplus\Z,$$
as expected for a quark in the adjoint representation (since $\pi_1(G_\text{eff})=\Z_2$), with the free part the weight lattice of $SU(2)_\text{flav}$. Hence the flavor Weyl group acts on $K_0(\mathcal{C}_\text{Mar})$ as $-1$, that is, as the cluster auto-equivalence $[1]$. Notice that the cluster category is $2$-periodic, as expected for a UV SCFT with integral dimensions $\Delta$.
 
 \begin{figure}
\centering
\includegraphics[width=0.8\textwidth]{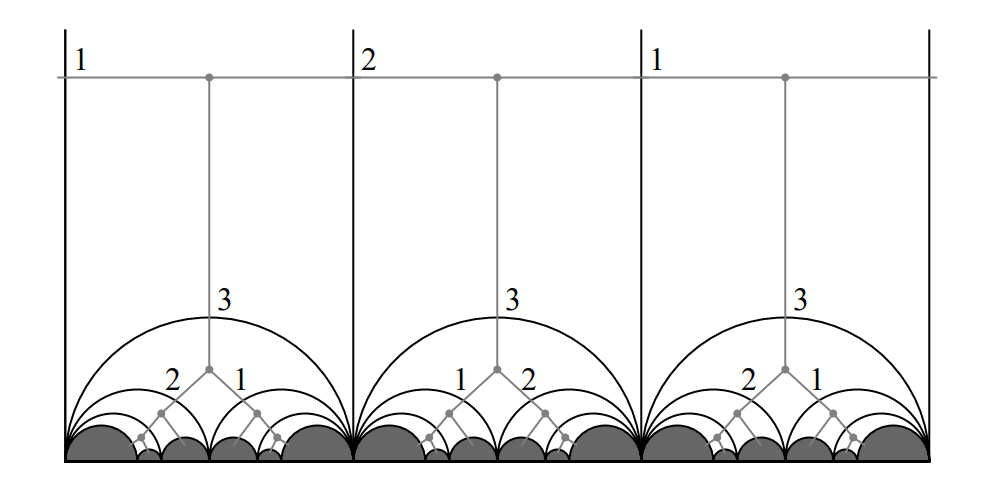}
\caption{\label{hhhxz}The modular triangulation of the upper half plane and its dual graph $\widetilde{CEG}$. The picture is reproduced from \cite{fock2009cluster}. }
\end{figure}

\subsection{Relation to duality walls and $3d$ mirrors}
\label{sec:dualitywalls}
The UV $S$-duality group $\mathbb{S}$ has a clear interpretation: it is the usual $S$-duality group of the $\mathcal{N}=2$ theory (twisted by the flavor Weyl group).
What about its IR counterpart $\mathfrak{S}$?

For Argyres-Douglas models we can put forward a precise physical interpretation based on the findings of \cite{braidmirror}. Similar statements should hold in general. 
\medskip 

Given an element of the $S$-duality group, $\sigma\in\mathbb{S}$ we may construct a half-BPS duality wall in the $4d$ theory \cite{terash,dimo1,dimo2}: just take the theory for $x_3<0$ to be the image through $\sigma$ of the theory for $x_3>0$ and adjust the field profiles along the hyperplane $x_3=0$ in such a way that the resulting Janus configuration is $\tfrac{1}{2}$-BPS.  It is a domain wall interpolating between two dual $\mathcal{N}=2$ theories in complementary half-spaces. On the wall live suitable $3d$ degrees of freedom interacting with the bulk $4d$ fields on both sides \cite{terash,dimo1,dimo2}. 
In this construction we may use a UV duality as well as an IR one \cite{dimo2}. Hence we expect to
get duality walls for all elements of $\mathfrak{S}$.
An element $\mathfrak{s}\in\mathfrak{S}$ acts non-trivially
on the central charge $Z$ so, in general, as we go from $x_3=-\infty$ to $x_3=+\infty$ we induce a non-trivial flow of the central charge $Z$ in the space of stability functions. If $\lim_{x_3\to\pm\infty}Z$ is such that all the bulk degrees of freedom get an infinite mass and decouple, we remain with a pure $3d$ $\mathcal{N}=2$ theory on the wall. Of course this may happen only for  special choices of $\mathfrak{s}$.
Thus we may use (suitable) $4d$ dualities to engineer $3d$ $\mathcal{N}=2$ QFTs.  

The engineering of $3d$ $\mathcal{N}=2$ theories as a domain wall in a $4d$ $\mathcal{N}=2$ QFT, by central-charge flow in the normal direction, is precisely the set-up of
ref.\cite{braidmirror}. 
In that paper one started from a $4d$ Argyres-Douglas of type $\mathfrak{g}\in ADE$. The $Z$-flow along the $x_3$-axis was such that 
asymptotic behaviors as $x_3\to-\infty$ and $x_3\to+\infty$ were related in the UV by the action of the quantum half-monodromy $\mathbb{K}$, that is, in the categorical language by the shift $[1]\in\mathbb{S}$. Two choices of IR duality elements, $\mathfrak{s}, \mathfrak{s}^\prime\in\mathfrak{S}$, which produce the half-monodromy in the UV,  differ by an element of the spherical twist group (cfr.\! eqn.\eqref{eq:seqsdual2})
$$
\mathfrak{s}^\prime\mathfrak{s}^{-1}\in \mathrm{Sph}\,D^b.
$$
The arguments at the end of \S.\,\ref{braidactions} imply that for Argyres-Douglas of type $\mathfrak{g}$ the group $\mathrm{Sph}\,D^b$ is isomorphic to the Artin braid group of type $\mathfrak{g}$, $\mathcal{B}_\mathfrak{g}$.

As the title of ref.\cite{braidmirror} implies, the explicit engineering of a $3d$ $\mathcal{N}=2$ theory along those lines requires a specification of a braid, i.e.\! of an element of $\mathcal{B}_\mathfrak{g}$. More precisely, in \S.5.3.2 of ref.\cite{braidmirror} is given an explicit map (for $\mathfrak{g}=A_r$)
$$
\big(\text{a braid in }\mathcal{B}_\mathfrak{g}\big)\longleftrightarrow
\big(\text{a $3d$ $\mathcal{N}=2$ Lagrangian}\big). 
$$
So the Lagrangian description/$Z$-flow engineering of the $3d$ theories are in one-to-one correspondence with the $\mathfrak{s}\in\mathfrak{S}$ such that $r(\mathfrak{s})=[1]$. It is natural to think of the $3d$ Lagrangian theory associated to $\mathfrak{s}\in\mathfrak{S}$ as the duality wall associated to the IR duality $\mathfrak{s}$.
 Distinct $\mathfrak{s}$ lead to $3d$ theories which  superficially look quite different. However,  in this context, \textit{$3d$ mirror symmetry} is precisely the statement that two theories defined by different IR dualities $\mathfrak{s},\mathfrak{s}^\prime\in\mathfrak{S}$ which induce the same UV duality, $r(\mathfrak{s}^\prime)=r(\mathfrak{s})$ produce equivalent $3d$ QFTs. From this viewpoint $3d$ mirror symmetry is a bit tautological, since the condition $r(\mathfrak{s}^\prime)=r(\mathfrak{s})$ just says that the two $3d$ theories have the same description in terms of $4d$
 microscopic degrees of freedom.

\subsection{$S$-duality for Argyres-Douglas and $SU(2)$ gauge theories}

When $(Q,W)$ is in the mutation-class of an $ADE$ Dynkin graph (corresponding to an Argyres-Douglas model \cite{alim2014mathcal, alim2013bps}) or of an $\widehat{A}\widehat{D}\widehat{E}$ acyclic affine quiver (corresponding to $SU(2)$ SYM coupled to matter such that the YM coupling is asymptotically-free \cite{alim2013bps}) to get $\mathbb{S}$ we can equivalently study the automorphism of the transjective component of the AR quiver associated to the cluster category $\mathcal{C}(\Gamma)$: the inclusion above is due to the fact that we only consider the transjective component:
\begin{theorem}[See \cite{assem2012cluster}]
Let $C$ be an acyclic cluster algebra and $\Gamma_\mathrm{tr}$ the transjective component of
the Auslander-Reiten quiver of the associated cluster category $\mathcal{C}(\Gamma)$. Then $\mathrm{Aut}^+C$ is the quotient
of the group $\mathrm{Aut}\,\Gamma_\mathrm{tr}$ of the quiver automorphisms of $\Gamma_\mathrm{tr}$, modulo the stabilizer $\mathrm{Stab}(\Gamma_\mathrm{tr})_0$ of the points of this component. Moreover, if $\Gamma_\mathrm{tr} \cong \Z\Delta,$ where $\Delta$ is a tree or of type $\hat A$ then
$$\mathrm{Aut}\, C = \mathrm{Aut}^+C \rtimes \Z_2$$
 and this semidirect product is not direct.
\end{theorem}
In order to understand why this is the relevant component, we first recall that the Auslander-Reiten quiver of a cluster-tilted algebra always has a unique component containing local slices, which coincides with the whole
Auslander-Reiten quiver whenever the cluster-tilted algebra is representation-finite. This component
is called the \emph{transjective component} and an indecomposable module lying in it is called a transjective
module. With this terminology, the main result is:
\begin{theorem}[See \cite{assem2013modules}]
Let $C$ be a cluster-tilted algebra and $M, N$ be indecomposable transjective $C$-modules.
Then $M$ is isomorphic to $N$ if and only if $M$ and $N$ have the same dimension vector.
\end{theorem}
Therefore, since the dimension vector is the physical charge, we focus our attention to this class of autoequivalences. The classification results are summarized in table \ref{hhhaz}
where 
$$\begin{gathered}
H_{p,q}:=\vev{r,s|r^p=s^q,\ sr=rs}\\
G=\vev{\tau,\sigma,\rho_1,\rho_n|\rho_1^2=\rho_n^2=1, \tau \rho_1=\rho_1\tau,\tau\rho_n=\rho_n\tau, \tau\sigma=\sigma\tau, \sigma^2=\tau^{n-3},
\rho_1\sigma=\sigma \rho_n, \sigma\rho_1 = \rho_n\sigma}\end{gathered}$$

\begin{table}
\begin{center}
\begin{tabular}{c|c||c|c}
Q & $\mathrm{Aut}_Q(CEG)$ & Q & $\mathrm{Aut}_Q(CEG)$\\
\hline 
$A_{n>1}$ & $\Z_{n+3}$ &
$D_4$ & $\Z_{4}\times S_3$\\
$D_{n>4}$ & $\Z_{n}\times \Z_2$ &
$E_6$ & $\Z_{14}$\\
$E_7$ & $\Z_{10}$&
$E_8$ & $\Z_{16}$\\
$\hat A_{p,q}$ & 
$H_{p,q}$&
$\hat A_{p,p>1,1}$ & $H_{p,p}\rtimes \Z_2$\\
$\hat D_4$ & $\Z \times S_4$&
$\hat A_{1,1}$ & $\Z$\\
$\hat D_{n>4}$ & $G$&$\hat E_6$ & $\Z \times S_3$\\
$\hat E_7$ & $\Z \times \Z_2$ & $\hat E_8$ & $\Z $\\
\end{tabular}
\end{center}
\caption{\label{hhhaz} $S$-duality groups for $\mathcal{N}=2$ theory with an acyclic quiver.}
\end{table}

\begin{example}[$SU(2)$ with $N_f\leq3$]\label{kzanncv} $SU(2)$ SQCD with $N_f=0,1,2,3$ correspond, respectively, to the following four affine $\mathcal{N}=2$ theories \cite{cecotti2011classification}
$$
\hat A_{1,1},\quad \hat A_{2,1},\quad \hat A_{2,2},\quad \hat D_4. 
$$
A part for the flavor Weyl group $\mathrm{Weyl}(\mathfrak{spin}(2N_f))$ (cfr.\! \textbf{Example \ref{Nf3}}) we get a duality group $\Z$ generated by the shift $[1]$.
As discussed around eqn.\eqref{nnna129c}, this is equivalent to the shift of the Yang-Mills angle $\theta$
$$
\theta\to\theta-4\pi +N_f\pi.
$$
The case $N_f=0$ is special; physically one expects that the shift of $\theta$ by $-2\pi$ should also be a valid $S$-duality. This shift should correspond to an auto-equivalence $\xi$ of the $N_f=0$ cluster category with $\xi^2=\tau$. Indeed, this is what one obtains from the automorphism of the transjective component see figure \ref{rrrrv}.
Alternatively, we may see the cluster category of pure $SU(2)$ as the category of coherent sheaves on $\mathbb{P}^1$ endowed with extra \emph{odd}
morphisms \cite{barot}. In this language $\tau$ acts as the tensor product with the canonical bundle $\tau\colon\ca\mapsto\ca\otimes \mathcal{K}$ (cfr.\! eqn.\eqref{jzxce}). Let $\mathcal{L}$ be the unique \emph{spin structure} on $\mathbb{P}^1$; we have the obvious auto-equivalence $\xi\colon\ca\mapsto\ca\otimes \mathcal{L}$.
From $\mathcal{L}^2=\mathcal{K}$ we see that $\xi^2=\tau$.
\end{example}

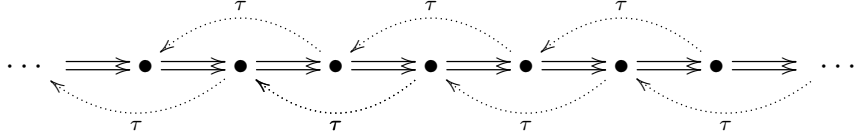
\begin{figure}
$$
\xymatrix{\cdots\text{\ }\ar@<0.3ex>[r]\ar@<-0.3ex>[r]&\bullet\ar@<0.3ex>[r]\ar@<-0.3ex>[r]&\ar@{..>}@/^1.5pc/[ll]^\tau\bullet\ar@<0.3ex>[r]\ar@<-0.3ex>[r]&\bullet\ar@{..>}@/_1.5pc/[ll]_\tau\ar@<0.3ex>[r]\ar@<-0.3ex>[r]&\ar@{..>}@/^1.5pc/[ll]^\tau\ar@{..>}@/^1.5pc/[ll]^\tau\bullet\ar@<0.3ex>[r]\ar@<-0.3ex>[r]&\ar@{..>}@/_1.5pc/[ll]_\tau\bullet\ar@<0.3ex>[r]\ar@<-0.3ex>[r]&\ar@{..>}@/^1.5pc/[ll]^\tau\bullet\ar@<0.3ex>[r]\ar@<-0.3ex>[r]&
\ar@{..>}@/_1.5pc/[ll]_\tau\bullet\ar@<0.3ex>[r]\ar@<-0.3ex>[r]&\ar@{..>}@/^1.5pc/[ll]^\tau\text{\ }\cdots
}
$$
\caption{The translation quiver $\Z\hat A_{1,1}$ ($\equiv$ the AR quiver of the transjective component of the cluster category for pure $SU(2)$). Dotted arrows stands for the action of the AR translation $\tau$. Clearly $\tau$ is the translation to the left by 2 nodes. The auto-equivalence $\xi$ is translation to the left by 1 node: $\xi^2=\tau$.}
\label{rrrrv}
\end{figure}

\section{Computer algorithm to determine the $S$-duality group}\label{duacccom}
\label{sec:algo}
The identification of the $S$-duality group with $\mathrm{Aut}_Q(CEG)$ yield a combinatoric characterization of $S$-dualities which leads to an algorithm to search $S$-dualities for an arbitrary $\mathcal{N}=2$ model having a BPS quiver. This algorithm is similar in spirit to the \emph{mutation algorithm} to find the BPS spectrum \cite{alim2013bps} but in a sense more efficient. The algorithm may be easily implemented on a computer; if the ranks of the gauge and flavor groups are not too big (say $< 10$), running the procedure on a laptop typically produces the generators of the duality group in a matter of minutes.

\subsection{The algorithm}

The group $\mathrm{Aut}_Q(CEG)$ 
may be defined in terms of the transformations under quiver mutations of the $d$-vectors which specify the denominators of the generic cluster variables \cite{dupont2011generic}. The actions of the elementary quiver mutation at the $k$--th node, $\mu_k$, on the exchange matrix $B$ and the $d$-vector $d_i$ are
\begin{align}
\mu_k(B)_{ij}&=\begin{cases}
- B_{ij}, & i=k \ \text{or }j=k\\
B_{ij}+\max[-B_{ik},0]\,B_{kj}+B_{ik}\,\max[B_{kj},0]\phantom{-----} & \text{otherwise.}
\end{cases}\\
\mu_k(d)_{l}&=\left\{
\begin{array}{lr}
d_{l}, & l\neq k\\
-d_{k}+\max\!\Big[\sum_i\max\!\big[B_{ik},0\big]d_{i},\sum_i\max\!\big[-B_{ik},0\big]d_i\Big] & l=k
\end{array}
\right.
\label{eq:dvecs}
\end{align}

A quiver mutation $\mu=\mu_{k_s}\mu_{k_{s-1}}\cdots\mu_{k_1}$ is the composition  of a finite sequence of elementary quiver mutations $\mu_{k_1}, \mu_{k_2},\cdots,\mu_{k_s}$. We write $\mathsf{Mut}$ for the set of all quiver mutations.
$\mathrm{Aut}_Q(CEG)$
is the group of quiver mutations which  
 leave invariant the quiver $Q$ up to a permutation $\pi$ of its nodes, 
modulo the ones which leave the $d$-vector invariant up to $\pi$: 
\begin{equation}\label{lllzx}
\mathrm{Aut}_Q(CEG)=\frac{
\Big\{\mu\in\mathsf{Mut}\;\big|\;\exists\,\pi\in {S}_n\colon
\mu(B)_{i,j}=B_{\pi(i),\pi(j)}\Big\}}{\Big\{\mu\in\mathsf{Mut}\;\big|\;\exists\,\pi\in {S}_n\colon
\mu(B)_{i,j}=B_{\pi(i),\pi(j)}\ \text{and }
\mu(d)_{i}=d_{\pi(i)}
\Big\}},
\end{equation}
while $\mathbb{S}=\mathrm{Aut}_Q(CEG)\rtimes \mathrm{Aut}(Q)$.
\begin{example}[$A_2$ cluster automorphisms]
\label{ex:a2c}
Consider the quiver $\bullet_1\to \bullet_2$. The CEG is the pentagon in figure \ref{ggaq}: every vertex is associated to a quiver of the form $\bullet_1\to \bullet_2$ or $\bullet_2\to \bullet_1$. 
\begin{figure}
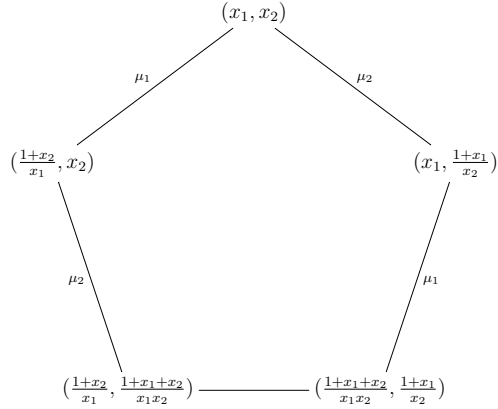

\centering
\resizebox{0.40\textwidth}{!}{\xygraph{
!{<0cm,0cm>;<0.5cm,0cm>:<0cm,0.5cm>::}
!{(0,10) }*+{{(x_1,x_2)}}="b1"
!{(8,4) }*+{{(x_1,\frac{1+x_1}{x_2})}{}}="b3"
!{(-8,4) }*+{{(\frac{1+x_2}{x_1},x_2)}}="b2"
!{(-5,-5) }*+{{(\frac{1+x_2}{x_1},\frac{1+x_1+x_2}{x_1x_2})}}="b4"
!{(5,-5) }*+{{(\frac{1+x_1+x_2}{x_1x_2},\frac{1+x_1}{x_2})}}="b5"
"b1" - "b2"_{\mu_1} "b2" - "b4"_{\mu_2} "b4" - "b5" "b5" - "b3"_{\mu_1} "b3" - "b1" _{\mu_2}
}}
\caption{\label{ggaq}The CEG of the $A_2$ Argyres-Douglas theory.}
\end{figure}
Thus, in this case every sequence of mutations gives rise to a cluster automorphism. For example, consider $\mu_1$: the quiver nodes get permuted under $\pi=(1 \ 2)$. We explicitly check -- for example using Keller applet\footnote{See \url{https://webusers.imj-prg.fr/~bernhard.keller/quivermutation/}.} -- that
$$(\mu_\text{source})^5=\mu_1\mu_2\mu_1\mu_2\mu_1=1$$
since $(\mu_\text{source})^5$ leaves the $d$-vectors invariant. From figure \ref{ggaq} one sees that $\Z_5$ is indeed the full automorphism group of the CEG of $A_2$. This result is coherent with
the analysis leading to table \ref{hhhaz}, as well as with the 
 tagged mapping class group of the associated Gaiotto surface, see \textbf{Example \ref{ex:a2}}.
\end{example}
 The explicit expression \eqref{lllzx} of the $S$-duality group is the basis of a computer search for $S$-dualities.
Schematically: let the computer generate a finite sequence of nodes of $Q$, $k_1,\cdots, k_s$, then construct the corresponding mutation $\mu_{k_s}\mu_{k_{s-1}}\cdots\mu_{k_1}=\mu$, and check whether it leaves the exchange matrix $B$ invariant up to a permutation $\pi$. If the answer is \textit{yes,} let the machine check whether
 $\mu(d)_i\neq d_{\pi(i)}$. If the answer is again \textit{yes}
the computer has discovered a non-trivial $S$-duality and prints it. Then the computer generates another  sequence and go cyclically through  the same steps again and again.
After running the procedure for some time $t$, we get a print-out with a list $\mathcal{L}_t$ of non-trivial $S$-dualities of our $\mathcal{N}=2$ theory.
A \textsc{Mathematica} Code performing this routine is presented in Appendix \ref{app:mathemut}.
\medskip 

If the $S$-duality group is finite (and not too huge) $\mathcal{L}_t$ will contain the full list of $S$-dualities. However, the most interesting $S$-duality groups are \emph{infinite}, and the computer cannot find all its elements in finite time.
This is not a fundamental problem for the automatic computation of the $S$-duality group. Indeed, the $S$-duality groups, while often infinite,
are expected to be \emph{finitely generated,} and in fact finitely presented. If this is the case, we need only that the finite list $\mathcal{L}_t$ produced by the computer contains a complete set of generators of $\mathbb{S}$.
Taking various products of these generators, and checking which products act trivially on the $d$-vectors, we may find the finitely many relations. The method works better if we have some physical hint on what the generators and relations may be.

Of course, the duality group obtained from the computer search is \emph{a priori} only a subgroup of the actual $\mathbb{S}$ because there is always the possibility of further generators of the group which are outside our range of search. However, pragmatically, running the procedure for enough time, the group one gets is the full one at a high confidence level. 

\subsection{Sample determinations of $S$-duality groups}

We present a sample of the results obtained 
by running our \textsc{Mathematica} Code.  

\begin{example}[$SU(2)$ $\mathcal{N}=2^*$ again] The $CEG$ automorphism group for this model was already described in \S.\,\ref{ex2star}.
Recall that $PSL(2,\mathbb Z)$ is the quotient of the braid group over three strands, $\mathcal{B}_3$ by its center $\boldsymbol{Z}(\mathcal{B}_3)$
$$PSL(2,\mathbb Z)\cong \mathcal{B}_3\big/\boldsymbol{Z}(\mathcal{B}_3).
$$ 
Running our algorithm for a short time returns a list of dualities which in particular contains the two standard generators of the braid group $\sigma_1,\sigma_2\in\mathcal B_3$, which correspond to  the following sequences of elementary quiver mutations:
\begin{equation}\label{lllazq6}
\sigma_1:=\mu_1\mu_2,\qquad\text{and}\qquad \sigma_2:=\mu_1\mu_3, \text{ with permutation }\pi=(1\,3\,2).
\end{equation}
One easily checks the braid relation
$$\sigma_1\sigma_2\sigma_1=\sigma_2\sigma_1\sigma_2\quad \text{up to permutation,}
$$
as well as that the generator of the center $\boldsymbol{Z}(\mathcal{B}_3)$, $(\sigma_2\sigma_1)^3$,
acts trivially on the cluster category: indeed, it sends the initial dimension vector $\vec d=-Id_{3 \times 3}$ to itself. From eqn.\eqref{lllazq6} we conclude that the two $S$-dualities $\sigma_1$, $\sigma_2$ generate a $PSL(2,\Z)$ duality (sub)group. 
In facts, $\mathbb{S}/PSL(2,\Z)\cong\Z_2$ where the class of the non-trivial $\Z_2$ element may be
represented (say) by $\mu_1$. Indeed the map
$$
\mathbb{S}\to \Z_2\equiv\text{Weyl}(F_\text{flav})
$$ 
send the mutation $\mu$ to $(-1)^{\ell(\mu)}$, where the \emph{length} $\ell(\mu)$ of
$\mu\equiv\mu_{k_s}\mu_{k_{s-1}}\cdots\mu_{k_1}$ is $s$ (length is well defined mod 2).
\end{example}

\begin{example}[$SU(2)$ with $N_f=4$] \label{ex:su2n4}We use the quiver in figure \ref{fig:su2n4} where for future reference we also draw the corresponding ideal triangulation of the sphere with 4 punctures \cite{cecotti2011classification}.
\begin{figure}
\centering
\begin{tabular}{cc}
\resizebox{0.42\textwidth}{!}{\xygraph{
!{<0cm,0cm>;<1.5cm,0cm>:<0cm,1.2cm>::}
!{(0,0) }*+{\bullet}="a"
!{(0,-2.5) }*+{\bullet}="b"
!{(2.5,0) }*+{\bullet}="c"
!{(2.5,-2.5)}*+{\bullet}="d"
!{(0,0.5)}*+{}="au"
!{(0,-0.5)}*+{}="ad"
!{(-0.5,0)}*+{}="al"
!{(2.5,-2)}*+{}="du"
!{(2.5,-3)}*+{}="dd"
!{(3,-2.5)}*+{}="dr"
!{(1.25,-1.25) }*+{}*\cir<107pt>{}="center"
"a"-@[blue]"c"^{4} "b"-@[blue]"d"^{2} "b"-@[blue]"c"^{0} 
"c"-@/_3.5cm/@[blue]"b" ^{1}
"c" -@`{"au", "al","ad"}@[blue] "c"^{5}
"b" -@`{"du", "dr","dd"}@[blue] "b"^{3}
}} &
  \ \ \  \xymatrix{0 \ar[ddr]\ar[ddrr]\ar[ddrrr]\ar[ddrrrr] &  & & & \\
&  & & & \\
 & 2 \ar[ddl]& 3\ar[ddll] & 4 \ar[ddlll] & 5 \ar[ddllll]\\
&  & & & \\
             1 \ar@2{->}[uuuu] & & & &}
\\
\text{Surface and triangulation} & \text{Quiver}
\end{tabular}
\caption{The Gaiotto surface $(S,M)$ of the theory $SU(2)$, $N_f=4$ and its associated quiver. }
\label{fig:su2n4}
\end{figure}
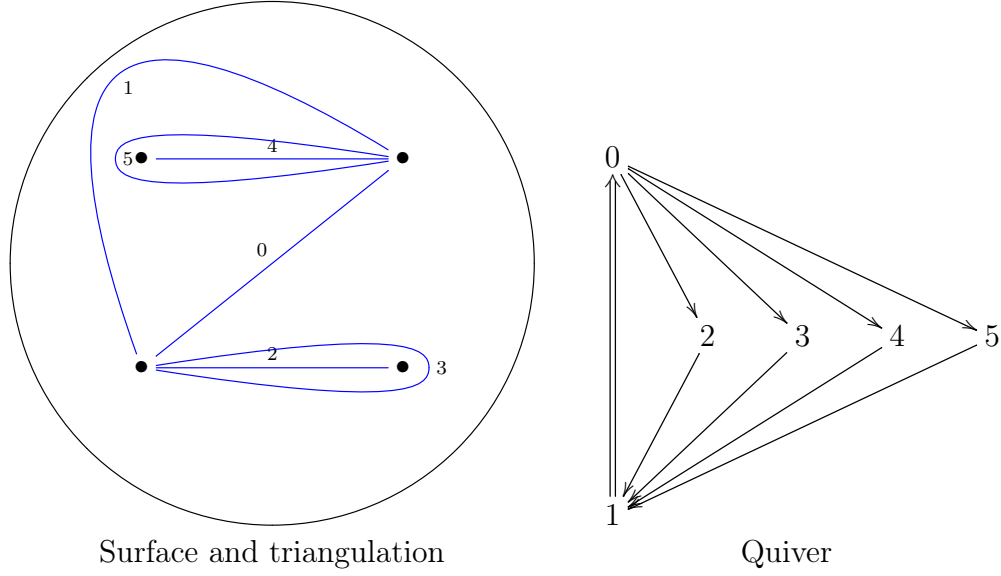
The following two even-length sequences of mutations leave the quiver invariant:
$$
\begin{gathered}
S=\mu_2\mu_3\mu_2\mu_0\mu_2\mu_5\mu_3\mu_0,\\
T= \mu_5\mu_2\mu_0\mu_3\mu_5\mu_3\mu_4\mu_2\mu_4\mu_1\mu_4\mu_2\mu_4\mu_5\mu_1\mu_2.
\end{gathered}
$$
These sequences of mutations satisfy the following relations:
$$S^4=1, \ \ (ST)^6=1, \ \ T \text{ has infinite order.}$$
Moreover, $T$ and $S$ commute with $S^2$ and $(ST)^3$. Write $\Z_2\times \Z_2$ for the subgroup generated by $S^2$ and $(ST)^3$.
Then we have
$$
1 \to \Z_2\times \Z_2\to\langle\, S, T\,\rangle\to PSL(2,\Z)\to 1.
$$
Again 
this shows that
the duality sub-group $\langle\,S,T\,\rangle$ generated by $S$ and $T$ is equal to the
 the mapping class group of the sphere with four punctures (cfr.\! \textbf{Proposition 2.7} of \cite{bookmap}). 
 In fact one has $\mathrm{Aut}_Q(CEG)/\langle\,S,T\,\rangle\cong\Z_2$; geometrically (see next section) the extra $\Z_2$ arises because for class $\mathcal{S}[A_1]$ theories $\mathrm{Aut}_Q(CEG)$
 is the \emph{tagged} mapping class group of the corresponding Gaiotto surface (Bridgeland theorem\cite{bridgeland2015quadratic});
 the extra $\Z_2$ is just the change
 in tagging. This extra $\Z_2$ is also detected by the computer program which turns out dualities of order $12$ and $8$ which are not contained in $\langle\,S,T\,\rangle$ but in its $\Z_2$ extension. 
Taking into account the $S_4$ automorphism of the quiver, we recover $PSL(2,\Z)\ltimes \text{Weyl}(\mathfrak{spin}(8))$
with the proper triality action of the modular group on the flavor weights
\cite{seiberg1994monopoles}. For an alternative discussion of the $S$-duality group of this model as the automorphism group of the corresponding cluster category, see ref.\cite{cecotti2015higher}.
\end{example}

\begin{example}[$E_6$ Minahan-Nemeschanski] This SCFT is the $T_3$ theory, that is, the Gaiotto theory obtained by compactifying the $6d$ $(2,0)$ SCFT of type $A_2$ on a sphere with three maximal punctures \cite{gaiotto2012n}. Since the three-punctured sphere is rigid, geometrically we expect a finite $S$-duality group. The homological methods of \cite{caorsi2016homological} confirm this expectation.
The computer search produced a list of group elements of order
$2,3,4,5,6,8,9,10,12$ and $18$. Since, with our definition, the $S$-duality group should contain the 
Weyl group of $E_6$, we may compare this list with the list of orders of elements of $\mathrm{Weyl}(E_6)$,
$$\big\{2,3,4,5,6,8,9,10,12\big\}.$$ We see that the two lists coincide, except for $18$.
Thus the $S$-duality group is slightly larger than the Weyl group, possibly just $\mathrm{Weyl}(E_6)\rtimes\Z_2$, where $\Z_2$ is the automorphism of the Dynkin diagram.
Notice that this is the largest group which may act on the free part of the cluster Grothendieck group
(since it should act by isometries of the Tits form).
\end{example}

\begin{example}[Generic $T_\mathfrak{g}$ theories]
By the same argument as in the 
previous \textbf{Example}, we expect the $S$-duality group to be finite for all $T_\mathfrak{g}$ ($\mathfrak{g}\in ADE$) theories. 
We performed a few sample computer searches getting agreement with the expectation.
\end{example}

\begin{example}[$E_7$ Minahan-Nemeschanski]
The computer search for this example produced a list of group elements of order
$ 2, 3, 4, 5, 6, 7, 8, 9, 10, 12, 14, 15, 18$ and $30$. Since our $S$-duality group contains the 
Weyl group of the flavor $E_7$, we  compare this list with the list of orders of elements of $\mathrm{Weyl}(E_7)$,
$\{2, 3, 4, 5, 6, 7, 8, 9, 10, 12, 14, 15, 18, 30\}$. We see that the two lists coincide. It is reasonable to believe that the full S-duality group coincides with $\mathrm{Weyl}(E_7)$. This is also the largest group preserving the flavor Tits form. 
\end{example}

\subsection{Asymptotic-free examples}

As an appetizer, let us consider a 
$\mathcal{N}=2$ gauge theory 
with a gauge group of the form $SU(2)^k$ coupled to (half-)hypermultiplets in some representation of the gauge group so that all Yang-Mills couplings $g_i$ ($i=1,\dots,k$) have strictly negative $\beta$-functions. As discussed in \S.\,\ref{34zaq}, the fact that the theory is asymptotically-free means that its cluster category $\mathcal{C}$ is not periodic. However, its Coulomb branch is parametrized by $k$ operators whose dimension in the UV limit $g_i\to0$ becomes $\Delta=2$. As in \textbf{Example \ref{kkz19ab}}, this implies the existence of a $1$-periodic sub-category $\mathcal{F}(1)\subset\mathcal{C
}$. Iff all YM couplings $g_i$ are strictly asymptotically-free, the category $\mathcal{F}(1)$ consists of $k$ copies of the $1$-periodic sub-category of pure $SU(2)$, \textbf{Example \ref{kkz19ab}}.
In such an asymptotic-free theory the $S$-duality group is bound to be `small' since all auto-equivalence $\sigma$ of the cluster category should preserve the 1-periodic sub-category $\mathcal{F}(1)$; therefore, up to (possibly)  permutations of the various $SU(2)$ gauge factors, $\sigma$ should restrict to a subgroup of autoequivalences of the periodic category $\mathcal{F}(1)_\text{pure}$ of pure $SU(2)$ SYM.
As we saw in \textbf{Example \ref{kzanncv}}, the $S$-dualities corresponding to shifts of the Yang-Mills angle $\theta$ preserve\footnote{\ Physically this is obvious. Mathematically, consider e.g.\! the shift shift $\theta\to \theta-4\pi+N_f\pi$ in $SU(2)$ with $N_f$ flavors. It corresponds to the auto-equivalence $\ca\mapsto\ca[1]$, which acts trivially on the 1-periodic subcategory.} the subcategory $\mathcal{F}(1)_\text{pure}$.
Thus
besides shifts of the various theta angles, permutations of identical subsectors, and flavor Weyl groups/Dynkin graph automorphism, we do not expect additional $S$-dualities in these models.
Let us check this expectaction against the computer search for dualities in a tricky example.

\begin{example}[$SU(2)^3$ with $\tfrac{1}{2}(\boldsymbol{2},\boldsymbol{2},\boldsymbol{2})$] 
A quiver for this model is given in figure \ref{fig:quiv}. 
\begin{figure}
\centering
\includegraphics[width=0.4\textwidth]{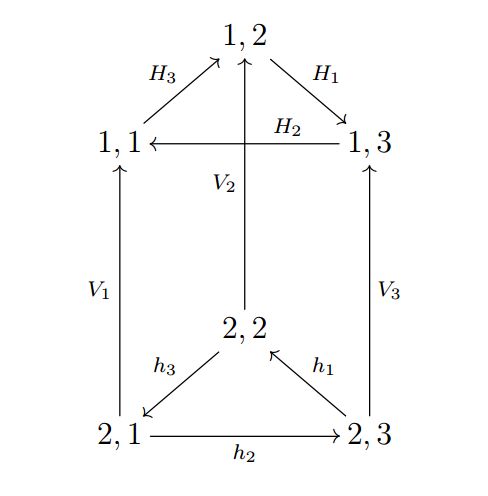}
\caption{A quiver $Q_\text{pris}$ for the gauge theory with $G_\text{gauge}=SU(2)^3$ coupled to a half-hyper in the three-fundamental.
The superpotential for $Q_\text{pris}$ is 
$W_\text{pris}=\mathrm{Tr}(H_1H_2H_3)+\mathrm{Tr}(h_1h_2h_3).$ }
\label{fig:quiv}
\end{figure}
In this case the cluster Grothendieck group $K_0(\mathcal{C}_\text{pris})$  is pure torsion, since a single half-hyper carries no flavor charge.
The three $SU(2)$ gauge couplings $g_i$ are asymptotically-free and the cluster category $\mathcal{C}_\text{pris}$ is \emph{not} periodic but it contains the
$1$-periodic subcategory $\mathcal{F}(1)\subset \mathcal{C}_\text{pris}$ described above\footnote{\ Notice that there is no periodic sub-category associated to the quark sector; this is related to the absence of conserved flavor currents in this model.}. The $S$-duality group is then expected to consists of permutations of the three $SU(2)$'s and the three independent shifts of the Yang-Mills angles $\theta_i\to \theta_i-2\pi$, that is,
$\mathbb{S}=S_3\ltimes \Z^3$.

The computer algorithm produced the following three commuting generators of the cluster automorphism group of infinite
$$\theta_1=\mu_{23}\mu_{22}\mu_{11},\qquad\theta_2=\mu_{21}\mu_{23}\mu_{12},\qquad
\theta_3=\mu_{22}\mu_{21}\mu_{13}.$$
These three generators are identified with the three $\theta$-shifts.
\end{example}

\begin{remark} Since the model is of class $\mathcal{S}[A_1]$ (with irregular poles), the $S$-duality group may also be computed geometrically (see section 7). The
computer result is of course consistent with geometry: each $\theta_i$ translation correspond to a twist around one of the three holes on the sphere: their order is clearly infinite and the three twists commute with one another.
\end{remark}

\subsection{$Q$-systems as groups of $S$-duality}

The above discussion may be generalized to all $\mathcal{N}=2$ QFTs having a weakly coupled Lagrangian formulation. If the gauge group $G$ is a product of $k$ simple factors $G_i$, we expect the $S$-duality group to contain a universal subgroup $\Z^k$ consisting of shifts $\theta_i\to\theta_i-b_i\pi$, with $b_i$ the $\beta$-function coefficient of the $i$-th YM coupling. One may run the algorithm and find the universal subgroup; however, just because it is \emph{universal,} its description in terms of quiver mutations also has a universal form which is easy to describe.

We begin with an example.

\begin{example}[Pure SYM: simply-laced gauge group] If the gauge group $G$ is simply-laced, the exchange matrix of its quiver may be put in the form \cite{cecotti2010r,alim2014mathcal,cecotti2013categorical}\footnote{\ In particular, $\mathrm{coker}\,B=\boldsymbol{Z}(G)^\vee\oplus\boldsymbol{Z}(G)$ is the correct 't Hooft group for pure SYM.}
$$
B=\left(\begin{array}{c|c}
0 & \phantom{-}C\\\hline
-C & 0\end{array}\right)\equiv C\otimes i\sigma_2,\qquad C\equiv \text{the Cartan matrix of $G$.}
$$
For instance, the quiver for $SU(N)$ SYM is represented in figure \ref{quiv:sun}.
\begin{figure}
\begin{displaymath}
    \xymatrix{\circ \ar[dr] & \circ \ar[ld] \ar[rd] & \circ \ar[ld] \ar[rd]& ... \ar[rd] \ar[ld] &  \circ \ar[ld]\\
              \bullet \ar@2{->}[u] & \bullet \ar@2{->}[u] &  \bullet \ar@2{->}[u] & ... & \bullet \ar@2{->}[u]}
\end{displaymath}
\caption{The BPS quiver for pure SYM theory with gauge group $SU(N)$.}
\label{quiv:sun}
\end{figure}
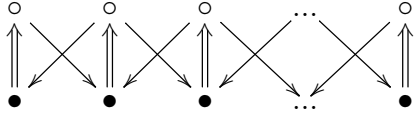
These quivers are \emph{bipartite}:
we may color the nodes black and white so that a node is linked only to nodes of the opposite color. Quiver mutations at nodes of the same color commute, so the product $$\nu=\prod_{i\ \text{white}}\mu_i$$ is well-defined. Moreover interchanging (black) $\leftrightarrow$ (white) yields the opposite quiver $Q^\text{opp}$ which is isomorphic to $Q$ \emph{via} the node permutation $\pi=\boldsymbol{1}\otimes\sigma_1$. The effect of the canonical mutation $\nu$ on the quiver is to invert all arrows i.e.\! it gives back the same quiver up to the involution $\pi$. Thus $\nu$ corresponds to an universal duality of pure SYM. One checks that it has infinite order, i.e.\! generates a subgroup of $S$-dualities isomorphic to $\Z$.

This sub-group $\Z$ of $S$-dualities has different physical interpretations/applications in statistical physics \cite{kedem1,kedem2,DiFrancesco:2008mc} as well as in the context of the Thermodynamical Bethe Ansatz \cite{cecotti2014systems}.
Indeed, consider its index 2 subgroup generated by the square of $\nu$
$$ \nu^2=\prod_{j \ \text{black}}\mu_j \prod_{i\ \text{white}}\mu_i.$$
The repeated application of the $S$-duality $\nu^2$ generates a recursion relation for the cluster variables which is known as the \textit{$Q$-system of type $G$}. It has deep relation with the theory of quantum groups; moreover it generates a linear recursion relation of finite length and has many other ``magical'' properties
\cite{kedem2,DiFrancesco:2008mc,cecotti2014systems}.

We claim that the duality $\nu^2$
corresponds to a shift of $\theta$.
Indeed, the cluster category in this case is the triangular hull of the orbit category of $D^b(\mathsf{mod}\,\C\hat A_{1,1}\otimes \C G)$ and $\nu^2$ corresponds to the auto-equivalence $\tau\otimes \mathrm{Id}$
\cite{caorsi2016homological}. Comparing the action of $\tau\otimes\mathrm{Id}$ in the covering category with the Witten effect (along the lines of \S.\ref{pppq12x}) one gets the claim.
\end{example}

\begin{example}[SYM with non-simply laced gauge group]
The authors of ref.\cite{DiFrancesco:2008mc} defined $Q$-systems also for non-simply laced Lie groups. To a simple Lie group $G$ one associates a quiver
and a mutation $\nu^2$ which generates a group $\Z$ which has all the required ``magic'' properties.
In ref.\cite{Cecotti:2012gh}  it was shown that the non-simply-laced $Q$-system does give the quiver description of the BPS sectors of the corresponding SYM theories.
The $Q$-system group is again the group of $S$-dualities corresponding to $\theta$-shifts.  
\end{example}

\begin{example}[General $\mathcal{N}=2$ SQCD models]
We may consider the general Lagrangian case in which the gauge group is a product of simple Lie groups, $\prod_j G_j$ and we have hypermultiplets in some representation of the gauge group.
The quivers for such a theory may be found in refs.\cite{alim2014mathcal,Cecotti:2012gh,cecotti2012quiver}.
For instance figure \ref{quiv:sunmatter} shows the quiver for $SU(M)\times SU(N)$ gauge theory coupled to 2 flavors of quark in the $(\boldsymbol{N},\boldsymbol{1})$ and a quark bifundamental in the $(\boldsymbol{\bar N},\boldsymbol{M})$. It is easy to check that the two canonical mutations of the subquivers associated to the two simple factors of the
gauge group 
$$
\nu_\circ=\prod_{i=\circ}\mu_i,\qquad \nu_\square=\prod_{i=\square}\mu_i,
$$  
leave the quiver invariant up to the permutation 
$\circ\leftrightarrow\bullet$ and, respectively, $\square\leftrightarrow\blacksquare$. The construction extends straightforwardly to any number of gauge factors $G_j$ and all matter representations. The conclusion is that  we have a canonical $\Z$ subgroup of the $S$-duality group per simple factor of the gauge group. It corresponds to  shifts of the corresponding $\theta$-angle.
If the matter is such that the $\beta$-function vanishes, the full cluster category becomes periodic, and we typically get a larger $S$-duality group.
\end{example}

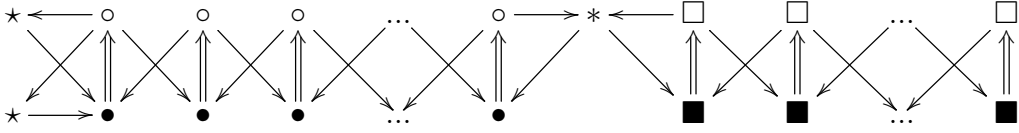
\begin{figure}
\begin{displaymath}
    \xymatrix{\star\ar[dr]&\circ \ar[dr] \ar[l] \ar[dl]& \circ \ar[ld] \ar[rd] & \circ \ar[ld] \ar[rd]& ... \ar[rd] \ar[ld] &  \circ \ar[ld] \ar[r] & \ast\ar[dr]\ar[dl] &\square \ar[dr]\ar[l] & \square \ar[ld] \ar[rd] &  ... \ar[rd] \ar[ld] &  \square \ar[ld]\\
            \star\ar[r]&  \bullet \ar@2{->}[u] & \bullet \ar@2{->}[u] &  \bullet \ar@2{->}[u] & ... & \bullet \ar@2{->}[u] & & \blacksquare \ar@2{->}[u] & \blacksquare \ar@2{->}[u] &   ... & \blacksquare \ar@2{->}[u]  }
\end{displaymath}
\caption{The BPS quiver of $SU(N) \times SU(M)$ SQCD with two quarks ($\star$ nodes) in the fundamental representation of $SU(N)$ and one (the $\ast$ node)  in the bifundamental representation. The number of $\circ$ (resp.\! $\square$)  is $N$ (resp.\! $M$).}
\label{quiv:sunmatter}
\end{figure}
One can convince himself that the sequence of mutations $\mu$ does not change the quiver and that its order, in all the above cases, is infinite, as it is for the shift is the $\theta_j$'s.

\section{Class $\mathcal{S}$ QFTs: Surfaces, triangulations, and categories}
\label{sec:surfaces}

In this section we focus on a  special class of $\mathcal{N}=2$ theories: the Gaiotto $\mathcal{S}[A_1]$ models \cite{gaiotto2012n}. We study them for two reasons: first of all they are interesting for their own sake, and second for these theories the three categories $D^b\Gamma$, $\mathfrak{Per}\,\Gamma$, 
and $\mathcal{C}(\Gamma)$ have geometric constructions, directly related to the WKB analysis of \cite{gaiotto2013wall,gaiotto2013framed}. Comparing the categorical description with the results of refs.\cite{gaiotto2013wall,gaiotto2013framed} we check the correctness of our physical interpretation of the various categories and functors. 
\medskip

Class $\mathcal{S}[A_1]$ theories are obtained by the compactification of the $6d$ $(2,0)$ SCFT of type $A_1$ over a complex curve $C$ having regular and irregular punctures \cite{gaiotto2013wall}. 
If there is at least one puncture, these theory have the quiver property \cite{cecotti2011classification}, and their quivers with superpotentials are constructed in terms of an ideal triangulation of $C$ \cite{labardini2008quivers}. In the geometrical setting of Gaiotto curves, we can interpret the categories defined in section \ref{sec:cat} as categories of (real) curves on the spectral cover of the Gaiotto curve $C$. 

When only irregular punctures are present, the quiver with potential arising from these theories \cite{labardini2008quivers} has a Jacobian algebra which is \emph{gentle} \cite{assem2010gentle,cecotti2015galois}, and hence all triangle categories associated to its BPS sector, eqn.\eqref{exactse}, have a simple explicit description.\footnote{\ In facts, there is a systematic procedure, called \emph{gentling} in ref.\cite{cecotti2015galois} which allow to reduce the general class $\mathcal{S}[A_1]$ model to one having a gente Jacobian algebra.} 
When only regular punctures are present, the $\mathcal{N}=2$ theory has a Lagrangian formulation (which is weakly coupled in some corner of its moduli space)
and is UV superconformal. In particular the corresponding cluster category is 2-periodic, as the arguments of \S.\,\ref{kkkaq12} imply.

\subsection{UV and IR descriptions}
The main reference for this part is \cite{gaiotto2012n}. 

In the deep UV a class $\mathcal{S}[A_1]$ $\mathcal{N}=2$ theory is described by the Gaiotto curve $C$, namely a complex curve of genus $g$
with a number of punctures $x_i\in C$. Punctures are of two kinds: \emph{regular punctures} (called simply \emph{punctures}) and \emph{irregular} ones (called \emph{boundaries}). The $i$--th boundary carries a positive integer $k_i\geq1$ (the number of its \emph{cilia}); sometimes it is convenient to regard regular punctures as  boundaries with $k_i=0$. Iff $k_i\leq 2$ for all $i$, the $\mathcal{N}=2$ theory is a Lagrangian model with gauge group\footnote{\ When $g=0$, the theory is defined only if $b\geq 1$ or $b=0$ and $p\geq 3$; in case $p=0$, $b=1$ we require $k\geq 4$; when $g=1$ we need $p+b\geq 1$. Except for the case $p=0$, $b=1$, corresponding to Argyres-Douglas of type $A$, $m$ in eqn.\eqref{gaugeggr} is $\geq0$. $m=0$ only for Argyres-Douglas of type $D$ \cite{cecotti2013more}.}
\begin{equation}\label{gaugeggr}
\mathcal{G}=SU(2)^m,\qquad m=3g-3+p+2b\quad \text{where }\begin{cases}
p=\#\{\text{regular punctures}\}\\
b=\#\{\text{boundaries}\}.\end{cases}
\end{equation} 
If $b=0$ the theory is superconformal in the UV, and the space of exactly
marginal coupling coincides with the moduli space of genus $g$ curves with $p$ punctures, $\mathcal{M}_{g,p}$, whose complex dimension is $m$ $\equiv$ the rank of the gauge group $\mathcal{G}$.
Instead, if $b\geq 1$ (and $m\geq 2$),
$b$ out of the $m$ $SU(2)$ factors in the Yang-Mills group $\mathcal{G}$ have \emph{asymptotically free} couplings; these $b$ YM couplings
go to zero in the extreme UV,
so that the UV marginal couplings are again equal in number to the complex deformations $\mathcal{M}_{g,p+b}$ of $C$.

If some of the boundaries have $k_i\geq 3$, we have a gauge theory with the same gauge group $SU(2)^m$ coupled to ``matter'' consisting, besides free quarks (in the fundamental, bi-fundamental, and three-fundamental of $\mathcal{G}$), in an Argyres-Douglas SCFT of type $D_{k_i}$ for each boundary\footnote{\ Argyres-Douglas of type $D_1$ is the empty theory and the one of type $D_2$ a fundamental quark doublet.} \cite{cecotti2013more}. The space of exactly marginal deformations is as before.\medskip

The IR description of the model is given by the Seiberg-Witten curve $\Sigma$ which, for class $\mathcal{S}[A_1]$, is a double cover of $C$. More precisely, one considers in the total space of the $\mathbb{P}^1$-bundle
\begin{equation*}
\mathbb{P}(K_C\oplus \mathcal{O}_C)\to C\qquad \text{where }\begin{cases}K_C&\text{canonical line bundle on }C\\
\mathcal{O}_C&\text{trivial line bundle on }C
\end{cases}
\end{equation*}
the curve 
\begin{equation*}
\Sigma\equiv\Big\{ y^2=\phi_2(x)\,z^2\;\Big|\; (y,z)\ \text{homogeneous coordinates in the fiber}\Big\}\to C,
\end{equation*}
where $\phi_2(x)$ is a quadratic differential on $C$ with poles of degree at most $k_i+2$ at $x_i$.
The Seiberg-Witten differential is the tautological one
\begin{equation*}
\lambda= \frac{y}{z}\,dx
\end{equation*}
whose periods in $\Sigma$ yield the $\mathcal{N}=2$ central charges of the BPS states.

The dimension of the space of IR deformations is then\footnote{\label{foota}\ This formula holds under the condition $\dim\mathcal{M}_{g,p+b}=3g-3+p+b\geq0$.} \footnote{\ Here the asymptotically-free gauge couplings are counted as IR deformations.}
\begin{equation*}
\begin{split}
s&=\dim H^0(C, P K_C^2)\equiv 3g-3+2p+2b+\sum_i k_i\\
&\text{where }P=\sum_i(k_i+2)[x_i],
\end{split}
\end{equation*}  
so that the total space of parameters, UV+IR, has dimension
\begin{equation}\label{whatrank}
n=m+s= 6g-6+3p+3b+\sum_ik_i.
\end{equation}
There are two kinds of IR deformations, \emph{normalizable} and \emph{unnormalizable} ones. The unnormalizable ones correspond to deformations of the Lagrangian, while the normalizable ones  to moduli space of vacua (that is, Coulomb branch parameters);
their dimensions are\footnote{\ As written, these equations hold even if the condition in footnote \ref{foota} is not satisfied. Notice that we count also the dimensions of the internal Coulomb branches of the matter Argyres-Douglas systems.}
\begin{equation*}
\begin{split}
s_\text{nor}&\equiv \dim \text{(Coulomb branch)}=3g-3+p+b+\sum_i \left(k_i-\left[\frac{k_i}{2}\right]\right),\\ s_\text{un-nor}&=s-s_\text{nor}.
\end{split}
\end{equation*}
\medskip

The double cover $\Sigma\to C$ is branched over the zeros 
$w_a\in C$ of the quadratic differential $\phi_2(x)$. Their number is
$$
t= 4g-4+2p+2b+\sum_i k_i,
$$
but their positions are constrained by the condition that the divisor $\sum_a[w_a]$ is linear equivalent to $PK_C^2$, so that their positions 
depend on only $t-g$ parameters;
$\phi_2(x)$ depends on one more parameter
$$s=t-g+1$$
since the positions of its zeros fix $\phi_2(x)$ only up to an overall scale (that is, up to the overall normalization of the Seiberg-Witten differential, which is the overall mass scale).

Therefore, up to the overall mass scale, giving the cover $\Sigma\to C$ is equivalent to specifying the zeros $w_a\in C$ of the quadratic differential. Indeed, double covers are fixed, up to isomorphism, by their branching points. 
We shall refer to the points $w_a\in C$ as \emph{decorated points} on the Gaiotto curve $C$. 

In summary: the UV description of a class $\mathcal{S}[A_1]$ amounts to giving the datum of the Gaiotto curve $C$, that is, a complex structure of a genus $g$ Riemann surface and a number of \emph{marked} points $x_i\in C$ together with a non-negative integer $k_i$ at each marked point. To get the IR description we have to specify, in addition, the \emph{decorated} points
$w_a\in C$ (whose divisor is constrained to be linear equivalent to $PK_C^2$). We may equivalently state this result in the form:

\begin{principle}
In theories of class $\mathcal{S}[A_1]$, to go from the IR to the UV description we simply delete (i.e.\! forget)
the decorated points of $C$.
\end{principle}

We shall see in \S.\ref{sec:geomrepcat} below that this `forget the decoration' prescription is precisely the map denoted $\mathsf{r}$ in the exact sequence of triangle categories of eqn.\eqref{exactse}.

\subsection{BPS states}

In class $\mathcal{S}[A_1]$ $\mathcal{N}=2$ theories, the natural BPS objects are described by (real) curves $\eta$ on the Seiberg-Witten
curve $\Sigma$ which are \emph{calibrated} by the Seiberg-Witten differential \cite{seiberg1994monopoles}
\begin{equation*}
\lambda= \frac{y}{z}\,dx\equiv \sqrt{\phi_2(x)}\,dx,
\end{equation*}
that is, they are required to satisfy the condition (we set $\phi\equiv \phi_2\,(dx)^2$),
%
%
%
%
%
%
%
\begin{equation}\label{whattheta}
\sqrt{\phi}\Big|_\eta=e^{i\theta}\,dt,\qquad \text{here }t \in \R,
\end{equation}
for some real constant $\theta$, and are \emph{maximal} with respect to this condition. Being maximal, $\eta$ may terminate only at marked or decorated points.\footnote{\ We call
marked/decorated points in $\Sigma$ the pre-images of marked/decorated points on $C$.}  BPS particles have finite mass, i.e.\! they correspond to calibrated arcs $\eta$ with $|Z(\eta)|<\infty$ where 
the central charge of the would be BPS state $\eta$,  is 
\begin{equation}\label{whatZ}
Z(\eta)=\int_\eta \lambda.
\end{equation}
In this case, the parameter $\theta$ in eqn.\eqref{whattheta} is given by $\theta=\arg Z(\eta)$.
Arcs $\eta$ associated to BPS particles may end only at decorations.
All other maximal calibrated arcs have infinite mass and are interpreted as BPS \emph{branes} \cite{shapere1999bps}.

There are two possibilities for BPS particles:
\begin{itemize}
\item they are closed arcs  connecting zeros of $\phi$. These calibrated arcs are rigid and hence correspond to BPS hypermultiplets;
\item they are loops. Such calibrated arcs form $\mathbb{P}^1$-families and give rise to BPS vector multiplets.
\end{itemize}

\paragraph{Conserved charges.} From eqn.\eqref{whatZ} we see that the central charge of an arc $\eta$
factors through its homology class $\eta\in H_1(\Sigma,\Z)$. More precisely,
since the Seiberg-Witten differential $\lambda$ is \emph{odd} under the covering group $\Z_2\cong \mathsf{Gal}(\Sigma/C)$, $Z$ factors 
through the free Abelian group
\begin{align}
&\Lambda\equiv H_1(\Sigma,\Z)_\text{odd},\\ &\mathrm{rank}\,\Lambda=2\big(g(\Sigma)-g(C)\big)+\#\big\{k_i\ \text{even}\big\}.\label{raanklamn}
\end{align}
Applying the Riemann-Hurwitz formula\footnote{\ Compare eqns.(6.26)-(6.28) in ref.\cite{cecotti2011classification}.} to the cover $\Sigma\to C$, we see that the rank of $\Lambda$ is equal to the number $n$ of UV+IR deformations, see eqn.\eqref{whatrank}. In turn $n$ is the number of conserved charges (electric, magnetic, flavor) of the IR theory. Hence the group homomorphism 
\begin{equation*}
Z\colon \Lambda\to \C,\qquad [\eta]\mapsto \int_\eta \lambda,
\end{equation*}
is the map which associates to the IR charge $\gamma\in\Lambda$ of a state of the $\mathcal{N}=2$ theory
the corresponding central charge $Z(\gamma)$. An arc with homology
$[\eta]\in\Lambda$ then has `mass'
$$
\int_\eta |\lambda|\geq |Z(\eta)|
$$
with equality iff and only if it is calibrated, that is, BPS.
\bigskip

To get the corresponding  UV statements,
we apply to these results our
\textbf{RG principle}, that is, we forget the decorations. BPS particles then disappear (as they should form the UV perspective), while BPS branes project to arcs on the Gaiotto curve $C$. Several IR branes project to the same arc on $C$. The arcs on $C$ have the interpretation of UV BPS \emph{line operators}, and the branes which project to it are the objects they create in the given vacuum (specified by the cover $\Sigma\to C$) which may be dressed (screened) in various ways by BPS states, so that the IR-to-UV correspondence is
many-to-one in the line sector. 

The UV conserved charges is the projection of $\Lambda$; over $\mathbb{Q}$ all electric/magnetic charges are projected out by the oddness condition, and we remain with just the flavor lattice. However over $\Z$ the story is more interesting and we get \cite{cecotti2011classification}
\begin{equation}\label{lalala}
\Lambda_\text{UV}\cong\Z^{\#\{k_i\ \text{even}\}}\oplus \text{2-torsion.}
\end{equation}

Comparing with our discussion in the \textbf{Introduction}, we see that
$\Lambda$ and $\Lambda_\text{UV}$ should be identified with the Grothedieck group of the triangle categories $D^b\Gamma$ and $\mathcal{C}(\Gamma)$, respectively. We shall check these identifications below.

\subsection{Quadratic differentials}
\label{sec:phys}
We want to study the BPS equations to find the BPS spectrum of the theory. We start by analyzing the local behavior of the flow of the quadratic differential.
The quadratic differential near a zero can be locally analyzed in a coordinate patch where $\phi \sim w$; thus we have to solve the following equation:
$$\sqrt{w}\,\frac{dw}{dt}=e^{i \theta},$$
which gives $w(t)=(\frac{3}{2}te^{i \theta}+w_0^{3/2})^{2/3}$. We plot here the solution:
\begin{figure}[H]
\centering
\includegraphics[width=0.6\textwidth]{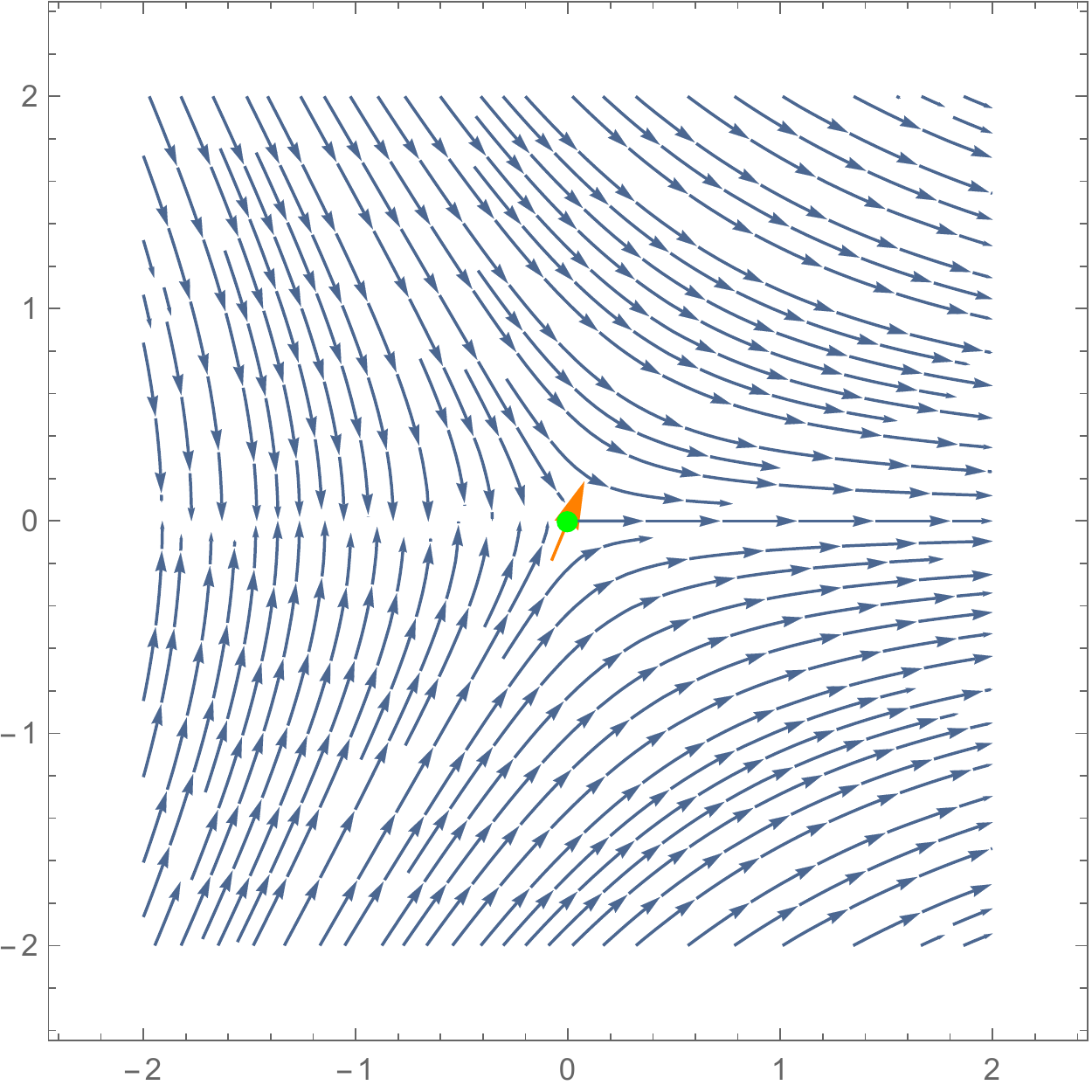}
\end{figure}
The three straight trajectories, which all start at the zero of $\phi$, end at infinity in the poles of $\phi$, i.e. the marked points on the boundaries of $C$. Since all zeros of $\phi$ are simple by hypothesis, it is possible to associate a triangulation to a quadratic differential by selecting the flow lines connecting the marked points as shown in the following figure: 
\begin{figure}[H]
\centering
\begin{tabular}{cc}
\resizebox{0.48\textwidth}{!}{\xygraph{
!{<0cm,0cm>;<0.5cm,0cm>:<0cm,0.5cm>::}
!{(0,0) }*+{\bullet}="b1"
!{(0,-10) }*+{\bullet}="b3"
!{(-5,-5) }*+{\bullet}="b2"
!{(5,-5) }*+{\bullet}="b4"
!{(2.5,-5) }*+[green]{\bullet}="b5"
!{(-2.5,-5) }*+[green]{\bullet}="b6"
"b1" -@[blue] "b2" "b2" -@[blue] "b3" "b3" -@[blue] "b4" "b4" -@[blue] "b1" "b3" -@[blue] "b1"
"b5" -@[green] "b1" "b5" -@[green] "b3" "b5" -@[green] "b4" "b6" -@[green] "b1" "b6" -@[green] "b3" "b6" -@[green] "b2"
"b1"-@/^0.3cm/"b2" 
"b1"-@/_0.3cm/"b2" 
"b1"-@/_0.6cm/"b2" 
"b3"-@/^0.3cm/"b2" 
"b3"-@/_0.3cm/"b2" 
"b3"-@/^0.6cm/"b2" 
"b3"-@/^0.3cm/"b4" 
"b3"-@/_0.3cm/"b4" 
"b3"-@/_0.6cm/"b4" 
"b1"-@/^0.3cm/"b4" 
"b1"-@/_0.3cm/"b4" 
"b1"-@/^0.6cm/"b4" 
"b1"-@/^0.3cm/"b3" 
"b1"-@/_0.3cm/"b3" 
"b1"-@/_0.6cm/"b3" 
"b1"-@/^0.6cm/"b3" 
}} &
\resizebox{0.48\textwidth}{!}{\xygraph{
!{<0cm,0cm>;<0.5cm,0cm>:<0cm,0.5cm>::}
!{(0,0) }*+{\bullet}="b1"
!{(0,-10) }*+{\bullet}="b3"
!{(-5,-5) }*+{\bullet}="b2"
!{(5,-5) }*+{\bullet}="b4"
!{(2.5,-5) }*+[green]{\bullet}="b5"
!{(-2.5,-5) }*+[green]{\bullet}="b6"
"b1" -@[blue] "b2" "b2" -@[blue] "b3" "b3" -@[blue] "b4" "b4" -@[blue] "b1" "b3" -@[blue] "b1"
}}\\
\text{Flow lines} & \text{Triangulation}
\end{tabular}
\end{figure}
Moreover, by construction, it is clear that to each triangle we associate a zero of $\phi$. These are the decorating points $\Delta$ of section \ref{sec:geomrepcat}.
If we make $\theta$ vary, we deform the triangulation up to a point in which the triangulation jumps: at that value of $\theta=\theta_c$, two zeros of $\phi$ are connected by a curve $\eta$: this curve is the stable BPS state. From what we have just stated, it will be clear that the closed arcs of section \ref{sec:geomrepcat} will correspond to BPS states.
Before and after the critical value of $\theta=\theta_c$, the triangulation undergoes a \emph{flip}. Flips of the triangulation correspond to mutations at the level of BPS quiver (see section \ref{sec:geomrepcat} on how to associate a quiver to a triangulation). This topic is develop in full details in \cite{bridgeland2015quadratic,assem2010gentle,qiu2013cluster}.\\
Moreover, the second class of BPS objects, i.e. loops representing vector-multiplets BPS states, appears in one-parameter families and behave as in example \ref{ex:su2}; the map $X$ of section \ref{sec:geomrepcat} will allow us to write the corresponding graded module $X(a) \in D^b(\Gamma)$.
\begin{figure}[H]
\centering
 \begin{tabular}{ccc}
\includegraphics[width=0.31\textwidth]{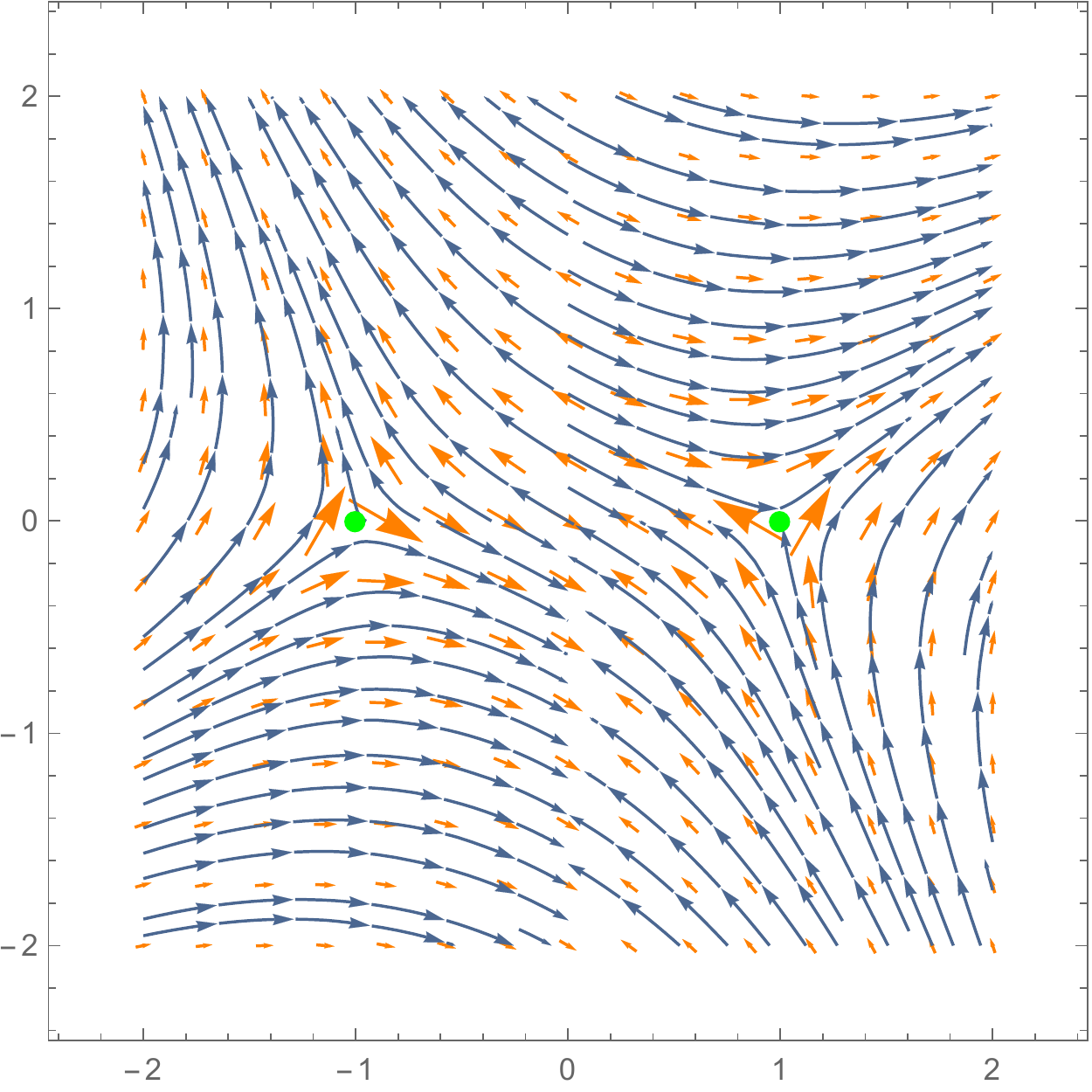} 
&
\includegraphics[width=0.31\textwidth]{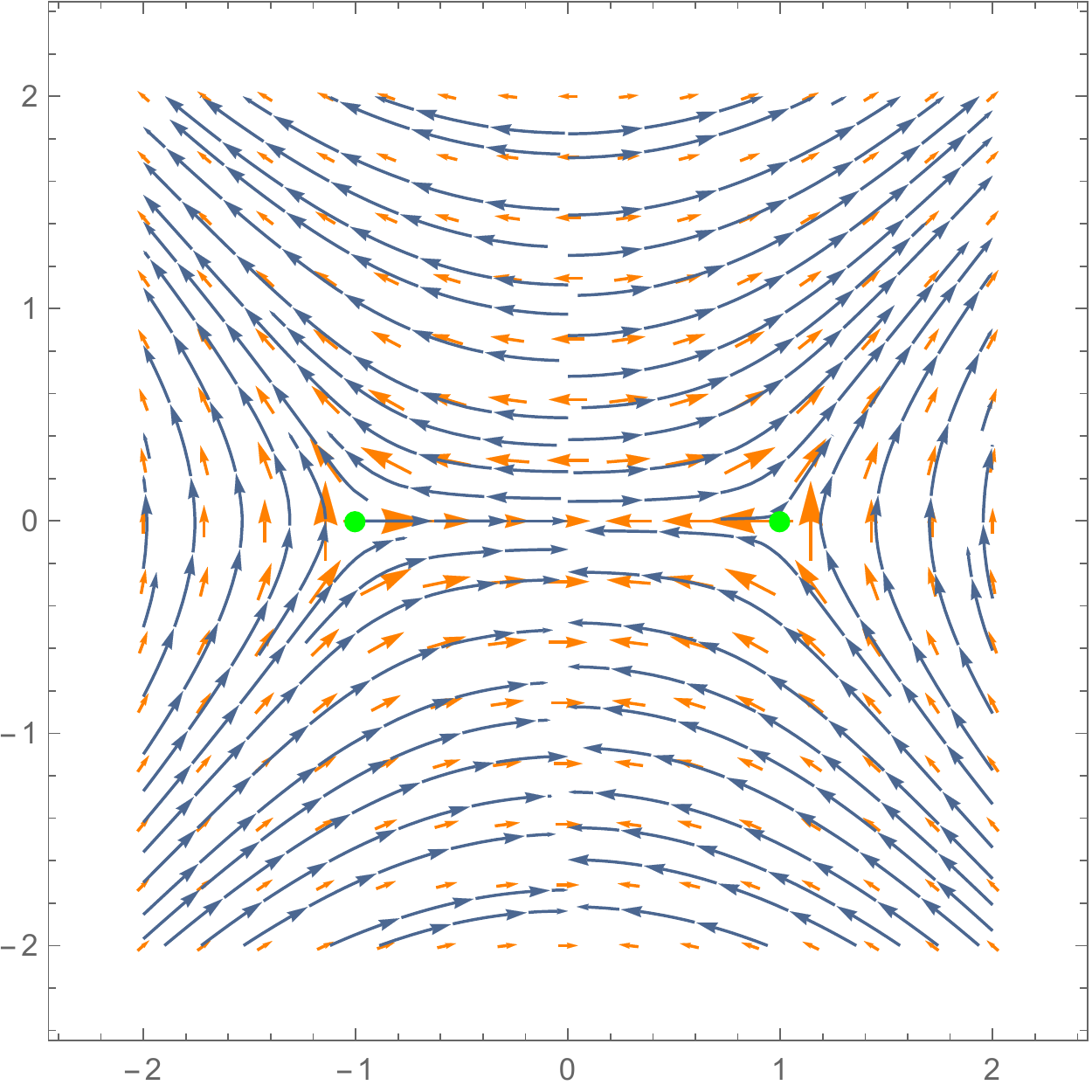} 
&
\includegraphics[width=0.31\textwidth]{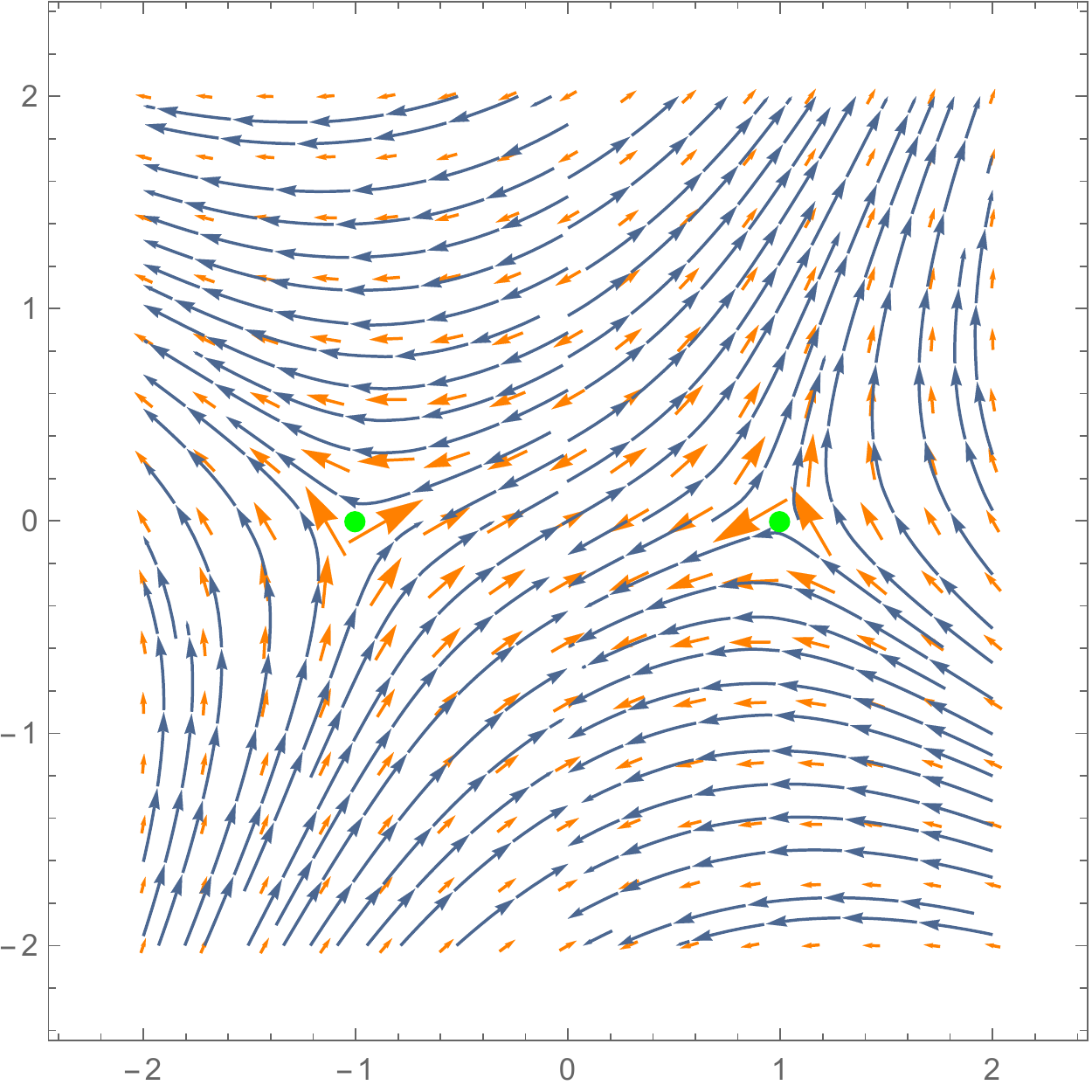}
 \\
$\theta < \theta_c$ & $\theta = \theta_c$ & $\theta > \theta_c$ 
\end{tabular}
\caption{For $\theta=\theta_c$ we have an arc corresponding to a BPS hypermultiplet: it is the solution connecting the two zero of the differential.}
\end{figure}
\begin{figure}[H]
\centering
 \begin{tabular}{ccc}
\includegraphics[width=0.31\textwidth]{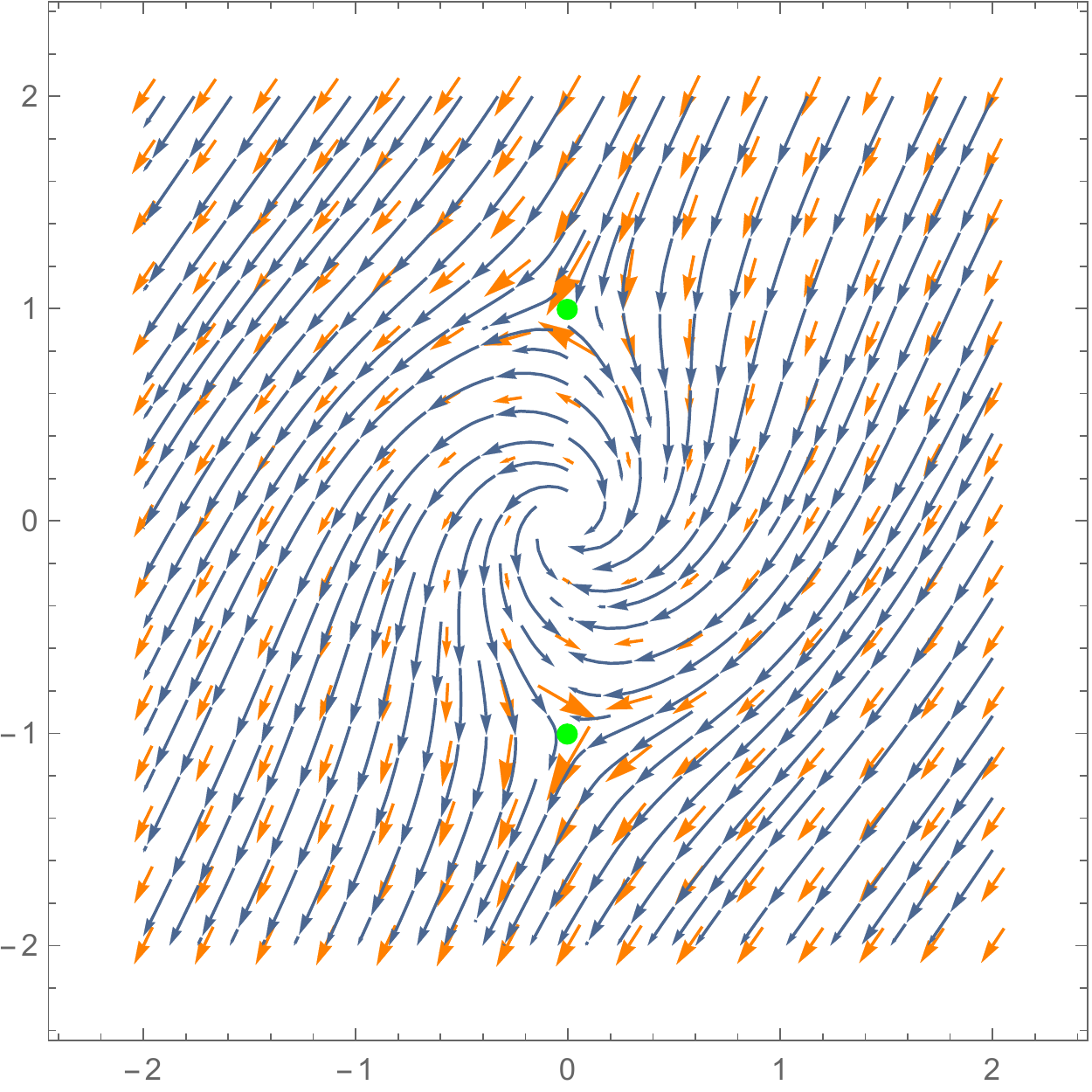}
 &
\includegraphics[width=0.31\textwidth]{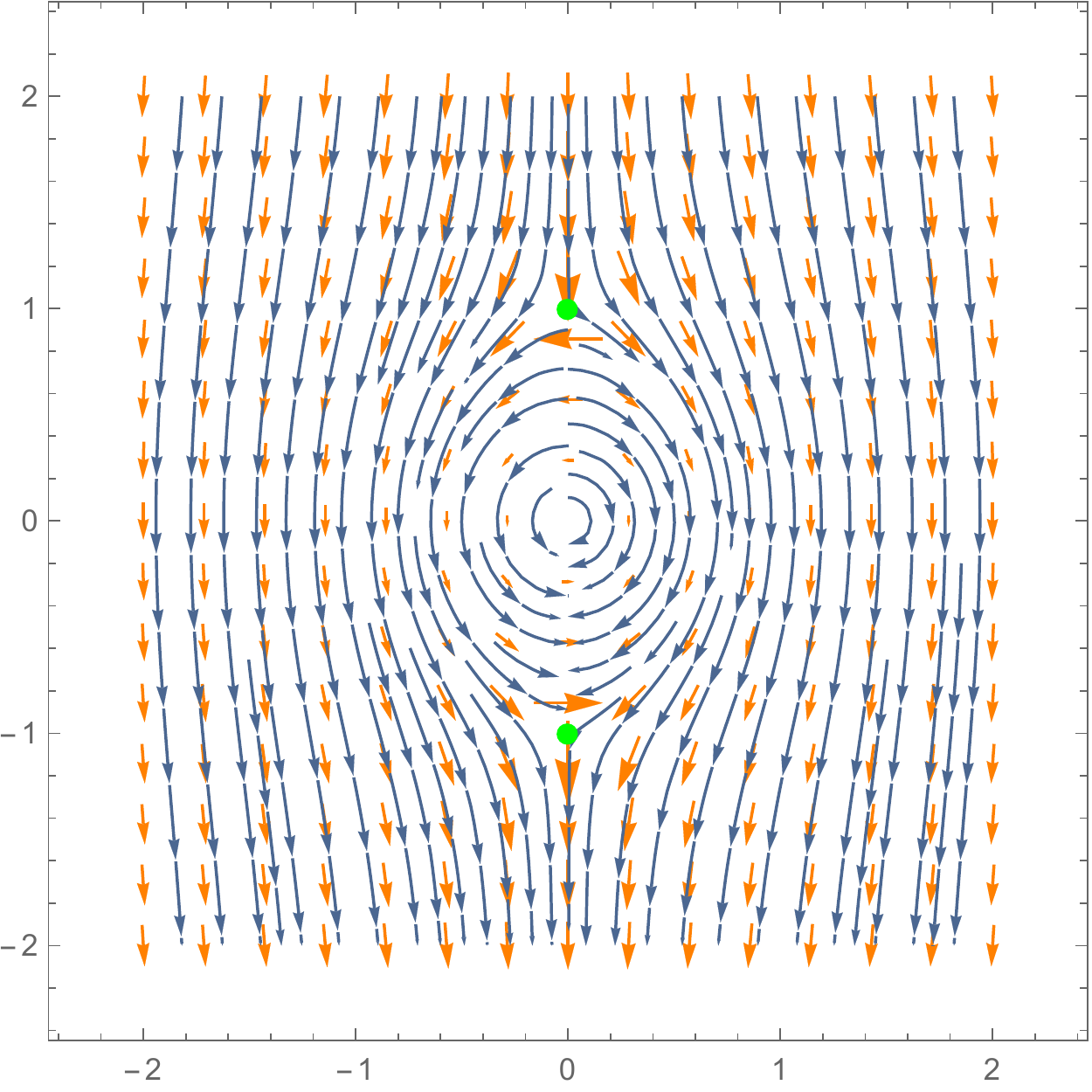} 
&
\includegraphics[width=0.31\textwidth]{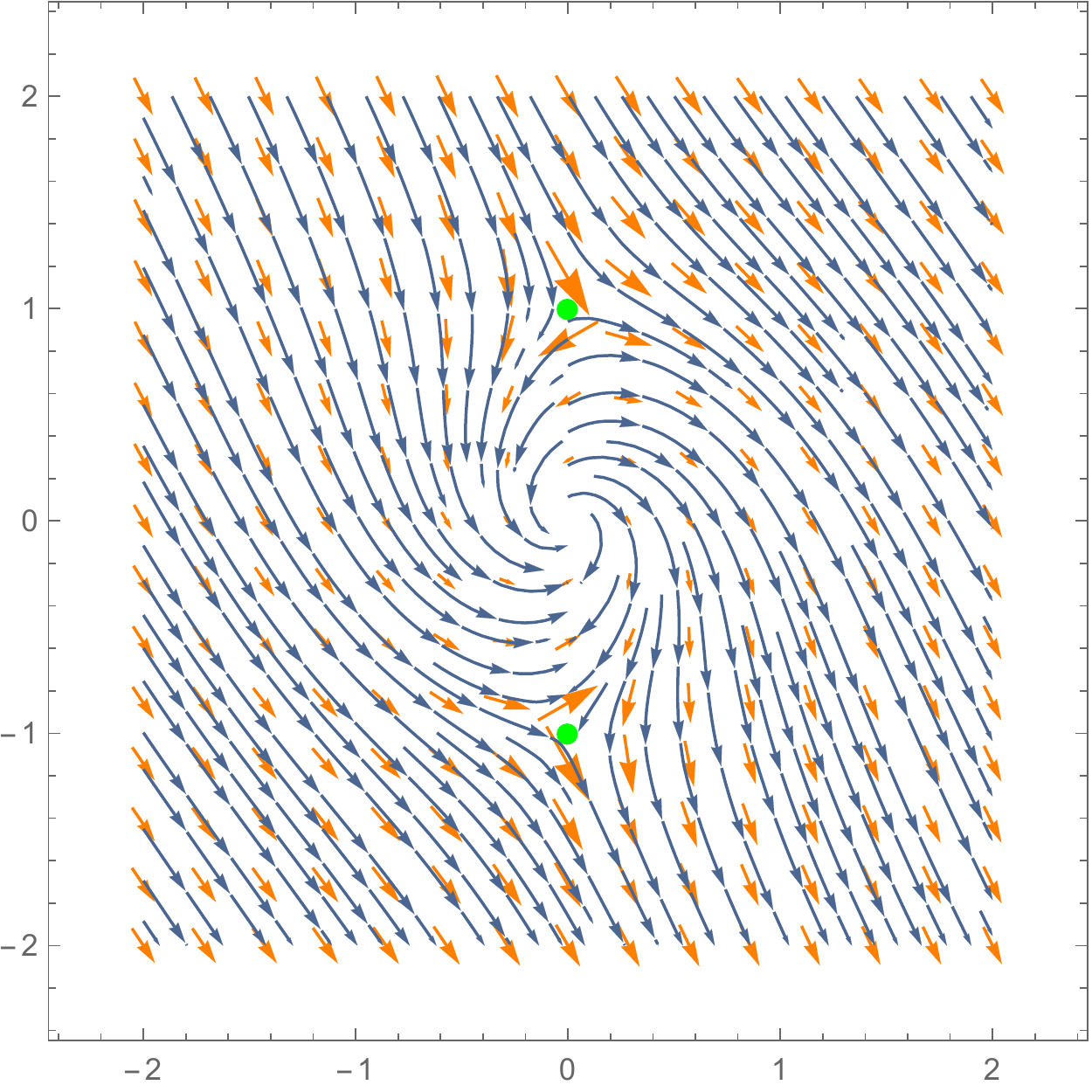} 
\\
$\theta < \theta_c$ & $\theta = \theta_c$ & $\theta > \theta_c$ 
\end{tabular}
\caption{The 1-parameter family of BPS curves corresponding to a vector multiplet appears for $\theta=\theta_c$.}
\end{figure}
What about the curves connecting punctures or marked points (but not zeros of the quadratic differential)?\footnote{\ These are the objects in the perfect derived category $\mathfrak{Per}\,\Gamma$} To answer this question, we have to take a detour into line operators.
\begin{remark}
We point out that the space of quadratic differential is isomorphic to the Coulomb branch of the theory. In a recent paper of Bridgeland \cite{bridgeland2015quadratic} it is stated that the space of stability conditions of $D^b \Gamma$ satisfy the following equation:
$$\mathrm{Stab}^0(D^b \Gamma)/\mathrm{Sph}(D^b\Gamma)\cong \mathrm{Quad}(S,M).$$
Thus, the Coulomb branch is isomorphic to the space of stability conditions, up to a spherical twist. Indeed, a stability condition is a pair $(Z,\mathcal P)$, where $Z$ is the stability function (i.e. the central charge) and the category $\mathcal P$ is the category of stable objects, i.e. stable BPS states (see section \ref{sec:stab} for more details).
\end{remark}
\subsection{Geometric interpretation of defect operators}
The main reference for this section is \cite{alday2010loop} and also \cite{gaiotto2013framed}. There it is explained how to use M-theory to construct vertex operators, line operators and surface operators by intersecting an M5 brane with an M2 one. In particular,
\begin{itemize} 
\item a \emph{vertex operator} corresponds to a point in the physical space: it means that the remaining two dimensions of the M2 brane are wrapped on the Gaiotto surface $S$, forming a co-dimension 0 object.
\item a \emph{line operator} corresponds to a one-dimensional object in the physical space: the remaining one dimension of the M2 brane is wrapped around the surface $S$ as a 1-cycle (i.e. non self-intersecting closed loop) or as an arc connecting two punctures or marked points.
\item a \emph{surface operator} has two spacial dimensions and thus it is represented by a puncture over the surface $C$.
\end{itemize}
These operators are physically and geometrically related. A vertex operator, since it is a sub-variety of codimension 0 over the complex curve $S$, can be interpreted as connecting line operators (i.e the loops at the boundary of the sub-variety representing the vertex operator). Moreover, a line operator $\gamma$ can be interpreted as acting on a surface operator $x$ and transporting the point corresponding to the surface operator around the curve $S$; the action on its dual Liouville field (\cite{alday2010loop}) is the monodromy action along the line associated to the line operator $\gamma$. Finally, as explained in \cite{gaiotto2014open}, when we describe the curve $S$ via a quadratic differential $\phi$, we can interpret the poles of $\phi$ as surface operators and the arcs connecting marked points and the closed loops correspond to Verlinde line operators\footnote{\ Traditionally, they are defined by
composing a sequence of elementary operations on conformal blocks, each corresponding
to a map between spaces of conformal blocks which may differ in the number or type
of insertions. Roughly speaking, one inserts an identity operator into the original
conformal block, splits into two conjugate chiral operators $\phi_a$ and $\bar \phi_a$, transports $\phi_a$
along $\gamma$ and then fuses the operators $\phi_a$ and $\bar\phi_a$ back to the identity channel.}. If we consider two curves $\gamma_1, \gamma_2$ that intersect at some point, the line operators $L(\gamma_1,\zeta),L(\gamma_2,\zeta)$ corresponding to these curves do not commute (see section \ref{sec:lindef}). In the case of self-intersecting arcs, we get more general Verlinde operators: as line operators, they can be decomposed into a linear combination of non self-intersecting line operators by splitting the intersection in all possible pairs.
\subsection{Punctures and tagged arcs}
The main reference, for this short section, is \cite{gaiotto2013wall}. In there it is pointed out how the arcs over the Gaiotto curve $C$ are tagged arcs: at each singularity (extrema of the arc) we have to choose the eigenvalue of the monodromy operator around that singular point. In particular, in section 8 of \cite{gaiotto2013wall}, we discover that for irregular singularity, the tagging is not necessary, since it is equal to an overall rotation of the marked points around that boundary component. For punctures, on the other hand, it is not the case: we have to specify a $\Z_2$-tagging. This boils down to a tagged triangulation, as explained in \cite{qiu2013cluster,gaiotto2013wall}. Therefore, when considering surfaces with punctures, we have to consider tagged arcs and not simple arcs. This observation will be important when we discover that the S-duality group is the tagged mapping class group of the Gaiotto surface $C$ (see section \ref{sec:sdualmap}).

%

\subsection{Ideal triangulations
surfaces and quivers with potential}

Let $C$ be the Gaiotto curve of a class $\mathcal{S}[A_1]$ model. The  invariant of the family of QFTs obtained by continuous deformations of it is the topological type of $C$.
More precisely, we
define the underlying topological surface $S$ of $C$ by the following steps: \textit{i)} forget the complex structure, and
\textit{ii)} replace each irregular punture with a boundary component $\partial S_i$ with $k_i\geq 1$ marked points (ordinary punctures on $C$ remain punctures on $S$). $S$ is then the invariant datum which describes the continuous family of $\mathcal{S}[A_1]$ theories.

 An \emph{ideal triangulation} of $S$ is a maximal set of pairwise non-isotopic arcs ending in punctures and marked points which do not intersect (except at the end points) and are not homotopic to a boundary arc.\footnote{\ A boundary arc is the part of a boundary component between two adjacent marked points.} All ideal triangulations have the same number of arcs \cite{fomin2008cluster}
\begin{equation*}
n=6g-6+3p+3b+\sum_ik_i.
\end{equation*}
Note that is the same number as the number of UV+IR deformations, eqn.\eqref{whatrank}, as well as the number of IR conserved charges $\mathrm{rank}\,\Lambda$, eqn.\eqref{raanklamn}.

To an ideal triangulation $T$ of the surface of $S$ we associate a quiver with superpotential $(Q,W)$. 
The association $S\leftrightarrow (Q,W)$ is intrinsic in the following sense:

\begin{proposition}[Labardini-Fragoso \cite{labardini2008quivers}] Let $(Q,W)$ be the quiver with potential associated to an ideal triangulation $T$ of the surface $S$. A quiver with potential $(Q^\prime,W^\prime)$ is \emph{mutation equivalent} to $(Q,W)$ if and only if it\footnote{\ This is slightly imprecise since, in presence of regular punctures $W$ contains free parameters \cite{labardini2008quivers}. The statement in the text refers to the full family of allowed $W$'s.} arises from an ideal triangulation $T^\prime$ of the \emph{same} surface $S$. 
\end{proposition}

In view of \textbf{Corollary \ref{corco}}, an important result is:

\begin{proposition}[Labardini-Fragoso \cite{labardini2008quivers}] The quiver with superpotential of an ideal triangulation is always Jacobi-finite.
\end{proposition}

Thus an ideal triangulation $T$ defines a Jacobi-finite Ginzburg DG algebra $\Gamma\equiv \Gamma(Q,W)$ and therefore also the three triangle categories $D^b\Gamma$,
$\mathfrak{Per}\,\Gamma$ and $\mathcal{C}(\Gamma)$ described in \S\S.\,2.5, 2.6. Then

\begin{corollary} Up to isomorphism, the three triangle categories $D^b\Gamma$,
$\mathfrak{Per}\,\Gamma$ and $\mathcal{C}(\Gamma)$ are independent of the chosen triangulation $T$, and hence are intrinsic properties of the topological surface $S$.
\end{corollary}

It remains to describe the quiver with potential $(Q,W)$ associated to the ideal triangulation $T$. The nodes of $Q$ are in one-to-one correspondence with the arcs $\gamma_i$ of $T$ (their number being equal to the number of IR charges, as required for the BPS quiver of any $\mathcal{N}=2$ theory). Giving the quiver, is equivalent to specifying its exchange matrix:

\begin{definition}
 For any triangle $D$ in $T=\{\gamma_i\}_{i=1}^n$ which is not self-folded, we define a matrix $B^D = (b^{D})_{ij}$ ,$1\leq i \leq n,1\leq j\leq n$ as follows.
\begin{itemize}
\item $b^D_{ij} = 1$ and $b^D_{ji} = -1$ in each of the following cases:
\begin{enumerate}
\item $\gamma_i$ and $\gamma_j$ are sides of $D$ with $\gamma_j$ following $\gamma_i$
in the clockwise order;
\item $\gamma_j$ is a radius in a self-folded triangle enclosed by a loop $\gamma_l$, and $\gamma_i$ and $\gamma_l$ are sides of
$D$ with $\gamma_l$ following $\gamma_i$ in the clockwise order;
\item $\gamma_i$ is a radius in a self-folded triangle enclosed by a loop $\gamma_i$, and $\gamma_l$ and $\gamma_j$ are sides of
$D$ with $\gamma_j$ following $\gamma_l$ in the clockwise order;
\end{enumerate}
\item $b^D_{ij} = 0$ otherwise.
\end{itemize}
Then define the matrix $B^T := (b_{ij} ),1\leq i \leq n,1\leq j\leq n$ by $b_{ij} =\sum_D b^D_{ij}$ , where the sum is taken over all triangles in $T$ that are not self-folded. The matrix $B^T$ is a skew-symmetric matrix whose incidence graph is the quiver $Q$ associated to the triangulation.
\end{definition}

\paragraph{The superpotential.} The superpotential $W$ is the sum of two parts. The first one is a sum over all internal triangles of $T$ (that is, triangles having no side on a boundary component). The full subquiver over the three nodes of $Q$ associated with an internal triangle of $T$ has the form
$$
\xymatrix{& \bullet\ar[dl]_\alpha\\
\bullet\ar[rr]_\beta&&\bullet\ar[ul]_\gamma}
$$
Such a triangle contributes a term
$\gamma\beta\alpha$ to $W$. The second part of $W$ is a sum over the regular punctures. Let $\gamma_1,\gamma_2\cdots, \gamma_n$  be the set of arcs ending at the puncture $p$ taken in the clockwise order. The full subquiver of $Q$ over the nodes corresponding to this set of arcs: it is an oriented $n$-cycle. The contribution to $W$ from the puncture $p$ is $\lambda_p$ times the associated $n$-cycle, where $\lambda_p\neq0$ is a complex coefficient. 
\medskip

\paragraph{No regular puncture: gentle algebras.} Suppose $S$ has no regular puncture. Since an arc of $T$ belongs to two triangles (which may be internal or not), 
\begin{equation}\label{ggg}
\text{at a node of $Q$ end (start) at most 2 arrows.}
\end{equation} The superpotential is a sum over the internal triangles
$\sum_i\gamma_i\beta_i\alpha_i$
and the Jacobi relations are of the form
\begin{equation}\label{gggg}
\text{the arrows }\alpha,\ \beta\ \text{arise from the same internal triangle}\ \Longrightarrow\ \alpha\beta=0.\end{equation}
A finite-dimensional algebra whose quiver satisfies \eqref{ggg} and whose relations have the form \eqref{gggg} is called a \emph{gentle} algebra \cite{assem2010gentle}, a special case of a string algebra. Thus, in absence of regular punctures, the Jacobian algebra $J(Q,W)$ is gentle. Indecomposable modules of a gentle algebra may be explicitly constructed in terms of string and band modules \cite{assem2010gentle} (for a review in the physics literature see \cite{cecotti2013categorical}). A gentle algebra is then automatically \emph{tame}; in particular, the BPS particles are either hypermultiplets or vector multiplets, higher spin BPS particles being excluded. 

How to reduce the general case of a class $\mathcal{S}[A_1]$ theory to this gentle situation is explained in ref. \cite{cecotti2015galois}.

\begin{example}
We give here an example of the quiver associated to an ideal triangulation $T$, whose incidence matrix is $B^T$:
\begin{figure}[H]
\begin{tabular}{cc}
\includegraphics[width=0.48\textwidth]{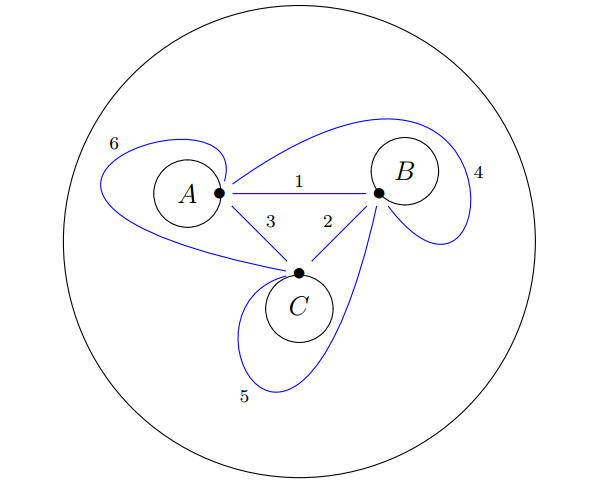} &  \includegraphics[width=0.48\textwidth]{quiver.png}
\end{tabular}
\caption{\textsc{Left}: the sphere with three holes and one marked point per each boundary component. \textsc{Right}: the quiver associated to the triangulation $T$ drawn on the surface on the right. Its adjacency matrix is the matrix $B^T$. }
\label{fig:su23}
\end{figure}
\end{example}

\subsection{Geometric representation of categories}
\label{sec:geomrepcat}
The main reference is \cite{qiu2016decorated}. There is a precise dictionary between curves over a decorated marked surface and the objects in the category $D\Gamma$. 

Then, let $\Delta$ be the set of \emph{decorated points}: to each triangle of $T$, choose a point in the interior of that triangle. With respect to a quadratic differential $\phi$ of section \ref{sec:phys}, the decorated points correspond to the simple zeros of $\phi$. The marked points, on the other hand, correspond to poles of $\phi$ of order $m_i+2$. Let $S_\Delta$ be the surface $S$ with the decorated points. The basic correspondence between geometry and categories is, on the one hand, between objects in $D(\Gamma)$ and curves over $S_{\Delta}$ and on the other hand between morphisms in $D(\Gamma)$ and intersections between curves. Here follow the complete dictionary:
\begin{enumerate}
\item \label{item:1}Recall that an object $S\in D^b (\Gamma)$ is spherical iff $$\mathrm{Hom}_{D^b\Gamma}(S,S[j])\cong k (\delta_{j,0}+\delta_{j,3}).$$ A spherical object in the category $D^b (\Gamma)$ corresponds to a simple\footnote{\ A simple arc does not have self intersections.} closed\footnote{\ A closed arc starts and ends in $\Delta$.} arc (CA) between 2 points in $\Delta$. In particular, simple objects are elements of the dual triangulation and are all spherical. These are some of the BPS hypermultiplets. We can describe these curves as elements of the relative homology with $\Z$-coefficients $H_1(S,\Delta,\mathbb Z)$ over the curve $S$ and the set of points $\Delta$: indeed, the operation of ``sum'' is well defined and it corresponds to the relation that defines the Grothendieck group $K_0(D^b\Gamma)$. We now describe the map that associates a graded module (or equivalently a complex) to a closed arc. Consider a closed arc $\gamma$ that is in minimal position with respect to the triangulation: every time $\gamma$ intersects the triangulation, we add to the complex a simple shifted module $S_i[j]$ and we connect it to the complex with a graded arrow: the grading of the map depends on how the curve and the decorated points are related. In particular, the grading -- corresponding to the Ginzburg algebra grading -- is defined in the following figure\footnote{\ The figure is taken from \cite{qiu2016decorated}.}:
\begin{figure}[H]
\centering
\includegraphics[width=\textwidth]{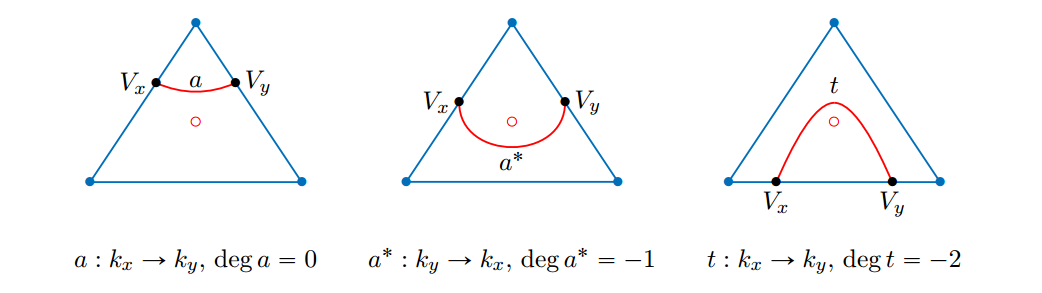}
\label{fig:grad}
\end{figure}
We call $X$ both the map $X: CA(S_\Delta) \to D^b (\Gamma)$ and $X:OA(S_\Delta) \to \mathfrak{Per}\, \Gamma$, where $CA(S_\Delta)$ are the arcs of the decorated surface $S_\Delta$ connecting at most two points in $\Delta$ and $OA(S_\Delta)$ are the arcs connecting marked points. Notice that for the open arc (OA) case, we also have to take into account the tagging at the punctures. In particular, the situation is the following:
\begin{itemize}
\item For an open curve ending on a puncture inside a monogon: if the curve is not tagged, then it intersects only the monogon boundary; if the curve is tagged, then the curve intersect the ray inside the monogon.
\item For an open curve ending on a puncture that is not inside a monogon, then the untagged curve does intersect the curve it would intersect as if it were untagged; if the curve is tagged, then it is as if the curve made a little loop around the puncture and so changes the intersection.
\end{itemize}
\begin{example}
\label{ex:primo}
Let us consider the quiver $A_2$ again. Then, consider the curves $\gamma_1$, $\gamma_2$ and $\gamma_3$ as in the picture. 
\begin{figure}[H]
\centering
\resizebox{0.60\textwidth}{!}{\xygraph{
!{<0cm,0cm>;<0.5cm,0cm>:<0cm,0.5cm>::}
!{(0,10) }*+{\bullet}="b1"
!{(8,4) }*+{\bullet}="b3"
!{(-8,4) }*+{\bullet}="b2"
!{(-5,-5) }*+{\bullet}="b4"
!{(5,-5) }*+{\bullet}="b5"
!{(0,0) }*+{\bullet_{a}}="a"
!{(-3,0) }*+{}="al"
!{(0,5) }*+{}="au"
!{(0,-5) }*+{}="ad"
!{(5,5) }*+{\bullet_{b}}="b"
!{(-5,5) }*+{\bullet_{c}}="c"
"b1" - "b2" "b2" - "b4" "b4" - "b5" "b5" - "b3" "b3" - "b1" "b1"-"b4"_{1} "b1" - "b5"^{2}
"b"-@[blue]"c" _{\gamma_1} 
"b" -@`{"ad","al","au"}@[red] "b"^{\gamma_3}
"b"-@/^3cm/@[green]"c" _{\gamma_2}
}}
\label{fig:A2123}
\end{figure}
By applying the rules above we get:
  $$X(\gamma_1): \ \ \ S_1 \overset{a=1}{\to}S_2 \ \cong \ k\overset{1}{\to}k$$
  $$X(\gamma_2): \ \ \ S_1[-1] \overset{a^*=1}{\leftarrow}S_2 \ \cong \ k[-1]\overset{1}{\leftarrow}k$$
  $$X(\gamma_3): \ \ \ S_2 \overset{t_1=1}{\to}S_2[-2] \ \cong  \ 0 \overset{0}{\to}k\oplus k[-2]\CircleArrowleft 1$$
\end{example}
This map also works for the graded modules corresponding to closed loops: both those starting at the zeros of the quadratic differential and those not. As we will see in the example \ref{ex:su2} of the Kronecker quiver, only certain loops are in the category $D^b \Gamma$. All other possible loops belong to $\mathfrak{Per}\, \Gamma$. Moreover, even in the case with punctures, the algorithm to get (graded) modules form the curves is the same (thanks to proposition 4.4 of \cite{qiu2013cluster}).
\item Since every simple object is a spherical objects in $D^b \Gamma$ -- in particular, since $D^b \Gamma$ is 3-CY, the simple objects are 3-spherical -- we can define the Thomas-Seidel twist $T_S$ associated to such a spherical object $S$ by the following triangle
$$\mathrm{Hom}_{D^b \Gamma}^\bullet(X,S)\otimes S \to X \to T_S(X) \to.$$
These twists are autoequivalences of $D\Gamma$. Geometrically, these spherical twists correspond to \emph{braid twists} associated to the simple closed arc $\gamma_S$: let $BT(S_\Delta)$ denote the full group generated by all the braid twists. The action of the braid twist is like in figure
\begin{figure}[H]
\centering
\includegraphics[width=0.6\textwidth]{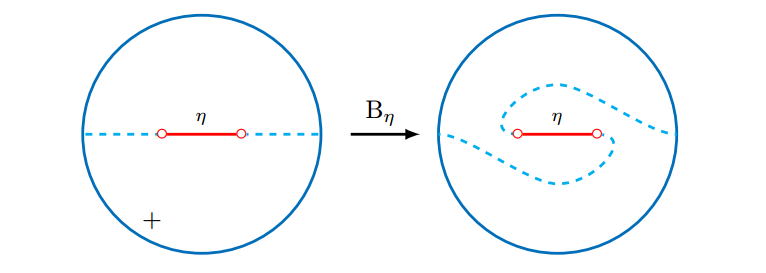}
\label{fig:braid}
\end{figure}
Moreover, being $T^*$ the dual triangulation of the triangulation $T$, $$CA(S_{\Delta})=BT(S_\Delta)\cdot T^*$$ as shown in \cite{qiu2016decorated}.
\item  In general, braid twists of $S_\Delta$ correspond to spherical twists of $D^b \Gamma$: $BT(T_{S_{\Delta}})=\mathrm{Sph}(D^b (\Gamma))$. Moreover, the quiver representing the braid relations is exactly the quiver $Q$: to each vertex we associate a twist $T_{S_i}$ and to each single arrow $i \to j$ a braid relation $T_{S_i}T_{S_j}T_{S_i}=T_{S_j}T_{S_i}T_{S_j}.$ If there is no arrow between $i$ and $j$, then $[T_{S_i},T_{S_j}]=0.$
\item \label{item:4}Rigid and reachable objects in $\mathfrak{Per}\,\Gamma$ correspond to the simple open arcs, i.e.\! simple curves connecting marked points. The other objects in $\mathfrak{Per}\,\Gamma$ correspond to generic arcs: both those arcs connecting two different punctures or marked points and closed loops encircling decorations or boundaries or punctures. We can describe these curves as elements of the relative homology $H_1(S_\Delta,M,\mathbb Z)$ over the curve $S_\Delta$ (where the points in $\Delta$ are topological points in $S_\Delta$) and the set of marked points $M$: indeed, the operation of ``sum'' is well defined and it corresponds to the relation that defines the Grothendieck group $K_0(\mathfrak{Per}\,\Gamma)$. 
\item Let $T$ be the triangulation of the surface. The arcs of the triangulations are associated to the $\Gamma e_i$ objects in $\mathfrak{Per}\,\Gamma$ and the elements of the dual triangulations are the simple objects in $D^b\Gamma$. This is the geometrical version of the simple-projective duality:
$$\dim \mathrm{Hom}^j(\Gamma e_i,S_k[l])=\delta_{lj}\,\delta_{ki}.$$
The choice of a heart in $D^b\Gamma$ corresponds to the choice of the simple objects; thus, via the simple-projective duality it also corresponds to the choice of a triangulation $T$. The relation between the Grothendieck groups of $D^b\Gamma$ and $\mathfrak{Per}\, \Gamma$ is via the Euler form (which corresponds to the intersection form, as pointed out in item \ref{item:int}) and it corresponds to Poincar\'e duality of the relative homology groups.

\item Flips of the triangulation (forward and backward) correspond to right and left mutations $\mu_i^\pm$. Two different flips of the same arc are connected by a braid twist assocaited to that simple closed arc: $T_{S_i}=\mu^+_i(\mu_i^-)^{-1}$. 
\item \label{item:int}$\mathrm{Hom}$ spaces correspond to intersection numbers:\footnote{\ $CA=$closed arc, $OA=$open arc.} the full proof of the following facts can be found in \cite{qiu2014decorated}:
$$\dim \mathrm{Hom}(X(CA),X(CA))=2\,  \mathrm{Int}(CA,CA)$$
$$\dim(X(OA),X(CA))= \mathrm{Int}(OA,CA)$$
The intersection numbers between arcs in $S_\Delta$ are defined as follows:
\begin{itemize} 
\item For an open arc $\gamma$ and any arc $\eta$, their intersection number is the geometric
intersection number in $S_\Delta -M$:
$$\mathrm{Int}(\gamma, \eta) = \min\big\{|\gamma^\prime\cap \eta^\prime \cap  (S_\Delta - M)||\gamma^\prime \cong \gamma, \eta^\prime \cong \eta\big\}.$$
\item For two closed arcs $\alpha, \beta$ in $CA(S_\Delta)$, their intersection number is an half integer
in $\frac{1}{2}\mathbb Z$ and defined as follows:
$$\mathrm{Int}(\alpha,\beta) = \frac{1}{2}\,\mathrm{Int}_\Delta (\alpha,\beta) + \mathrm{Int}_{S-\Delta}(\alpha,\beta),$$
where
$$\mathrm{Int}_{S_\Delta-\Delta}(\alpha,\beta) = \min\!\big\{|\alpha^\prime \cap \beta^\prime \cap  S_\Delta -\Delta|\; \big|\; \alpha \cong \alpha^\prime, \beta \cong \beta^\prime\big\}$$
and
$$\mathrm{Int}_\Delta(\alpha,\beta) = \sum_{Z\in \Delta}\Big|\big\{t\; \big|\; \alpha(t) = Z\big\}\Big| \cdot \Big|\big\{r \;\big|\; \beta(r) = Z\big\}\Big|.$$
Let $T_0$ be a triangulation and $\eta$ any arc; it is straightforward to see $\mathrm{Int}(T_0,\eta) \geq 2$ for a loop, and the equality holds if and
only if $\eta$ is contained within two triangles of $T_0$ (in this case, $\eta$ encircles exactly one decorating
point).
\end{itemize}

\begin{example}[Sphere with three holes and one marked point per each boundary component.]
In this particular case (the following results do not hold in general), the matrix corresponding to the bilinear form $\mathrm{Int}(-,-)$ can be obtained by taking the Cartan matrix of the quiver with potential $(Q,W)$, i.e.\! the matrix whose columns are the dimensions of the projective modules, inverting and transposing it. 
In particular, the incidence matrix for a sphere with three boundary components with $m_i=1, \forall i \in \{1,2,3\}$ -- over the basis of simples (corresponding to the edges of the dual triangulation) -- is 
\begin{center}
\begin{tabular}{cc}
\resizebox{0.35\textwidth}{!}{\xygraph{
!{<0cm,0cm>;<0.5cm,0cm>:<0cm,0.5cm>::}
!{(0,0) }*+{\bullet}="b1"
!{(5,0) }*+{\bullet}="b3"
!{(2.5,-2.5) }*+{\bullet}="b2"
!{(-1,0) }*+{}*\cir<15pt>{}="b4"
!{(5.8,0.7) }*+{}*\cir<15pt>{}="b5"
!{(2.5,-3.6) }*+{}*\cir<15pt>{}="b6"
!{(2.5,-1.5) }*+{}*\cir<105pt>{}="center"
"b1" -@[blue] "b2"^{3} "b2" -@[blue] "b3"^{2} "b3" -@[blue] "b1"_{1}
"b1" -@`{(10,7.5),(9,-6)}@[blue] "b3"^{4}
"b3" -@`{(2.5,-12),(-2,-3)}@[blue] "b2"^{5}
"b2"-@`{(-11,0),(2,4)}@[blue]"b1" ^{6}
}}&
$\mathrm{Int}=\left(
\begin{array}{cccccc}
 \frac{1}{2} & -\frac{1}{2} & \frac{1}{2} & 0 & 0 & 0 \\
 \frac{1}{2} & \frac{1}{2} & -\frac{1}{2} & 0 & 0 & 0 \\
 -\frac{1}{2} & \frac{1}{2} & \frac{1}{2} & 0 & 0 & 0 \\
 -1 & 0 & 0 & \frac{1}{2} & \frac{1}{2} & -\frac{1}{2} \\
 0 & -1 & 0 & -\frac{1}{2} & \frac{1}{2} & \frac{1}{2} \\
 0 & 0 & -1 & \frac{1}{2} & -\frac{1}{2} & \frac{1}{2} \\
\end{array}
\right)$
\end{tabular}
\end{center}
The Euler characteristic of $D^b (\Gamma)$, as a bilinear form defined in \ref{eq:euler}, on the other hand, is an antisymmetric integral matrix that is the antisymmetric part of $\mathrm{Int}(-,-)$:
$$\chi(-,-)=\left(
\begin{array}{cccccc}
 0 & 1 & -1 & -1 & 0 & 0 \\
 -1 & 0 & 1 & 0 & -1 & 0 \\
 1 & -1 & 0 & 0 & 0 & -1 \\
 1 & 0 & 0 & 0 & -1 & 1 \\
 0 & 1 & 0 & 1 & 0 & -1 \\
 0 & 0 & 1 & -1 & 1 & 0 \\
\end{array}
\right).
$$
\end{example}
Notice that the skew-symmetric matrix we have just found corresponds to the matrix $B^T$ associated to the ideal triangulation of figure \ref{fig:su23}. 
\item Relations amongst the exchange graphs (EG) of the surface $S$ (to each vertex of this graph we associate a triangulation and to each edge of the graph a flip of the triangulation) and the cluster exchange graph (CEG) of the cluster algebra (to each vertex of the graph we associate a cluster and to each edge a mutation):
$$EG(S) = CEG(\Gamma)$$
$$EG(D^b (\Gamma))/\mathrm{Sph}(D^b (\Gamma))=CEG(\Gamma)$$
\item A distinguished triangle in $D^b (\Gamma)$ corresponds to a \emph{contractible} triangle in $S_\Delta$ whose edges are 3 closed arcs $(\alpha,\beta,\eta)$ as in the figure, such that the categorical triangle is $X(\alpha) \to X(\eta) \to X(\beta) \to$.
\begin{figure}[H]
\centering
\includegraphics[width=0.4\textwidth]{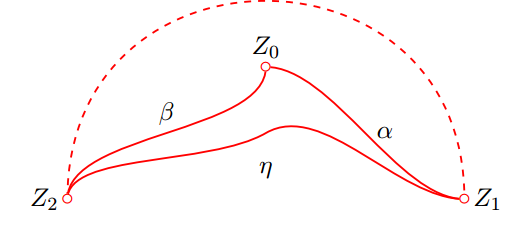}
\label{fig:triangle}
\end{figure}
Exploiting the group structure of the homology group $H_1(S,\Delta,\mathbb Z)$, we see that we have the following relation:
$$[\alpha]-[\eta]+[\beta]=0;$$
this is the defining relation of the Grothendieck group $K_0(D^b\Gamma)$.
\item The triangulated structure of $\mathfrak{Per}\,\Gamma$ is less easy to represent geometrically. Before proceeding with an example, we define the left and right mutations in $\mathfrak{Per}\,\Gamma$ starting from the \emph{silting set}. 
\begin{definition}
A \emph{silting set} $\mathbb P$ in a category $D$ is an $\mathrm{Ext}^{>0}$-configuration, i.e.\! a maximal collection of non-isomorphic indecomposables
such that $\mathrm{Ext}^i(P, T) = 0$ for any $P, T \in \mathbb P$ and integer $i > 0$. The forward mutation $\mu^-_P$ at an element $P \in \mathbb  P$ is another silting set $\mathbb P^-_P$, obtained from $\mathbb P$ by replacing $P$ with
\begin{equation}P^- = \text{Cone}\!\left(P \to \bigoplus_{T \in \mathbb P\\ \{P\} }D\mathrm{Hom}_{\mathrm{irr}}(P, T)\otimes T \right),
\label{eq:p-}
\end{equation}
where $\mathrm{Hom}_\mathrm{irr}(X, Y )$ is the space of irreducible maps $X \to Y$ , in the additive subcategory
$\mathsf{add}\,\bigoplus_{T \in \mathbb P} T$ of $D$. The backward mutation $\mu^+_P$ at an element $P \in \mathbb P$ is another silting
set $\mathbb P^+_P$, obtained from $\mathbb P$ by replacing $P$ with
\begin{equation}
P^+ = \mathrm{Cone}\!\left(
\bigoplus_{T \in \mathbb P\\ \{P\}} \mathrm{Hom}_\mathrm{irr}(T, P) \otimes T \to P\right)\![-1]
\label{eq:p+}
\end{equation}
\end{definition}
\begin{remark}
Notice that equations (\ref{eq:m-}) and (\ref{eq:m+}) are exactly the same as (\ref{eq:p-}) and (\ref{eq:p+}): the notation of the latter is more straightforward for the next computations. 
\end{remark}
We give now an example to show how one can get the triangulated structure of $per \Gamma$ and we relate it to the group structure of $H_1(S_\Delta,M,\Z)$.
\begin{example}[$A_2$ example]
The silting set for the quiver $A_2: \bullet_1 \to \bullet_2$ is $\mathbb P=\{\Gamma e_1,\Gamma e_2\}$ which we denote as $\{\Gamma_1,\Gamma_2\}$. We apply the left mutation corresponding to a forward flip: it gives the following triangle
$$\Gamma_2 \to \Gamma_1 \to \Gamma_2^-\to .$$
The element $\Gamma_2^-$ is an infinite complex of the form $S_1 \overset{t_1}{\to} \Gamma_1[-2]$. At a geometrical level it corresponds to the red curve in this figure:
\begin{figure}[H]
\centering
\resizebox{0.60\textwidth}{!}{\xygraph{
!{<0cm,0cm>;<0.5cm,0cm>:<0cm,0.5cm>::}
!{(0,10) }*+{\bullet}="b1"
!{(8,4) }*+{\bullet}="b3"
!{(-8,4) }*+{\bullet}="b2"
!{(-5,-5) }*+{\bullet}="b4"
!{(5,-5) }*+{\bullet}="b5"
!{(0,0) }*+{\bullet_{a}}="a"
!{(-3,0) }*+{}="al"
!{(0,5) }*+{}="au"
!{(0,-5) }*+{}="ad"
!{(5,5) }*+{\bullet_{b}}="b"
!{(-5,5) }*+{\bullet_{c}}="c"
"b1" - "b2" "b2" - "b4" "b4" - "b5" "b5" - "b3" "b3" - "b1" "b1"-"b4"_{\Gamma_1} "b1" - "b5"^{\Gamma_2}
"b4" -@`{"ad","al","au"}@[red] "b3"^{\Gamma_2^-}
}}
\label{fig:A2gamma2}
\end{figure}
As one can see, these three open arcs do not seem to be geometrically easily related. The next step, which is fundamental for consistency of the geometric representation, is to consider the following triangle:
$$\Gamma_1[-2]\to \Gamma_2^- \to S_1 \to$$
This triangle is exact and moreover $S_1 \in D^b\Gamma$. This implies that in the cluster category, both $\Gamma_2^-$ and $\Gamma_1[-2]$ map to the same curve. We can exploit the geometric effect of the shift (see item \ref{item:shift}) to compute $\Gamma_1[-2]$: it is the green curve in the next figure
\begin{figure}[H]
\centering
\includegraphics[width=0.9\textwidth]{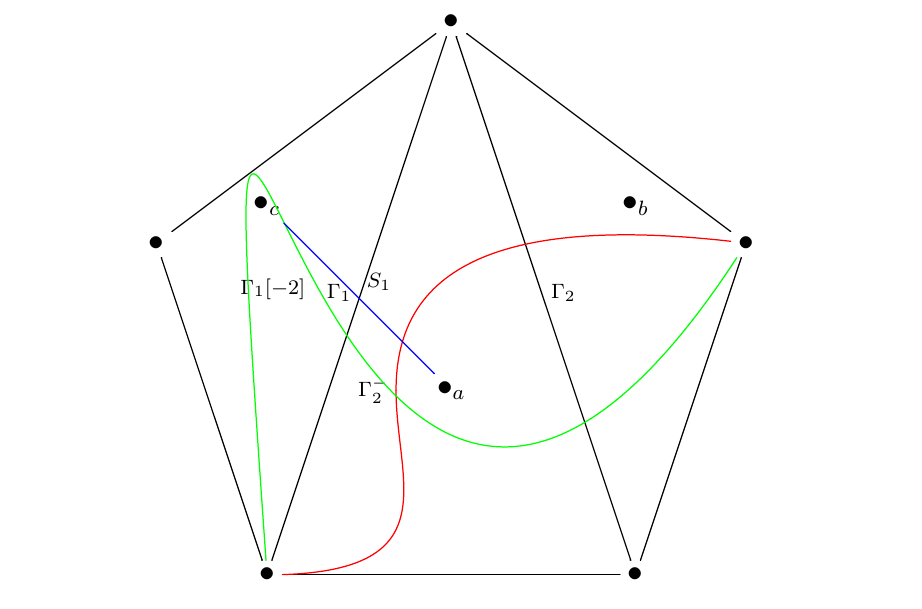}
\caption{The geometric representation of the distinguished triangle $\Gamma_1[-2]\to \Gamma_2^- \to S_1 \to$ in the perfect category $\mathfrak{Per}\, \Gamma$.}
\label{fig:perf}
\end{figure}
When we map to the cluster category via the forgetful functor (see item \ref{item:quot}), we have that both $\Gamma_1[-2]$ and $\Gamma_2^-$ are sent to the same object: $S_2$. Indeed, the corresponding cluster category is made of 5 indecomposables which form the following periodic AR diagram
\begin{displaymath}
    \xymatrix{ & \Gamma_1 \ar[dr]&  &  P_1\ar[dr] &  \\
            \Gamma_2  \ar[ur] &  & S_2 \ar[ur]&  & S_1 }
\end{displaymath}
where we can read the corresponding triangles in the cluster category. The geometric picture is
\begin{figure}[H]
\centering
\resizebox{0.50\textwidth}{!}{\xygraph{
!{<0cm,0cm>;<0.5cm,0cm>:<0cm,0.5cm>::}
!{(0,10) }*+{\bullet}="b1"
!{(8,4) }*+{\bullet}="b3"
!{(-8,4) }*+{\bullet}="b2"
!{(-5,-5) }*+{\bullet}="b4"
!{(5,-5) }*+{\bullet}="b5"
"b1" - "b2" "b2" - "b4" "b4" - "b5" "b5" - "b3" "b3" - "b1" "b1"-@[blue]"b4"_{\Gamma_1} "b1" -@[blue] "b5"^{\Gamma_2} "b5" -@[blue] "b2"^{S_1} "b2" -@[blue] "b3"^{P_1} "b3" -@[blue] "b4"^{S_2}}}
\label{fig:A2gamma222}
\end{figure}
\end{example}
We can now generalize the above results by stating that the triangulated structure is consistent with the group structure of the relative homology group $H_1(S_\Delta,M,\mathbb Z)$ paired with the relative homology of $H_1(S,\Delta,\mathbb Z)$: indeed, we see that in $H_1(S_\Delta,M,\mathbb Z)$ -- after choosing opposite direction for the red and green path of figure \ref{fig:perf}-- we have
$$[\Gamma_1[-2]]-[\Gamma_2^-]+[S_1]=0.$$
The relations defining the Grothendieck $K_0(\mathfrak{Per}\,\Gamma)$, such as
$$[\Gamma_2]-[\Gamma_1]+[\Gamma_2^-]=0,$$
are less obvious from a homological viewpoint: there is no other way but compute them explicitly when needed.
\item \label{item:quot} The Amiot quotient $\mathfrak{Per}\,\Gamma/ D^b (\Gamma)$ \cite{amiot2011generalized} -- through which the cluster category is defined -- corresponds to the forgetful map $F:S_\Delta \to S$. For a short reminder of triangulated quotients, see section \ref{sec:cat}. In particular, we recover easily the results of \cite{brustle2011cluster}: the indecomposable objects of the cluster categories are string modules or band modules. Geometrically a string is an open arc and the procedure to associated a module to it is the same as the one described in item \ref{item:1} for closed arcs. So indeed, since open curves are indecomposable objects of $\mathfrak{Per}\,\Gamma$, as pointed out in item \ref{item:4}, the following diagram -- at least for string modules -- commutes:
\begin{displaymath}
    \xymatrix{ 
  \mathcal OA(S_\Delta) \ar[r]^F \ar[d]^X & OA(S) \ar[d]^{X} \\
\mathfrak{Per}\,\Gamma \ar[r]^\pi   & \mathcal C(\Gamma)
              }
\end{displaymath}
We expect no difference in the case of band modules (which correspond to families of loops).
 \item \label{item:shift} In the cluster category $\mathcal C(\Gamma)$ and in the perfect derived category $\mathfrak{Per}\, \Gamma$, the shift [1] corresponds to a global anticlockwise rotation of all the marked points on each boundary component. For punctures, the action of $[1]$ corresponds to a change in the tagging. In $\mathcal{C}(\Gamma)$, it is equivalent to the operation $\tau$ (the AR translation), as defined in \cite{brustle2011cluster}. In particular we see that in presence of only regular punctures, $[2]$ flips twice the tagging getting back to the original situation, that is, in this case the cluster category is 2-periodic, as expected on physical grounds.
\end{enumerate}
Let us consider here a simple example that will allow us to clarify some aspects.
\begin{example}[$A_2$ again]
\label{ex:a2}
Let us consider the following curves over a disk with 5 marked points on the boundary and no punctures. The triangulation is made of the black lines $1$ and $2$ and the boundary arcs:
\begin{figure}[H]
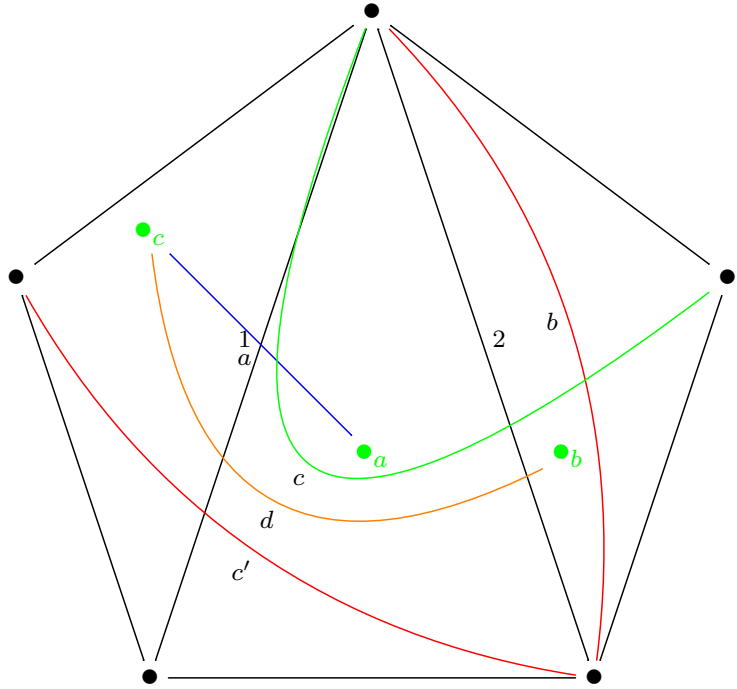

\centering
\resizebox{0.60\textwidth}{!}{\xygraph{
!{<0cm,0cm>;<0.5cm,0cm>:<0cm,0.5cm>::}
!{(0,10) }*+{\bullet}="b1"
!{(8,4) }*+{\bullet}="b3"
!{(-8,4) }*+{\bullet}="b2"
!{(-5,-5) }*+{\bullet}="b4"
!{(5,-5) }*+{\bullet}="b5"
!{(0,0) }*+[green]{\bullet_{a}}="a"
!{(-3,0) }*+{}="al"
!{(0,5) }*+{}="au"
!{(0,-5) }*+{}="ad"
!{(4.4,0) }*+[green]{\bullet_{b}}="b"
!{(-5,5) }*+[green]{\bullet_{c}}="c"
"b1" - "b2" "b2" - "b4" "b4" - "b5" "b5" - "b3" "b3" - "b1" "b1"-"b4"_{1} "b1" - "b5"^{2}
"c"-@[blue]"a" _{a} 
"b5" -@/^1cm/@[red] "b2"^{c'}
"b5" -@/_1cm/@[red] "b1"^{b}
"b3" -@/^4.5cm/@[green] "b1"^{c}
"b" -@/^2cm/@[orange] "c"^{d}
}}
\caption{This is the surface corresponding to the quiver $A_2:1 \to 2$.}
\label{fig:A2big}
\end{figure}
The curves $a$ and $d$ correspond to graded modules in $D^b \Gamma$ as shown in the example \ref{ex:primo}; indeed, the curve $d$ corresponds to the graded module $S_1 \leftarrow S_2[-1]$ in $D^b \Gamma$.
The red curves $b$ and $c'$, on the other hand, cannot be associated to any closed curve, since they cannot intersect any closed curve in minimal position: they are elements only of $\mathfrak{Per}\,\Gamma$ and not of $D^b \Gamma$. When we take the Amiot quotient, the green dots disappear and so do the curves $a$ and $d$. Moreover, the curve $c$ is homotopic to a boundary arc and $b$ is homotopic to $2$. Thus the quotient does what we expect: only the curves in $\mathfrak{Per}\,\Gamma$ that are not in $D^b \Gamma$ do not vanish. 
The intersection form in this example is 
$$\mathrm{Int}=
\left(
\begin{array}{cc}
 1 & -1 \\
 0 & 1 \\
\end{array}
\right)
$$
and thus the Euler characteristic is 
$$\chi=
\left(
\begin{array}{cc}
 0 & 1 \\
 -1 & 0 \\
\end{array}
\right)
$$
We can thus construct the Thomas-Seidel twists associated to the simple modules:
$$T_{S_1}=Id-\ket{S_1}\bra{S_1}\cdot\chi=\left(
\begin{array}{cc}
 1 & -1 \\
 0 & 1 \\
\end{array}
\right)$$
$$T_{S_2}=Id-\ket{S_2}\bra{S_2}\cdot\chi=\left(
\begin{array}{cc}
 1 & 0 \\
 1 & 1 \\
\end{array}
\right)$$
And we can explicitly check the braid relation 
$$T_{S_1}\cdot T_{S_2}\cdot T_{S_1}=T_{S_2}\cdot 
T_{S_1}\cdot T_{S_2}$$ 
both on $K_0(D^b \Gamma)$
 and geometrically. Moreover it is clear from the matrix representation that $T_{S_1}$ and $T_{S_2}$ generate $SL_2(\mathbb Z)$. We can also act with these twists on the graded modules:
$$ T_{S_1}\cdot X(S_2)=X((-1,1))$$
The graded module whose dimension vector is $(-1,1)$ is $k[-1] \overset{a^*=1}{\leftarrow }k$. This is consistent with the geometric picture, as one can verify:
\begin{figure}[H]
\centering
\resizebox{0.60\textwidth}{!}{\xygraph{
!{<0cm,0cm>;<0.5cm,0cm>:<0cm,0.5cm>::}
!{(0,10) }*+{\bullet}="b1"
!{(8,4) }*+{\bullet}="b3"
!{(-8,4) }*+{\bullet}="b2"
!{(-5,-5) }*+{\bullet}="b4"
!{(5,-5) }*+{\bullet}="b5"
!{(0,0) }*+[green]{\bullet_{a}}="a"
!{(-3,0) }*+{}="al"
!{(0,5) }*+{}="au"
!{(0,-5) }*+{}="ad"
!{(5,5) }*+[green]{\bullet_{b}}="b"
!{(-5,5) }*+[green]{\bullet_{c}}="c"
"b1" - "b2" "b2" - "b4" "b4" - "b5" "b5" - "b3" "b3" - "b1" "b1"-"b4" "b1" - "b5"
"b"-@[green]"a" _{S_1} 
"c"-@[orange]"a" _{S_2}
"b" -@/^4cm/@[red] "c"^{T_{S_1}S_2}
}}
\label{fig:A2braid}
\end{figure}

\end{example}
\begin{example}[Kronecker quiver]
\label{ex:su2}
The surface corresponding to the Kronecker quiver is an annulus with one marked point on each boundary component. This theory corresponds to a pure $SU(2)$ SYM. Thus, we shall find closed curves corresponding to BPS vector bosons. They must also appear a one-parameter family. In the following figure, the green loops around the two points represent a band module, the black dots the marked points; the black lines are the flow lines associated to the quadratic differential $\phi$ described in section \ref{sec:phys}.
\begin{figure}[H]
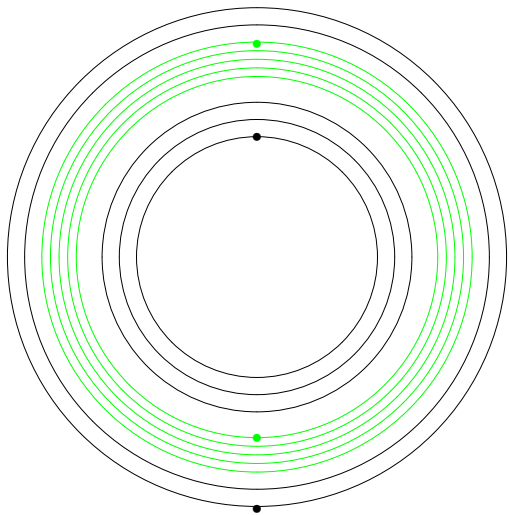

\centering
\resizebox{0.40\textwidth}{!}{\xygraph{
!{<0cm,0cm>;<0.5cm,0cm>:<0cm,0.5cm>::}
!{(0,0) }*+{}*\cir<70pt>{}
!{(0,0) }*+{}*\cir<80pt>{}
!{(0,0) }*+{}*\cir<90pt>{}
!{(0,0) }*+{}*[green]\cir<105pt>{}
!{(0,0) }*+{}*[green]\cir<110pt>{}
!{(0,0) }*+{}*[green]\cir<115pt>{}
!{(0,0) }*+{}*[green]\cir<120pt>{}
!{(0,0) }*+{}*[green]\cir<125pt>{}
!{(0,0) }*+{}*\cir<135pt>{}
!{(0,0) }*+{}*\cir<145pt>{}
!{(0,8.7) }*+[green]{\bullet}="b1"
!{(0,-7.4) }*+[green]{\bullet}="b3"
!{(0,4.9) }*+{\bullet}="b1"
!{(0,-10.3) }*+{\bullet}="b3"
}}
\caption{The 1-parameter family of curves corresponding to the vector multiplet of charge $(1,1).$}
\label{fig:puresu2}
\end{figure}
The module corresponding to (one of) the green curve can be computed via the map $X: a \mapsto X(a) \in D^b \Gamma$, starting at a generic point along the curve. We get
$$X(a): \ \ S_1 \overset{a=1,b=\lambda}{\Longrightarrow} S_2\ \cong\ k \overset{a=1,b=\lambda}{\Longrightarrow} k .$$
This module is stable (see section \ref{sec:stab}) since it is equivalent to the regular module in the homogeneous tube of $\text{\textsf{mod}-}k\,\mathsf{Kr}$.
\end{example}
\subsection{Summary}
We give here a sketchy summary of what we have written so far; recall that, given a quadratic differential $\phi(z) dz\otimes dz$ there are three kind of markings: the zeros of $\phi$ (decorations), the simple poles of $\sqrt{\phi}$ (punctures) and the irregular singularities of $\phi$, which generate the marked points on the boundary segments.
\begin{itemize}
\item BPS vector multiplets correspond to loops, not crossing the separating arcs of the flow of $\phi$, for $\theta=\arg Z(BPS)=\theta_c$. They belong to the category $D^b \Gamma$ via the maps $X$.
\item BPS hypermultiplets correspond to arcs connecting two zeros of $\phi$. They belong to the category $D^b \Gamma$ via the map $X$ and can be identified with elements in the relative homology $H_1(S,\Delta,\Z)$, where $\Delta$ are the zeros of $\phi$.
\item Surface operators corresponds to punctures and marked points.
\item The objects in the perfect derived category $\mathfrak{Per}\,\Gamma$ are the screening states created by line operators acting on the vacuum. They correspond to arcs connecting marked points and punctures and thus belong to $H_1(S_{\Delta},M,\Z)$.
\item UV line operators correspond to arcs over $S$ connecting marked points and punctures (but not zeros of $\phi$).  They belong to the cluster category $\mathcal C(\Gamma)$ and can also be identified with elements of the relative homology $H_1(S,M,\mathbb Z)$, where $M$ is the set of markings of the surface $S$. 
\end{itemize}
With this dictionary in mind, we can now exploit the categorical language to compute physical quantities (such as vacuum expectation values of UV line operators). 
\begin{remark}
A generalization of these concepts, in particular towards ideal webs and dimer models can be found in \cite{goncharov2016ideal}. There it is argued that the 3-CY category $D^b\Gamma$ is the physical BPS states category, in accordance with our analysis. Moreover, in the case in which we have a geometrical interpretation via bipartite graphs, the mapping class group of the punctured surface, is a subgroup of the full S-duality group.
Here we give a simple example.
\begin{example}[Pure $SU(3)$]
The pure $SU(3)$ theory can be described as a $S[A_2]$ Gaiotto theory over a cylinder with one full punctures per each boundary \cite{xie2012network}. The mapping class group is generated by a single Dehn twist around the cylinder. It acts on the bipartite graph associated to the $SU(3)$ theory and it must be a subgroup of the full S-duality group. It is isomorphic to $\mathbb Z$. Since we also have that pure $SU(3)$ theory can be described by the quiver
$$\mathsf{Kr} \boxtimes A_2,$$
We explicitly find -- using he techniques of \cite{caorsi2016homological} -- that a cluster automorphism is given by $\tau_\mathsf{Kr}\otimes \tau_{A_2}$, which generates a free group, thus isomorphic to $\mathbb Z$.
\end{example}
\end{remark}

\section{Cluster characters and line operators}
\label{sec:linop}
\subsection{Quick review of line operators}
\label{sec:lindef}
The main reference for this part is \cite{cordova2013line}. We are going to study, in the following section, the algebra of line operators: we shall discover that this algebra is closely related to the cluster algebra of Fomin and Zelevinski \cite{fomin2002cluster}. Recall that an IR line operator (also called \emph{framed BPS state} \cite{gaiotto2013framed}) is characterize by a central charge $\zeta$ and a charge $\gamma$. Similarly for a UV line operator. The starting point is to consider the RG flow:
$$
\begin{gathered}
RG(\cdot,\alpha,\zeta):\{\text{UV \ line \ operators}\}\to \{\text{IR \ line \ operators}\}\\
L(\alpha,\zeta) \mapsto \sum_{\gamma\in \Gamma}\bar \Omega(\alpha,\zeta,\gamma,u,y)\,L(\gamma,\zeta),
\end{gathered}
$$
where $L(\alpha,\zeta)$ is a supersymmetric UV line operator of UV charge $\alpha$ (see \S \ref{sec:uvcharge}). We can think of it as a supersymmetric Wilson line operator:
$$L(\alpha,\zeta):=\exp\!\left(i \alpha \int_\mathrm{time} A+\frac{1}{2}(\zeta^{-1}\phi +\zeta \bar \phi)\right).$$
where $A$ is the gauge connection and $\phi$ and $\bar \phi$ are the supersymmetric partners. The idea is that \emph{cluster characters} provide the coefficients $\bar \Omega(\alpha,\zeta,\gamma,u,y)$; moreover, the OPE's of line operators can be identified with the \emph{cluster exchange relations}. The physical definition of $\bar \Omega(\alpha,\zeta,\gamma,u,y)$ as supersymmetric index is the following. Define the Hilbert space of our system with a line operator in it polarizing the vacuum: $H_{L,\zeta,u}=\bigoplus_{\gamma \in \Gamma_u}H_{L,\zeta,u,\gamma},$ where $\Gamma_u$ is the charge lattice and $u$ a point in the Coulomb branch. When we restrict only to BPS states we have $H^{BPS}_{L,\zeta,u}$. We now define the following index (i.e. a number that counts the line operators):
$$\bar \Omega(\alpha,\zeta,\gamma,u,y)=\mathrm{Tr}_{H_{L,\zeta,u,\gamma}}(y^{2J_3}(-y)^{2I_3}),$$
where the $I$ and $J$ operators are the Cartan generators associated to the unbroken Lorentz symmetry $SO(3)$ and unbroken R-symmetry $SU(2)_R$ by the presence of the line operator which moves along a straight line in the time direction. In particular, if $y=1$:
$$\bar \Omega(\alpha,\zeta,\gamma,u,1)=\mathrm{Tr}_{H_{L,\zeta,u,\gamma}}(1^{2J_3}(-1)^{2I_3})=\sum_m1^{2m}(-1)^0=\dim H_{L,\zeta,u,\gamma}^{BPS}.$$
This index corresponds to the Poincar\'e polynomial of the quiver Grassmannian $Gr_\gamma(\alpha)$, where we interpret the UV line operator $L$ of charge $\alpha$ as the quiver representation of which we compute the cluster character (see \S\ref{sec:char} for more details).
\subsubsection{Algebra of UV line operators}
Let the OPE's of UV line operators be defined as follows
$$L(\alpha, \zeta)L(\alpha',\zeta)=\sum_\beta c(\alpha,\alpha',\beta)L(\beta,\zeta).$$
From now on, let us fix the generic point of the Coulomb branch $u$. We also define the generating functions for the indexes $ \bar \Omega(L, \gamma)= \bar \Omega (L,\gamma,u=\text{fixed},y=1)$:
$$F(L)=\sum_\gamma \bar \Omega (L,\gamma)X_\gamma,$$
where the formal variable $X_\gamma$ is such that $X_\gamma X_{\gamma'}=X_{\gamma+\gamma'}$. One can check that $F(LL')=F(L)F(L')$. This equality gives a recursive formula to compute $F(LL')$. Furthermore, we can study the wall-crossing of UV line operators via the formula of KS \cite{gaiotto2013framed}:
$$F^+(L_{\gamma_c})=F^-(L_{\gamma_c})\prod_{\gamma}\prod_{m=-M_\gamma}^{M_\gamma}(1+(-1)^mX_\gamma)^{|\vev{\gamma,\gamma_c}|a_{m,\gamma}}$$
where $\gamma_c$ is the charge of the wall we are crossing, and $M_\gamma$ is the maximal value of the operator $\mathcal J_3=J_2+I_3$ and the $a_{m,\gamma}$ are the coefficients of the index:
$$\bar \Omega(\alpha, \zeta,\gamma,u,y)=\sum_{m=-M_\gamma}^{M_\gamma} a_{m,\gamma}y^m.$$
We can transfer this formula on the newly defined variables to implement the wall-crossing more efficiently: 
\begin{equation}
X_\gamma'=X_\gamma \prod_{\gamma}\prod_{m=-M_\gamma}^{M_\gamma}(1+(-1)^mX_\gamma)^{\vev{\gamma,\gamma_c}a_{m,\gamma}}.
\label{eq:linopvar}
\end{equation}
Moreover, the transformation of the indices $\bar \Omega(L, \gamma)$ is an
automorphism of the OPE algebra. Thus, the algebra obeyed by the generating functionals
\emph{is in fact an invariant of the UV sector theory}. The properties of the KS wall-crossing formula of \ref{eq:linopvar} and the fact that these generating functionals are invariants of the UV theory (i.e. of the cluster category $C(\Gamma)$), tell us that $F(L)$ are \emph{exactly the cluster character}.
The non-commutative generalization is done via the star product \cite{cordova2013line}:
$$L(\alpha, \zeta)*_yL(\alpha',\zeta)=\sum_\beta c(\alpha,\alpha',\beta,y)L(\beta,\zeta)$$
and $X_{\gamma}*_yX_{\gamma'}=y^{\vev{\gamma,\gamma'}_D}X_{\gamma+\gamma'}$. We define again the generating functions $F(L)=\sum_\gamma \bar \Omega (L,\gamma,y)X_\gamma,$ and then verify that 
$$F(L*_yL')=F(L)*_yF(L').$$
Indeed we discover that the non-commutative version of $F$ behaves exactly like a \emph{quantum cluster character} (usual cluster characters are obtained by setting $y=1$). We can thus analyze these objects from a purely algebraic point of view and study cluster algebras and cluster characters: we shall do this is the next section.
\subsection{Cluster characters}
The main references for this section are \cite{dominguez2014caldero,palu2008cluster}.
We begin by recalling some basic definitions and properties of cluster characters. Let $\mathcal C:=C(\Gamma)$ be a cluster category.
\begin{definition} A cluster character on $\mathcal C$ with values in a commutative ring $A$ is
a map
$$X : \mathrm{obj}(\mathcal C) \to A$$
such that
\begin{itemize}
\item for all isomorphic objects $L$ and $M$, we have $X(L) = X(M)$,
\item for all objects $L$ and $M$ of $\mathcal C$, we have $X(L \oplus M) = X(L)X(M)$,
\item for all objects $L$ and $M$ of $\mathcal C$ such that $\dim \mathrm{Ext}^1_\mathcal C(L, M) = 1$, we have
$$X(L)X(M) = X(B) + X(B'),$$
where $B$ and $B'$ are the middle terms of the non-split triangles
$$L \to  B \to M \to  \ \ \  \text{and} \ \ \  M \to B^\prime \to  L \to $$
with end terms L and M.\footnote{\ If $B^\prime$ does not exist, then $X(B^\prime)=1.$}
\end{itemize}
\end{definition}
Let $T=\bigoplus_i T_i$ be a cluster tilting object, and let $B=\mathrm{Ent}_\mathcal C T$. The functor $$F_T:\mathcal C \to \text{\textsf{mod}-}B,\qquad X \mapsto \mathrm{Hom}(T,X)$$
 is the projection functor that induces an equivalence between $\mathcal C/\mathsf{add}\,T[1] \to \mathsf{mod}$-$B$. Then the Caldero-Chapoton map\cite{dominguez2014caldero}, 
$$X^T_?: \mathrm{ind}\,\mathcal C \to \mathbb Q(x_1, . . . , x_n)$$
is given by
$$X^T_M =\left\{
\begin{array}{lr}
x_i \ \ \text{if} \  M \cong  \Sigma T_i\\
\sum_e \chi(\mathrm{Gr}_e F_TM)\prod^n_{i=1} x_i^{\vev{S_i,e}_D-\vev{S_i,FM}}\ \ \text{else,}
\end{array}\right.$$
where the summation is over the isoclasses of submodules\footnote{\ Recall that a module $N$ is a submodule of $M$ iff there exists an injective map $N \to M$.} of $M$ and $S_i$ are the simple B-modules. Moreover, the Euler form and Dirac form in the formula above, are those of $\mathsf{mod}$-$B$.
We now recall some properties of quiver Grassmannians\footnote{\ $\mathrm{Gr}_e(FM):=\{N \subset FM\; |\; \dim N =e\}$, i.e. it is the space of subrepresentations of $M$ with fixed dimension $e$.} and in particular their Euler Poincar\'e characteristic $\chi$ (with respect to the \'etale cohomology).
\begin{definition} Let $\Lambda$ be a finite dimensional basic $\mathbb C$-algebra. For a $\Lambda$-module $M$ we define the \emph{$F$-polynomial} to be the generating function for the
Euler characteristic of all possible quiver grassmanians, i.e.
$$F_M := \sum_e\chi(\mathrm{Gr}_e(M))y^e \in \mathbb Z[y_1, . . . , y_n]$$
where the sum runs over all possible dimension vectors of submodules of $M$.
\end{definition}
Moreover, we assume that $S_1, . . . , S_n$ is a complete system of representatives
of the simple $\Lambda$-modules, and we identify the classes $[S_i] \in K_0(\Lambda)$ with the
natural basis of $\mathbb Z^n.$
\begin{proposition}Let $\Lambda$ be a finite dimensional basic $\mathbb C$-algebra. Then the
following holds:
\begin{enumerate}
\item If $$0 \to L \overset{i}{\to} M \overset{\pi}{\to} N \to 0$$ is an Auslander-Reiten sequence in
$\Lambda$-mod, then
$$F_L  F_N = F_M + y^{\dim N }.$$
\item For the indecomposable projective $\Lambda$-module $P_i$ with top $S_i$ we have
$$F_{P_i} = F_{\mathrm{rad}\, P_i} + y^{\dim P_i}$$
for $i = 1, . . . , n.$
\item For the indecomposable injective $\Lambda$-module $I_j$ module with socle $S_j$
we have
$$F_{I_j} = y_j  F_{I_j/S_j} + 1$$
for $j = 1, 2, . . . , n.$
\end{enumerate}
\end{proposition}
The recursive relations of cluster characters and $F$-polynomials are the key tools to find a computational recipe: the next section is devoted to pointing out this algorithm. All aspects will be clarified in {\bf Example \ref{ex:cluchar}}.
\subsubsection{Computing cluster characters}
\label{sec:char}
The best way to compute cluster characters, is to exploit the results in \cite{assem2012friezes}. The idea is to associate a Laurent polynomial to a path in the quiver. If the algebra is gentle, to a path we can associate a string module: computing the cluster character associated to string modules (up to an overall monomial factor) becomes a simple combinatorics problem. For any locally finite quiver $Q$, we define a family of
matrices with coefficients in $\mathbb Z[x_Q] = \mathbb Z[x_i|i \in Q_0]$ as follows.
For any arrow $\beta \in Q_1$, we set
$$A(\beta) :=\left[\begin{array}{cc}
x_{t(\beta)} & 0\\
1 & x_{s(\beta)}
\end{array}\right]
\ \ 
\ and \ \ A(\beta^{-1}) :=\left[\begin{array}{cc}
x_{t(\beta)} & 1\\
0 & x_{s(\beta)}
\end{array}\right].$$
Let $c = c_1 ... c_n$ be a walk of length $n \geq 1$ in $Q$. For any $i \in \{0, . . . , n\}$ we set
$$v_{i+1 }= t(c_i)$$
(still with the notation $c_0 = e_{s(c)}$) and
$$V_c(i): =
\left[\begin{array}{cc}
\prod_{\alpha \in Q_1(v_i,-), \alpha \neq c_i^{\pm 1},c_{i-1}^{\pm 1}}x_{t(\alpha)} & 0\\
0 & \prod_{\alpha \in Q_1(-,v_i), \alpha \neq c_i^{\pm 1},c_{i-1}^{\pm 1}}x_{s(\alpha)}
\end{array}\right].$$
We then set
$$L_c =\frac{1}{x_{v_1}... x_{v_{n+1}}}[1, 1]V_c(1) \prod_{i=1}^nA(c_i)V_c(i + 1) \left[\begin{array}{c}1\\ 1\end{array}\right]
\in \mathcal L (x_Q).$$
If $c = e_i$ is a walk of length 0 at a point $i$, we similarly set
$$V_{e_i}(1): =
\left[\begin{array}{cc}
\prod_{\alpha \in Q_1(v_i,-)}x_{t(\alpha)} & 0\\
0 & \prod_{\alpha \in Q_1(-,v_i)}x_{s(\alpha)}
\end{array}\right].$$
and
$$L_{e_i} =\frac{1}{x_i}[1, 1]V_{e_i}\left[\begin{array}{c}1\\ 1\end{array}\right] \in \mathcal L (x_Q).$$
In other words, if $c$ is any walk, either of length zero, or of the form $c = c_1 ... c_n,$
we have
$$L_c =\frac{1}{\prod^n_{i=0} x_{t(c_i)}}[1, 1]\prod^n_{i=0}A(c_i)V_c(i + 1)\left[\begin{array}{c}1\\ 1\end{array}\right] \in \mathcal L (x_Q).$$
with the convention that $A(c_0)$ is the identity matrix. In general, we have the following result:
$$X_M=\frac{1}{x^{n_M}}L_c,$$
where $M$ is the string module associated to the path $c$ and the monomial $x^{n_M}$ is the normalization coefficient.
\begin{example}
\label{ex:cluchar}
Let us consider the cluster category of $A_4$: its AR quiver is the following
\begin{figure}[H]
\includegraphics[width=0.8\textwidth]{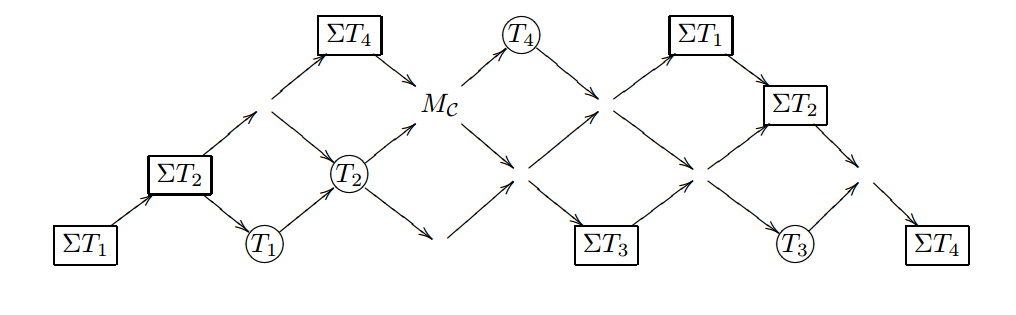}
\centering
\end{figure}
We have also made the choice of tilting objects. The algebra $\mathrm{End}\, T$ is given by the following quiver:\footnote{\ The vertices are the $T_i$ and the arrows $j \to i$ correspond to $\mathrm{Hom}_{\mathcal C}(T_i,T_j).$}
\begin{figure}[H]
\includegraphics[width=0.3\textwidth]{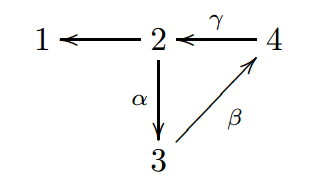}
\centering
\end{figure}
with relations $\beta \alpha = \gamma \beta = \alpha \gamma = 0.$ The Dirac form is the following:
$$\vev{-,-}_D=\left(
\begin{array}{cccc}
 0 & 1 & 0 & 0 \\
 -1 & 0 & -1 & 1 \\
 0 & 1 & 0 & -1 \\
 0 & -1 & 1 & 0 \\
\end{array}
\right),$$
whereas the Euler form is $\vev{a,b}:=\dim\mathrm{Hom}_\mathcal C(a,b)-\dim\mathrm{Hom}_\mathcal C(a,b[1])$.
Consider the B-module $F_TM=(1,1,0,0)$. Its submodules are $0, S_1$ and $F_TM$ itself. The corresponding path is just the arrow $c:2 \to 1$. By applying the formulas above we get:
\begin{eqnarray}
L_c&=&\frac{1}{x_1x_2}[1,1] A(c_0)\cdot V_c(1)\cdot A(c)\cdot V_c(2)\left[\begin{array}{c}1\\ 1\end{array}\right]\\
&=&\frac{1}{x_1x_2}[1,1] \mathrm{Id} \left[\begin{array}{cc}x_3& 0\\ 0& x_4 \end{array}\right] \left[\begin{array}{cc}x_1& 1\\ 0& x_2 \end{array}\right] \left[\begin{array}{cc}1& 0\\ 0& 1 \end{array}\right]  \left[\begin{array}{c}1\\ 1\end{array}\right]\\
&=&\frac{x_1x_3+x_4+x_2x_4}{x_1x_2}.
\end{eqnarray}
Notice that the denominator is exactly $x^{\dim FM}$: this is a general feature for the decategorification process \cite{dupont2011generic}.
Moreover, we know that the Euler characteristic of a point is 1 and thus $\chi(\mathrm{Gr}_0F_TM)=1=\chi(\mathrm{Gr}_{F_TM}F_TM)$. We then exploit the AR sequence $$0 \to S_1 \to F_TM \to S_2 \to 0$$ and get the recursive relation
$$F_{S_1}F_{S_2}=F_{F_TM}+y_2,$$
which is equivalent to the following polynomial equation:
$$1+y_1 \chi_{S_1}+y_2 \chi_{S_2}+y_1y_2 \chi_{S_1}\chi_{S_2}=1+\chi_{S_1}y_1+\chi_{FM}y_1y_2+y_2,$$
which implies that $\chi_{S_1}=\chi_{S_2}=1$ and this is consistent with the previous result. One can check this and many other computations using appendix \ref{app:mathe}.
\end{example}
\begin{remark}
One final remark is needed: we could have computed the cluster characters by a sequence of mutations of the standard seed of the cluster algebra associated to the quiver of $B=\mathrm{End}\,T$. For the non-commutative case, i.e. when $$x^\alpha x^\beta=q^{\frac{1}{2}\vev{\alpha,\beta}_D}\,x^{\alpha+\beta},$$
this procedure is the only one we know to compute quantum cluster characters. From the physics point of view, this is the important quantity: since cluster variables behave like UV line operators, they must satisfy the same non-commutative algebra.
\end{remark}

\subsection{Cluster characters and \textsc{vev}'s of UV line operators}
Let us start by recalling how the vacuum expectation values of line operators are computed in \cite{gaiotto2013framed}. The idea is associate to a loop over a punctured Gaiotto surface a product of matrices. In the case of irregular singularities, since these singularities can be understood as coming from a collision of punctures, loops can get pinched and become \emph{laminations}. Thus to both string modules (associated to \emph{laminations}) and band modules (associated to loops), we can associate a rational function in some shear variables $Y_i$. Their expression turns out to be equal to cluster characters: for string modules we can use section \ref{sec:char}, whereas for band modules we can use the bangle basis of \cite{musiker2013bases} and the multiplication formula or the Galois covering technique of \cite{cecotti2015galois}. We shall now give some detailed examples in which we apply what we just described.
\begin{example}[$A_2$ quiver]
The computations of \cite{gaiotto2013framed} of the \textsc{vev}'s of the UV line operators can be found in section 10.1. They have been made using the ``traffic rule''. The idea is to follow the lamination path and create a sequence of matrices according to the traffic rule. In the end, one takes the trace of the product of matrices (loop case) or contract the product of matrices with special vectors (open arcs case). For example,
\begin{figure}[H]
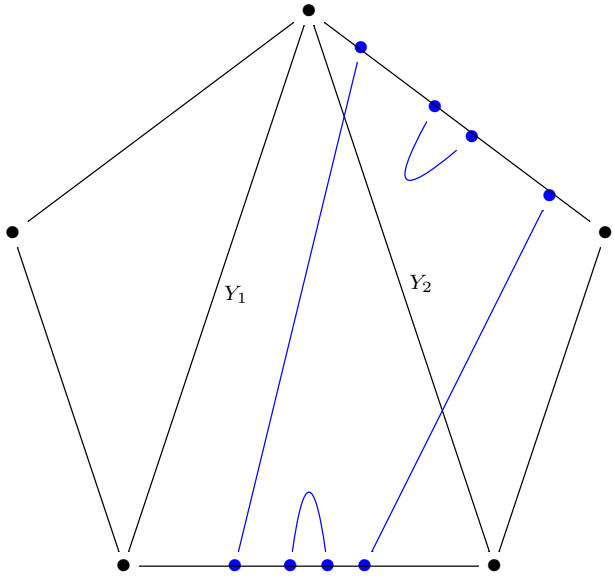

\centering
\resizebox{0.50\textwidth}{!}{\xygraph{
!{<0cm,0cm>;<0.5cm,0cm>:<0cm,0.5cm>::}
!{(0,10) }*+{\bullet}="b1"
!{(8,4) }*+{\bullet}="b3"
!{(-8,4) }*+{\bullet}="b2"
!{(-5,-5) }*+{\bullet}="b4"
!{(5,-5) }*+{\bullet}="b5"
!{(1.5,-5) }*+[blue]{\bullet}="l1"
!{(6.5,5) }*+[blue]{\bullet}="l2"
!{(1.4,9) }*+[blue]{\bullet}="l3"
!{(-2,-5) }*+[blue]{\bullet}="l4"
!{(0.5,-5) }*+[blue]{\bullet}="m1"
!{(-0.5,-5) }*+[blue]{\bullet}="m2"
!{(3.4,7.4) }*+[blue]{\bullet}="m3"
!{(4.4,6.6) }*+[blue]{\bullet}="m4"
"b1" - "b2" "b2" - "b4" "b4" - "b5" "b5" - "b3" "b3" - "b1" "b1"-"b4"^{Y_1} "b1" - "b5"^{Y_2}
"l1" -@[blue] "l2" "l3" -@[blue] "l4"
"m1" -@/_1cm/@[blue] "m2"
"m3" -@/_1cm/@[blue] "m4"
}}
\caption{The $L_1$ lamination of the $A_2$ theory.}
\end{figure}
The matrix product is the following:
$$\vev{L_1}=(B_R\cdot R\cdot M_{Y_2}\cdot L\cdot E_R) (B_R\cdot R\cdot M_{Y_2}\cdot L\cdot E_R),$$
where the matrices are
$$L=\left(\begin{array}{cc}
1 & 1\\
0 & 1
\end{array}\right)
, \ \ \  R=\left(\begin{array}{cc}
1 & 0\\
1 & 1
\end{array}\right)
, \ \ \  M_X=\left(\begin{array}{cc}
\sqrt{X} & 0\\
0 & 1/\sqrt{X}
\end{array}\right),
$$
and the vectors are
$$B_R=(1 \ 0),\ \ B_L=(0\ 1),\  \ E_L=(1 \ 0)^t, \ \ E_R=(0 \ 1).$$
Then we get:
$$R.M_{Y_2}.L=
\left(
\begin{array}{cc}
 \sqrt{Y_2} & \sqrt{Y_2} \\
 \sqrt{Y_2} & \sqrt{Y_2}+\frac{1}{\sqrt{Y_2}} \\
\end{array}
\right) ,$$
and finally
$$\vev{L_1}=\sqrt{Y_2}\sqrt{Y_2}=Y_2.$$
The other four line operators, corresponding to the four remaining indecomposable objects of the cluster category of $\mathcal{C}(\Gamma_{A_2})$ (or equivalently the remaining four cluster variables) are:
$$
\vev{L_2}= Y_1 + Y_2Y_1, \ \vev{L_3}=\frac{1}{Y_2}+\frac{Y_1}{Y_2}+ Y_1,\ \vev{L_4} = \frac{1}{Y_2}+\frac{1}{Y_2Y_1}, \ \vev{L_5} =\frac{1}{Y_1}.$$
On the other hand, the cluster characters of $A_2$ are:
$$x_1,\;x_2,\;\frac{1}{x_2}+\frac{x_1}{x_2},\;\frac{1}{x_1}+\frac{x_2}{x_1},\frac{1}{x_1x_2}+\frac{1}{x_1}+\frac{1}{x_2}.$$
The following map $(Y_1,Y_2) \mapsto (\frac{1}{x_2},x_1)$ transforms one set to the other. This map is the tropicalization map of Fock and Goncharov \cite{fock2005dual}:
$$Y_i=\prod_jx_j^{B_{ij}}.$$
Notice that this result was expected from the general algebraic properties of the line operators algebra and the cluster algebras.
\end{example}
We now proceed to a more interesting example: the pure $SU(2)$ theory. The computations of section \ref{sec:char} has to be modified a bit: as we will see, it is convenient to exploit the Galois covering techniques of \cite{cecotti2015galois}.
\begin{example}[Kronecker quiver]
Let us focus our attention to the non rigid modules, i.e. those belonging to the homogeneous tubes of the AR quiver of the cluster category $C(\Gamma_{Kr})$. There is a $\mathbb P^1$ family of these modules and, amongst them, two of them are string modules (those of the form $1\underset{1}{\overset{\lambda=0}{\Rightarrow}} 2$). By the theorems of \cite{dupont2011generic}, the cluster characters do not depend of the value of the parameter $\lambda$ and we are thus free to choose the simplest one to compute the character. Geometrically, this family of modules corresponds to loops around the cylinder (see figure \ref{fig:puresu2}). With the traffic rule techniques -- with the slight modification of taking the trace instead of using the $B$ and $E$ vectors -- we compute the \textsc{vev} of the line operator whose e.m. charge is $(1,1)$:
\begin{equation}
\vev{L_{(1,1)}}=\sqrt{Y_1Y_2}+\sqrt{\frac{Y_1}{Y_2}}+\frac{1}{\sqrt{Y_1Y_2}}.
\label{eq:yvars}
\end{equation}
We can reproduce this result using cluster characters. The only observation is that we cannot simply use section \ref{sec:char} to compute the character associated to the module $dim M=(1,1)$: we have a path ambiguity. We thus have to construct a $\Z_2$ Galois cover \cite{cecotti2015galois}, compute the character on the cover, and then project it down to the $Kr$ quiver. The reason is that the Kronecker quiver has a double arrow and we have to be able to specify the path we follow uniquely. On the $\Z_2$ cover the ambiguity is lifted and the character can be computed. The $\Z_2$ cover is:
\begin{displaymath}
    \xymatrix{  &  \bullet_3  & & & \bullet_{2'}  \\
            \bullet_1  \ar[ur]_{c} \ar[dr]&   & \bullet_2 \ar[ul] \ar[dl] & \overset{\pi}{\to} &   \\
 & \bullet_4 &  & & \bullet_{1'} \ar@{=>}[uu]}
\end{displaymath}
where the covering map $\pi$ sends $1,2 \mapsto 1^\prime$ and $3,4 \mapsto 2'$. The character corresponding to the string $c$ with respect to the covering quiver is
$$ 
\frac{1+x_1x_2+x_3x_4}{x_1x_3}
$$
Therefore, if we identify the cluster variables following the covering map $\pi$, we get
\begin{equation}
\frac{1+x_1'^2+x_2'^2}{x_1'x_2'}.
\label{eq:xvars}
\end{equation}
This result is consistent with what we find in literature (e.g. \cite{dupont2011generic}).
Also in this case, we find that the tropicalization map \footnote{\ \textit{Id est} $Y_i=\prod_jx_j^{B_{ij}}.$} sends the rational function \ref{eq:yvars} to \ref{eq:xvars}:
$$(Y_1,Y_2) \mapsto (x_2^{-2},x_1^2).$$
\end{example}
In this final example we show how to compute the cluster character associated to a band module in a more complicated quiver.
\begin{example}[$SU(2)$ with $N_f=4$]
We are interested in the module $M$ corresponding to the purple loop in the following figure
\begin{figure}[H]
\centering
\resizebox{0.52\textwidth}{!}{\xygraph{
!{<0cm,0cm>;<1.5cm,0cm>:<0cm,1.2cm>::}
!{(0,0) }*+{\bullet}="a"
!{(0,-2.5) }*+{\bullet}="b"
!{(2.5,0) }*+{\bullet}="c"
!{(2.5,-2.5)}*+{\bullet}="d"
!{(0,0.5)}*+{}="au"
!{(0,-0.5)}*+{}="ad"
!{(-0.5,0)}*+{}="al"
!{(2.5,-2)}*+{}="du"
!{(2.5,-3)}*+{}="dd"
!{(3,-2.5)}*+{}="dr"
!{(0,-1.25)}*+{}*[purple]\cir<60pt>{}="cr"
!{(1.25,-1.25) }*+{}*\cir<117pt>{}="center"
"a"-@[blue]"c"^{2} "b"-@[blue]"d"^{5} "b"-@[blue]"c"^{1} 
"c"-@/_3.5cm/@[blue]"b" ^{4}
"a" -@[blue] "b"^{3}
"c" -@[blue] "d"^{6}
}}
\end{figure}
Using the traffic rule, it is rather straightforward to compute the VEV of the line operator associated to the module $M$. The result -- which is similar to the ones computed in \cite{gaiotto2013framed} -- is:
\begin{equation}
\vev{L_M}=\mathrm{Tr}(L\cdot M_{Y_1}\cdot R\cdot M_{Y_2}\cdot R.M_{Y_4}\cdot L\cdot M_{Y_5})=\frac{1 + Y_4 + Y_2 Y_4 + Y_4 Y_5 + Y_2 Y_4 Y_5 + Y_1 Y_2 Y_4 Y_5}{\sqrt{{Y_1} {Y_2} {Y_4} {Y_5}}}
\label{eq:gaiot}
\end{equation}
The cluster character computation is more involved than the simple application of section \ref{sec:char}: we exploit the techniques of \cite{fock2005dual}. The idea is similar to the traffic rule: we find a path over the hexagonal graph of \cite{fock2005dual} that is homotopic to the path considered. Then, to each edge of the hexagonal graph we associate a matrix with the following rule:
\begin{figure}[H]
\centering
\includegraphics[width=0.8\textwidth]{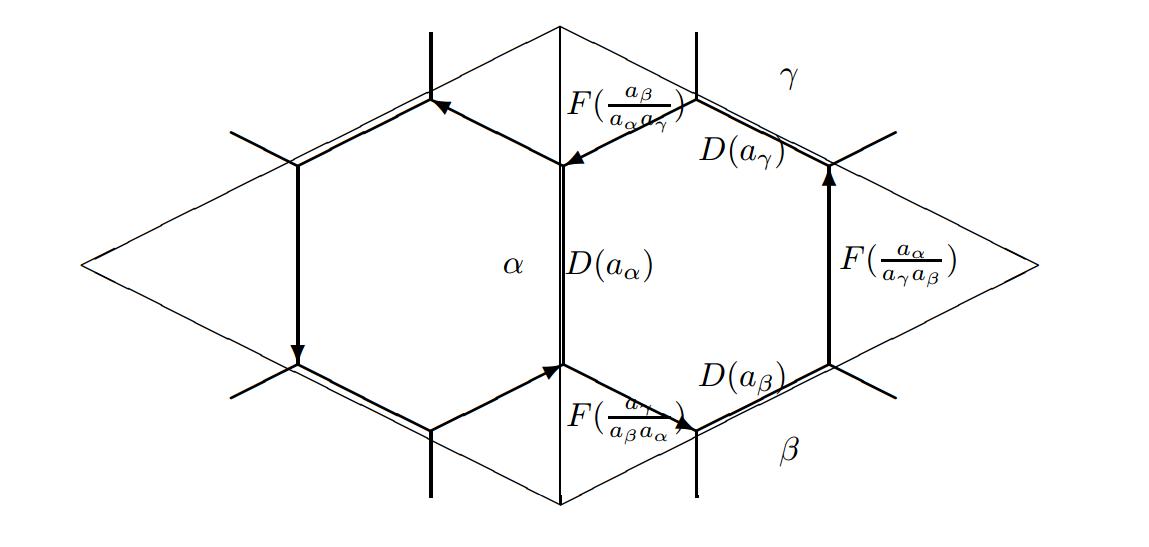}
\end{figure}
The matrices $D$ and $F$ are:
$$D(x)=\left(
\begin{array}{cc}
 0 & x \\
 -\frac{1}{x} & 0 \\
\end{array}
\right),\  \ \ F\left(\frac{x_\alpha}{x_\beta x_\gamma}\right)=\left(
\begin{array}{cc}
 1 & 0 \\
 \frac{x_\alpha}{x_\beta x_\gamma} & 1 \\
\end{array}
\right).$$
In our example we find:
$$\mathrm{Tr}\!\left(D(x_1)F(\frac{x_3}{x_1x_2})F(\frac{x_3}{x_4x_2})D(x_4)F^{-1}(\frac{x_6}{x_4x_5})F^{-1}(\frac{x_6}{x_1x_5})\right)$$
\begin{equation}
=\frac{{x_2} {x_5} {x_1}^2+{x_2} {x_5}{x_4}^2+{x_3} {x_6} {x_1^2}+{x_3} {x_6} {x_4}^2+2 x_1x_4{x_3} {x_6} }{{x_1} {x_2} {x_4} {x_5}}.
\label{eq:hexa}
\end{equation}
Again we can check that the result \eqref{eq:gaiot} is the tropicalization of \eqref{eq:hexa}.
\end{example}

\section*{Acknowlegments}
We have benefit from discussions with Michele Del Zotto, Dirk Kussin and Pierre-Guy Plamondon. SC thanks the Simons Center for Geometry and Physics, where this work was completed, for hospitality.

\appendix

\section{Code for cluster characters}
\label{app:mathe}
This is a short \textsc{Mathematica} code that computes the $L_c$ polynomials of section \ref{sec:char}. Up to an overall normalization factor, the $L_c$ polynomials are the cluster characters. The algorithm follows precisely the procedure described in section \ref{sec:char}.
\begin{verbatim}

(*set the incidence matrix of the cluster algebra*)
Dirac = {{0, 0, 1, 0, 0, 0, 0, 1}, {0, 0, 0, 1, 0, -1, 0, 0}, {-1, 0, 
   0, 1, 0, 0, 0, 0}, {0, -1, -1, 0, 0, 0, 0, 0}, {0, 0, 0, 0, 0, 
   0, -1, -1}, {0, 1, 0, 0, 0, 0, 1, 0}, {0, 0, 0, 0, 1, -1, 0, 
   0}, {-1, 0, 0, 0, 1, 0, 0, 0}}
(*the rows are the dimension vectors of the projectives*)
string = {1, 3, 4, 2};
(*=======================================================*)

Print["Pfaffian: ", Sqrt[Det[Dirac]] ];
AllArrows = Position[-Dirac, 1];
var = Table[ToExpression["x" <> ToString[i]], {i, 1, Length[Dirac]}];
(*indicare come stringa i vertici successivi raggiunti*)

arrows[string_] := 
  Table[{string[[i]], string[[i + 1]]}, {i, 1, Length[string] - 1}];
arrowsandinverse[str_] := Join[arrows[str], Reverse /@ (arrows[str])];
NoStringArrows = Complement[AllArrows, arrowsandinverse[string]];
Amat[ci_] := 
  If[MemberQ[AllArrows, 
    ci], {{var[[ ci[[2]] ]], 0}, {1, 
     var[[ ci[[1]] ]]}}, {{var[[ ci[[1]] ]], 1}, {0, 
     var[[ ci[[2]] ]]}}];
texp[str_, n_] := 
  Plus @@ (SparseArray[# -> 1, Length[Dirac]] & /@ (#[[2]] & /@ 
       Select[NoStringArrows, #[[1]] == str[[n]] &]));
sexp[str_, n_] := 
  Plus @@ (SparseArray[# -> 1, Length[Dirac]] & /@ (#[[1]] & /@ 
       Select[NoStringArrows, #[[2]] == str[[n]] &]));
Vmat[str_, 
   n_] := {{Times @@ (var^texp[str, n]), 0}, {0, 
    Times @@ (var^sexp[str, n])}};
char[str_] := 
 1/(Product[var[[ str[[i]] ]], {i, 1, Length[str]}]) ({1, 1}.Vmat[str,
      1].Dot @@ 
     Table[Amat[str[[i ;; i + 1]] ].Vmat[str, i + 1], {i, 1, 
       Length[str] - 1}].{1, 1})
charloop[str_] := 
 1/(Product[
     var[[ str[[i]] ]], {i, 1, Length[str]}]) Tr@(Vmat[str, 1].Dot @@ 
      Table[Amat[str[[i ;; i + 1]] ].Vmat[str, i + 1], {i, 1, 
        Length[str] - 1}])
(*example*)
Print["The cluster character corresponding to ", string, " is ", 
  If[Length[string] > 1 && string[[1]] == string[[Length[string]]], 
   Simplify[charloop[string]],
   Simplify[char[string]] ]];


\end{verbatim}
%
%
%
%
\section{Code for cluster automorphisms}
\label{app:mathemut}
This short \textsc{Mathematica} script is useful to find generators and relations for the automorphisms of the cluster exchange graph. The formulas used to implement the mutations for the exchange matrix $B_{ij}$ and the dimension vectors $d_l$ (where $l$ is an index that runs over the nodes) are the following:
\begin{align}
\mu_k(B)_{ij}&=\begin{cases}
- B_{ij}, & i=k \ \text{or }j=k\\
B_{ij}+\max[-B_{ik},0]\,B_{kj}+B_{ik}\,\max[B_{kj},0]\phantom{-----} & \text{otherwise.}
\end{cases}\\
\mu_k(d)_{l}&=\left\{
\begin{array}{lr}
d_{l}, & l\neq k\\
-d_{k}+\max\!\Big[\sum_i\max\!\big[B_{ik},0\big]d_{i},\sum_i\max\!\big[-B_{ik},0\big]d_i\Big] & l=k
\end{array}
\right.
\label{eq:dvecs}
\end{align} 
The procedure of this script is explained in section \ref{sec:algo}.

\begin{verbatim}

(*general functions*)
(*mutation b matrix*)
mub[b_, k_] := 
  Table[If[i == k || j == k, -b[[i, j]], 
    b[[i, j]] + Max[0, -b[[i, k]]] b[[k, j]] + 
     b[[i, k]] Max[0, b[[k, j]]]], {i, 1, Length[b]}, {j, 1, 
    Length[b]}];
(*mutation d-vectors*)
mud[d_, {b_, k_}] := 
 Table[If[l != k, 
   d[[l]], -d[[k]] + 
    Max /@ Transpose[{Sum[
        Max[b[[i, k]], 0] d[[i]], {i, 1, Length[d]}], 
       Sum[Max[-b[[i, k]], 0] d[[i]], {i, 1, Length[d]}]}]], {l, 1, 
   Length[d]}]
(*how b transforms after a sequence of mutations*)
mudseqb[seq_, b_] := 
  Thread[List[FoldList[mub, b, seq], Join[seq, {0}]]];
(*how a d-vector transforms after a sequence of mutations*)
mudseq[seqBmenoLast_, b_, d_] := Fold[mud, d, seqBmenoLast];
(*creating the permutation associated to a mutation sequence*)
PermD[seq_, b_, d_] := 
  FindPermutation[
   Plus @@ b + Sqrt[2] (Max /@ Transpose[b]) + 
    Sqrt[3] (Min /@ Transpose[b]), 
   Plus @@ Last[mudseqb[seq, b]][[1]] + 
    Sqrt[2] (Max /@ Transpose[Last[mudseqb[seq, b]][[1]] ]) + 
    Sqrt[3] (Min /@ Transpose[Last[mudseqb[seq, b]][[1]] ])];

(*composing different sequences*)
ComposizioneSequenzeConPermutazione[{seq2_, perm2_}, {seq1_, 
   perm1_}] := {Join[seq1, PermutationReplace[seq2, perm1]], 
  PermutationProduct[perm2, perm1]}
(*checking whether two b matrices are related by a permutation*)
EqualPermb[A_, B_] := 
  Expand[CharacteristicPolynomial[A, z]] == 
    Expand[CharacteristicPolynomial[B, z]] && (Sort[
      Plus @@ A + Sqrt[2] (Max /@ Transpose[A]) + 
       Sqrt[3] (Min /@ Transpose[A])] == 
     Sort[Plus @@ B + Sqrt[2] (Max /@ Transpose[B]) + 
       Sqrt[3] (Min /@ Transpose[B])]);
(*checking whether two d vectors are related by a permutation*)
EqualPermd[A_, B_] := Sort[A] == Sort[B];



(*checking the order of a sequence*)
OrdineNEW[randseqCONPerm_] := 
 Module[{ord = 0, index = 2, randseq1 = randseqCONPerm,
   randseqtemp = randseqCONPerm, b = b, d = d},
  If[EqualPermd[mudseq[Drop[mudseqb[randseq1[[1]], b], -1], b, d], d],
    ord = 1, ord = 0];
  While[ord == 0 && index <= 45, 
   If[EqualPermd[
      mudseq[Drop[
        mudseqb[ComposizioneSequenzeConPermutazione[randseqtemp, 
           randseq1][[1]], b], -1], b, d], d], ord = index; 
     randseq1 = 
      ComposizioneSequenzeConPermutazione[randseqtemp, randseq1];, 
     index++; 
     randseq1 = 
      ComposizioneSequenzeConPermutazione[randseqtemp, randseq1];];];
  ord]

(*checking the sl2Z relations for S and T generators*)
Relationssl2NEW[{Sconperm_, Tconperm_}] := 
  OrdineNEW[
     ComposizioneSequenzeConPermutazione[Sconperm, Tconperm]] == 6 && 
   EqualPermd[
    mudseq[Drop[
      mudseqb[ComposizioneSequenzeConPermutazione[Sconperm, 
         ComposizioneSequenzeConPermutazione[Sconperm, Tconperm]][[
        1]], b], -1], b, d],
    mudseq[
     Drop[mudseqb[
       ComposizioneSequenzeConPermutazione[Tconperm, 
         ComposizioneSequenzeConPermutazione[Sconperm, Sconperm]][[
        1]], b], -1], b, d] ] && 
   EqualPermd[
    mudseq[Drop[
      mudseqb[ComposizioneSequenzeConPermutazione[Sconperm, 
         ComposizioneSequenzeConPermutazione[Tconperm, 
          ComposizioneSequenzeConPermutazione[Tconperm, Tconperm]]][[
        1]], b], -1], b, d],
    mudseq[
     Drop[mudseqb[
       ComposizioneSequenzeConPermutazione[Tconperm, 
         ComposizioneSequenzeConPermutazione[Tconperm, 
          ComposizioneSequenzeConPermutazione[Tconperm, Sconperm]]][[
        1]], b], -1], b, d] ];

(*Example E7 MN*)
b = {{0, 3, -1, -1, -1, -1, -1, -1, -1}, {-3, 0, 1, 1, 1, 1, 1, 1, 
    1}, {1, -1, 0, 0, 0, 0, 0, 0, 0}, {1, -1, 0, 0, 0, 0, 0, 0, 
    0}, {1, -1, 0, 0, 0, 0, 0, 0, 0}, {1, -1, 0, 0, 0, 0, 0, 0, 
    0}, {1, -1, 0, 0, 0, 0, 0, 0, 0}, {1, -1, 0, 0, 0, 0, 0, 0, 
    0}, {1, -1, 0, 0, 0, 0, 0, 0, 0}};
d = -IdentityMatrix[9];




(*NEW Algorithm to find some generators and their order*)
Print[Dynamic[ii]];
MaxLength = 60;
MAX = 200000000;
ListAutomorph = {};
For[ii = 0, ii < MAX, ii++,
 length = RandomInteger[{1, MaxLength}];
 randseq = RandomInteger[{1, Length@b}, length];
 
 If[EqualPermb[Last[mudseqb[randseq, b]][[1]], b], index = 2; 
  randseq = {randseq, PermD[randseq, b, d]};
  randseqtemp = randseq;
  AppendTo[ListAutomorph, randseq];
  If[EqualPermd[mudseq[Drop[mudseqb[randseq[[1]], b], -1], b, d], d], 
   ord = 1, ord = 0]; 
  While[ord == 0 && index <= 19, 
   If[EqualPermd[
      mudseq[Drop[
        mudseqb[ComposizioneSequenzeConPermutazione[randseqtemp, 
           randseq][[1]], b], -1], b, d], d], ord = index; 
     randseq = 
      ComposizioneSequenzeConPermutazione[randseqtemp, randseq];, 
     index++; 
     randseq = 
      ComposizioneSequenzeConPermutazione[randseqtemp, randseq]; 
     If[! EqualPermb[Last[mudseqb[randseq[[1]], b]][[1]], b], 
      Print["Failed: "]; index = 10000;];];]; 
  If[ord != 1 , Print["Order: ", ord, "  ; Sequence: ", randseqtemp];]
  ];]

\end{verbatim}
\section{Weyl group of $E_6$}
\label{app:weyl}
With this short \emph{Mathematica} script, we explicitly construct the Weyl group of $E_6$ over the basis of simple roots.
\begin{verbatim}
n = 6;
Projectives = {{1, 1, 1, 0, 0, 0}, {0, 1, 1, 0, 0, 0}, {0, 0, 1, 0, 0,
     0}, {0, 0, 1, 1, 0, 0}, {0, 0, 1, 0, 1, 0}, {0, 0, 1, 0, 1, 1}};
(*Cartan Matrix*)
Cartan = Inverse[Projectives] + Transpose[Inverse[Projectives]];
SimpleRoots = IdentityMatrix[n];
SimpleWeylGroup = 
  Join[Table[
    IdentityMatrix[n] - 
     KroneckerProduct[SimpleRoots[[i]], SimpleRoots[[i]]].Cartan, {i, 
     1, n}], {IdentityMatrix[6]}];
sr1 = SimpleWeylGroup[[1]];
sr3 = SimpleWeylGroup[[2]];
sr3 = SimpleWeylGroup[[3]];
sr4 = SimpleWeylGroup[[4]];
sr5 = SimpleWeylGroup[[5]];
sr6 = SimpleWeylGroup[[6]];
WeylGroup = 
  FixedPoint[
   DeleteDuplicates@
     Partition[
      Partition[Flatten[Outer[Dot, SimpleWeylGroup, #, 1]], 6], 
      6] &, {IdentityMatrix[6]}, 36];
Print["Order of the group: "]
Length@WeylGroup
Print["Order of the elements: "]
MatrixOrder[M_, i0_] := 
 Module[{i = i0, Mat = M}, 
  While[MatrixPower[Mat, i] != IdentityMatrix[n], i++]; i]
DeleteDuplicates[MatrixOrder[#, 1] & /@ WeylGroup]
\end{verbatim}
We directly checked that the longest elements has length $36$, that the order of the Weyl group is
$$51840=2^7 3^45$$
and the order of each element belongs to this set:
$$\{1,3,2,5,4,6,12,8,10,9\}.$$


\begin{thebibliography}{10}

\bibitem{alday2012wilson}
L.~F.~Alday.
\newblock {W}ilson loops in supersymmetric gauge theories.
\newblock {\em Lecture Notes, CERN Winter School on Supergravity, Strings, and
  Gauge Theory}, 2012.

\bibitem{alday2010loop}
{L}uis~{F}. {A}lday, {D}avide {G}aiotto, {S}ergei {G}ukov, {Y}uji {T}achikawa,
  and {H}erman {V}erlinde.
\newblock {L}oop and surface operators in gauge theory and {L}iouville modular
  geometry.
\newblock {\em {J}ournal of {H}igh {E}nergy {P}hysics}, 2010(1):1--50, 2010.

\bibitem{alim2013bps}
{M}urad {A}lim, {S}ergio {C}ecotti, {C}lay {C}{\'o}rdova, {S}am {E}spahbodi,
  {A}shwin {R}astogi, and {C}umrun {V}afa.
\newblock {BPS} {Q}uivers and {S}pectra of {C}omplete {{N}}=2 {Q}uantum {F}ield
  {T}heories.
\newblock {\em {C}ommunications in {M}athematical {P}hysics},
  323(3):1185--1227, 2013.

\bibitem{alim2014mathcal}
{M}urad {A}lim, {S}ergio {C}ecotti, {C}lay {C}ordova, {S}am {E}spahbodi,
  {A}shwin {R}astogi, and {C}umrun {V}afa.
\newblock {N}=2 quantum field theories and their {BPS} quivers.
\newblock {\em {A}dvances in {T}heoretical and {M}athematical {P}hysics},
  18(1):27--127, 2014.

\bibitem{amiot2009cluster}
{C}laire {A}miot.
\newblock {C}luster categories for algebras of global dimension 2 and quivers
  with potential.
\newblock In {\em {A}nnales de l'{I}nstitut {F}ourier}, volume~59, page 2525,
  2009.

\bibitem{amiot2011generalized}
{C}laire {A}miot.
\newblock {O}n generalized cluster categories.
\newblock {\em {R}epresentations of algebras and related topics}, pages 1--53,
  2011.

\bibitem{aspinwall2009dirichlet}
Paul~S. Aspinwall et al.
\newblock {\em Dirichlet branes and mirror symmetry}.
\newblock{ American Mathematical Soc. Vol. 4}, 2009. 

\bibitem{aspinwall2006superpotentials}
Paul~S. Aspinwall and Lukasz~M. Fidkowski.
\newblock Superpotentials for quiver gauge theories.
\newblock {\em Journal of High Energy Physics}, 2006(10):047, 2006.

\bibitem{assem2010gentle}
{I}brahim {A}ssem, {T}homas {B}r{\"u}stle, {G}abrielle {C}harbonneau {J}odoin,
  and {P}ierre-{G}uy {P}lamondon.
\newblock {G}entle algebras arising from surface triangulations.
\newblock {\em {A}lgebra \& {N}umber {T}heory}, 4(2):201--229, 2010.

\bibitem{assem2013modules}
{I}brahim {A}ssem and {G}r{\'e}goire {D}upont.
\newblock {M}odules over cluster-tilted algebras determined by their dimension
  vectors.
\newblock {\em {C}ommunications in {A}lgebra}, 41(12):4711--4721, 2013.

\bibitem{assem2012friezes}
{I}brahim {A}ssem, {G}r{\'e}goire {D}upont, {R}alf {S}chiffler, and {D}avid
  {S}mith.
\newblock {F}riezes, strings and cluster variables.
\newblock {\em {G}lasgow {M}athematical {J}ournal}, 54(01):27--60, 2012.

\bibitem{assem2012cluster}
{I}brahim {A}ssem, {R}alf {S}chiffler, and {V}asilisa {S}hramchenko.
\newblock {C}luster automorphisms.
\newblock {\em {P}roceedings of the {L}ondon {M}athematical {S}ociety},
  104(6):1271--1302, 2012.

\bibitem{assem2006elements}
{I}brahim {A}ssem, {A}ndrzej {S}kowronski, and {D}aniel {S}imson.
\newblock {\em {E}lements of the {R}epresentation {T}heory of {A}ssociative
  {A}lgebras: {V}olume 1: {T}echniques of {R}epresentation {T}heory},
  volume~65.
\newblock {C}ambridge {U}niversity {P}ress, 2006.

\bibitem{barot2008grothendieck}
{M}ichael {B}arot, {D}irk {K}ussin, and {H}elmut {L}enzing.
\newblock {T}he {G}rothendieck group of a cluster category.
\newblock {\em {J}ournal of {P}ure and {A}pplied {A}lgebra}, 212(1):33--46,
  2008.
  
  \bibitem{barot}
  M. Barot, D. Kussin, and H. Lenzing,
 \newblock The cluster category of a canonical algebra,
   \newblock  \textit{arXiv:0801.4540.}


\bibitem{bridgeland2007stability}
{T}om {B}ridgeland.
\newblock {S}tability conditions on triangulated categories.
\newblock {\em {A}nnals of {M}athematics}, pages 317--345, 2007.

\bibitem{bridgeland2015quadratic}
{T}om {B}ridgeland and {I}van {S}mith.
\newblock {Q}uadratic differentials as stability conditions.
\newblock {\em {P}ublications math{\'e}matiques de l'{IH}{\'{E}}{S}},
  121(1):155--278, 2015.

\bibitem{brustle2011cluster}
{T}homas {B}r{\"u}stle and {J}ie {Z}hang.
\newblock {O}n the cluster category of a marked surface without punctures.
\newblock {\em {A}lgebra \& {N}umber {T}heory}, 5(4):529--566, 2011.

\bibitem{burban2008cluster}
{I}gor {B}urban, {O}samu {I}yama, {B}ernhard {K}eller, and {I}dun {R}eiten.
\newblock {C}luster tilting for one-dimensional hypersurface singularities.
\newblock {\em {A}dvances in {M}athematics}, 217(6):2443--2484, 2008.



\bibitem{caorsi2016homological}
{M}atteo {C}aorsi and {S}ergio {C}ecotti.
\newblock {H}omological {S}-{D}uality in 4d {N}=2 {QFT}s.
\newblock {\em arXiv preprint arXiv:1612.08065}, 2016.

\bibitem{Cecotti:2010qn} 
  S.~Cecotti and C.~Vafa.
 \newblock 2d Wall-Crossing, R-Twisting, and a Supersymmetric Index.
 \newblock {\em  arXiv:1002.3638 [hep-th].}

\bibitem{cecotti2010r}
  S.~Cecotti, {A}ndrew {N}eitzke, and {C}umrun {V}afa.
\newblock {R}-twisting and 4d/2d correspondences.
\newblock {\em arXiv preprint arXiv:1006.3435}, 2010.

\bibitem{cecotti2011classification}
  S.~Cecotti and C.~Vafa,
 \newblock Classification of complete N=2 supersymmetric theories in 4 dimensions,''
  \newblock {\em Surveys in differential geometry {\bf 18} }, 2013,
  [arXiv:1103.5832 [hep-th]].
  
  
  \bibitem{braidmirror}
       S. Cecotti, C. Cordova and C. Vafa.
 \newblock Braids, Walls, and Mirrors,
 \newblock {\em  arXiv:1110.2115 [hep-th].}
  
  \bibitem{Cecotti:2012gh} 
  S.~Cecotti and M.~Del Zotto.
 \newblock 4d N=2 Gauge Theories and Quivers: the Non-Simply Laced Case,
  \newblock {\em  JHEP {\bf 1210}, 190 }, 2012,
  [arXiv:1207.7205 [hep-th]].


\bibitem{cecotti2012quiver}
  S.~Cecotti.
 \newblock The quiver approach to the BPS spectrum of a 4d N=2 gauge theory,
 \newblock {\em  Proc.\ Symp.\ Pure Math.\  {\bf 90}, 3 }, 2015,
  [arXiv:1212.3431 [hep-th]].

\bibitem{cecotti2013categorical}
 S.~Cecotti.
 \newblock Categorical Tinkertoys for N=2 Gauge Theories.
 \newblock {\em  Int.\ J.\ Mod.\ Phys.\ A {\bf 28}, 1330006 }, 2013,
  [arXiv:1203.6734 [hep-th]].
  

  
  \bibitem{cecotti2013more}
  S.~Cecotti, M.~Del Zotto and S.~Giacomelli.
 \newblock More on the N=2 superconformal systems of type $D_p(G)$,
 \newblock {\em  JHEP {\bf 1304}, 153 }, 2013,
  [arXiv:1303.3149 [hep-th]].


\bibitem{cecotti2014systems}
  S.~Cecotti and M.~Del Zotto,
 \newblock $Y$-systems, $Q$-systems, and 4D $\mathcal{N}=2$ supersymmetric QFT,
 \newblock {\em  J.\ Phys.\ A {\bf 47}, no. 47, 474001}, 2014,
  [arXiv:1403.7613 [hep-th]].

\bibitem{cecotti2015galois}
S.~Cecotti and M.~Del Zotto.
 \newblock Galois covers of $\mathcal{N}=2$ BPS spectra and quantum monodromy,
 \newblock {\em  Adv.\ Theor.\ Math.\ Phys.\  {\bf 20}, 1227 }, 2016,
  [arXiv:1503.07485 [hep-th]].


\bibitem{cecotti2015higher}
  S.~Cecotti and M.~Del Zotto,
 \newblock Higher S-dualities and Shephard-Todd groups,
 \newblock{\em  JHEP {\bf 1509}, 035}, 2015,
  [arXiv:1507.01799 [hep-th]].
  
  \bibitem{toappear}
S. Cecotti and M. Del Zotto,
 \newblock to appear.


\bibitem{chen2014tilting}
Jianmin Chen, Yanan Lin, and Shiquan Ruan.
\newblock Tilting objects in the stable category of vector bundles on a
  weighted projective line of type (2, 2, 2, 2; $\lambda$).
\newblock {\em Journal of Algebra}, 397:570--588, 2014.

\bibitem{cordova2013line}
{C}lay {C}{\'o}rdova and {A}ndrew {N}eitzke.
\newblock {L}ine defects, tropicalization, and multi-centered quiver quantum
  mechanics.
\newblock {\em arXiv preprint arXiv:1308.6829}, 2013.

\bibitem{dehy2008combinatorics}
Raika Dehy and Bernhard Keller.
\newblock {O}n the {C}ombinatorics of {R}igid {O}bjects in 2--{C}alabi--{Y}au
  {C}ategories.
\newblock {\em International Mathematics Research Notices}, 2008, 2008.

\bibitem{del2013four}
{M}ichele {D}el {Z}otto.
\newblock {\em {F}our-dimensional {N}=2 superconformal quantum field theories
  and {BPS}-quivers}.
\newblock PhD thesis, {S}cuola {I}nternazionale {S}uperiore di {S}tudi
  {A}vanzati, 2013.

\bibitem{del2014absence}
Michele Del~Zotto and Ashoke Sen.
\newblock About the absence of exotics and the {C}oulomb branch formula.
\newblock {\em arXiv preprint arXiv:1409.5442}, 2014.

\bibitem{Denef:2002ru} 
  F.~Denef,
 \newblock Quantum quivers and Hall / hole halos,
 \newblock {\em JHEP {\bf 0210}, 023 }, 2002, 
  [hep-th/0206072].

\bibitem{derksen2008quivers}
Harm Derksen, Jerzy Weyman, and Andrei Zelevinsky.
\newblock Quivers with potentials and their representations {I}: {M}utations.
\newblock {\em Selecta Mathematica, New Series}, 14(1):59--119, 2008.

\bibitem{kedem2}
P. Di Francesco and R. Kedem.
\newblock Q-systems, heaps, paths and cluster positivity.
\newblock{\em Comm. Math. Phys. 293 (2009) 727802, arXiv:0811.3027 [math.CO]}

\bibitem{DiFrancesco:2008mc} 
  P.~Di Francesco and R.~Kedem,
\newblock Q-systems as cluster algebras II: Cartan matrix of finite type and the polynomial property,
\newblock {\em Lett.\ Math.\ Phys.\  {\bf 89}, 183 (2009)
  doi:10.1007/s11005-009-0354-z
  [arXiv:0803.0362 [math.RT]].}

\bibitem{dimo1}
    T. Dimofte,
\newblock Duality Domain Walls in Class $\mathcal{S}[A_1]$,
\newblock {\em Proc.\ Symp.\ Pure Math.\  {\bf 88}, 271}, 2014.


\bibitem{dimo2}
    T. Dimofte, D. Gaiotto and R. van der Veen.
\newblock RG Domain Walls and Hybrid Triangulations.
\newblock{\em  Adv.\ Theor.\ Math.\ Phys.\  {\bf 19}, 137 (2015)
  [arXiv:1304.6721 [hep-th]].}

\bibitem{dominguez2014caldero}
Salom{\'o}n Dominguez and Christof Geiss.
\newblock {A} {C}aldero--{C}hapoton formula for generalized cluster categories.
\newblock {\em Journal of Algebra}, 399:887--893, 2014.

\bibitem{drinfeld2004dg}
{V}ladimir {D}rinfeld.
\newblock {DG} quotients of {DG} categories.
\newblock {\em {J}ournal of {A}lgebra}, 272(2):643--691, 2004.

\bibitem{drukker2009loop}
{N}adav {D}rukker, {D}avid~{R} {M}orrison, and {T}akuya {O}kuda.
\newblock {L}oop operators and {S}-duality from curves on {R}iemann surfaces.
\newblock {\em {J}ournal of {H}igh {E}nergy {P}hysics}, 2009(09):031, 2009.

\bibitem{dupont2011generic}
{G}r{\'e}goire {D}upont.
\newblock {G}eneric variables in acyclic cluster algebras.
\newblock {\em {J}ournal of {P}ure and {A}pplied {A}lgebra}, 215(4):628--641,
  2011.
  
  \bibitem{bookmap}
  B. Farb and D. Margalit.
\newblock A Primer on Mapping Class Groups,
\newblock{\em Princeton University Press}, 2012.

\bibitem{fock2009cluster}
V.~Fock and A.~Goncharov.
\newblock {C}luster ensembles, quantization and the dilogarithm.
\newblock {\em arXiv preprint math.AG/0311245}.

\bibitem{fock2005dual}
{V}. {F}ock and {A.} {G}oncharov.
\newblock {D}ual {T}eichmuller and lamination spaces.
\newblock {\em arXiv preprint math/0510312}, 2005.

\bibitem{fomin2008cluster}
Sergey Fomin, Michael Shapiro, and Dylan Thurston.
\newblock Cluster algebras and triangulated surfaces. {P}art {I}: {C}luster
  complexes.
\newblock {\em Acta Mathematica}, 201(1):83--146, 2008.

\bibitem{fomin2002cluster}
Sergey Fomin and Andrei Zelevinsky.
\newblock {C}luster algebras {I}: foundations.
\newblock {\em Journal of the American Mathematical Society}, 15(2):497--529,
  2002.

\bibitem{gaiotto2012n}
{D}avide {G}aiotto.
\newblock {N}=2 dualities.
\newblock {\em {J}ournal of {H}igh {E}nergy {P}hysics}, 2012(8):1--58, 2012.


\bibitem{gaiotto2013wall}
{D}avide {G}aiotto, {G}regory~{W} {M}oore, and {A}ndrew {N}eitzke.
\newblock {W}all-crossing, {H}itchin systems, and the {WKB} approximation.
\newblock {\em {A}dvances in {M}athematics}, 234:239--403, 2013.


\bibitem{gaiotto2013framed}
{D}avide {G}aiotto, {G}regory~{W} {M}oore, and {A}ndrew {N}eitzke.
\newblock {F}ramed {BPS} states.
\newblock {\em {A}dvances in {T}heoretical and {M}athematical {P}hysics},
  17(2):241--397, 2013.
  
\bibitem{gaiotto2014open}
{D}avide {G}aiotto.
\newblock {O}pen {V}erlinde line operators.
\newblock {\em arXiv preprint arXiv:1404.0332}, 2014.


\bibitem{geigle}
W. Geigle and H. Lenzing.
\newblock A class of weighted projective lines arising in representation theory of finite dimensional algebras,
\newblock {\em Springer Lecture Notes in Mathematics \textbf{1273} (1987) 265-297.}


\bibitem{ginzburg2006calabi}
{V}ictor {G}inzburg.
\newblock {C}alabi-{Y}au algebras.
\newblock {\em arXiv preprint math/0612139}, 2006.

\bibitem{goncharov2016ideal}
{A}lexander~{B.} {G}oncharov.
\newblock {I}deal webs, moduli spaces of local systems, and 3d {C}alabi-{Y}au
  categories.
\newblock {\em arXiv preprint arXiv:1607.05228}, 2016.

\bibitem{GH}
P. Griffiths and J. Harris.
\newblock {\em Principles of Algebraic Geometry},
\newblock Wiley (1994).

\bibitem{hooft1978phase}
G.~'t Hooft.
\newblock {O}n the phase transition towards permanent quark confinement.
\newblock {\em {N}uclear {P}hysics: {B}}, 138(1):1--25, 1978.

\bibitem{hooft1979property}
G.~'t Hooft.
\newblock A property of electric and magnetic flux in non-{A}belian gauge
  theories.
\newblock {\em Nuclear physics: B}, 153:141--160, 1979.

\bibitem{hooft1980confinement}
G.~'t Hooft.
\newblock Confinement and topology in non-abelian gauge theories.
\newblock {\em Acta physica Austriaca: Supplementum}, 22:531--586, 1980.

\bibitem{hooft1980topological}
G.~'t Hooft.
\newblock Which topological features of a gauge theory can be responsible for
  permanent confinement?
\newblock {\em Recent developments in gauge theories. Proceedings of the NATO
  Advanced Study Institute on recent developments in gauge theories, held in
  Carg{\`e}se, Corsica, August 26-September 8, 1979}, pages 117--133, 1980.

\bibitem{iyama2008mutation}
{O}samu {I}yama and {Y}uji {Y}oshino.
\newblock {M}utation in triangulated categories and rigid {C}ohen--{M}acaulay
  modules.
\newblock {\em {I}nventiones mathematicae}, 172(1):117--168, 2008.

\bibitem{kapustin2008homological}
Anton Kapustin, Maximilian Kreuzer, and Karl-Georg Schlesinger.
\newblock {\em {H}omological mirror symmetry: new developments and
  perspectives}, volume 757.
\newblock Springer, 2008.

\bibitem{keller2005triangulated}
{B}ernhard {K}eller.
\newblock {O}n triangulated orbit categories.
\newblock {\em {D}oc. {M}ath}, 10(551-581):21--56, 2005.

 \bibitem{kedem1}
 R. Kedem, 
\newblock Q-systems as cluster algebras. 
\newblock{\em J. Phys. A: Math. Theor. 41 (2008) 194011 (14 pages). arXiv:0712.2695 [math.RT]}
  


\bibitem{keller2006differential}
{B}ernhard {K}eller.
\newblock {O}n differential graded categories.
\newblock {\em arXiv preprint math/0601185}, 2006.

\bibitem{keller2007derived}
{B}ernhard {K}eller.
\newblock {D}erived categories and tilting.
\newblock {\em {L}ondon {M}athematical {S}ociety {L}ecture {N}ote {S}eries},
  332:49, 2007.

\bibitem{keller2011cluster}
{B}ernhard {K}eller.
\newblock {O}n cluster theory and quantum dilogarithm identities.
\newblock In {\em {R}epresentations of {A}lgebras and {R}elated {T}opics,
  {E}ditors {A}. {S}kowronski and {K}. {Y}amagata, {EMS} {S}eries of {C}ongress
  {R}eports, {E}uropean {M}athematical {S}ociety}, pages 85--11, 2011.

\bibitem{keller2007cluster}
{B}ernhard {K}eller and {I}dun {R}eiten.
\newblock {C}luster-tilted algebras are {G}orenstein and stably
  {C}alabi--{Y}au.
\newblock {\em {A}dvances in {M}athematics}, 211(1):123--151, 2007.

\bibitem{keller2011deformed}
{B}ernhard {K}eller and {M}ichel {V}an~den {B}ergh.
\newblock {D}eformed {C}alabi--{Y}au completions.
\newblock {\em {J}ournal f{\"u}r die reine und angewandte {M}athematik
  ({C}relles {J}ournal)}, 2011(654):125--180, 2011.

\bibitem{keller2011derived}
{B}ernhard {K}eller and {D}ong {Y}ang.
\newblock {D}erived equivalences from mutations of quivers with potential.
\newblock {\em {A}dvances in {M}athematics}, 226(3):2118--2168, 2011.

\bibitem{kontsevich1994homological}
Maxim Kontsevich.
\newblock {H}omological algebra of mirror symmetry.
\newblock {\em arXiv preprint alg-geom/9411018}, 1994.

\bibitem{Kontsevich:2008fj} 
  M.~Kontsevich and Y.~Soibelman.
 \newblock Stability structures, motivic Donaldson-Thomas invariants and cluster transformations,
  \newblock {\em arXiv:0811.2435 [math.AG].}

\bibitem{kontsevich2014wall}
Maxim Kontsevich and Yan Soibelman.
\newblock {W}all-crossing structures in {D}onaldson--{T}homas invariants,
  integrable systems and mirror symmetry.
\newblock In {\em Homological mirror symmetry and tropical geometry}, pages
  197--308. Springer, 2014.

\bibitem{kussin2013triangle}
{D}irk {K}ussin, {H}elmut {L}enzing, and {H}agen {M}eltzer.
\newblock {T}riangle singularities, {ADE}-chains, and weighted projective
  lines.
\newblock {\em {A}dvances in {M}athematics}, 237:194--251, 2013.

\bibitem{labardini2008quivers}
Daniel Labardini-Fragoso.
\newblock {Q}uivers with potentials associated to triangulated surfaces.
\newblock {\em Proceedings of the London Mathematical Society}, 98(3):797--839,
  2008.
 
 \bibitem{lenzing1} 
  H. Lenzing, 
\newblock Hereditary categories, in \textit{Handbook of Tilting Theory,} eds. L. Angeleri H\"ugel, D. Happel and H. Krause, \textsc{London Mathematical Society Lecture Note Series} \textbf{332} CUP (2007), pages 105-146.
  
  \bibitem{miyaki}
  J.-I. Miyachi and A. Yekutieli, ``Derived Picard groups of finite dimensional hereditary algebras'', \textit{arXiv:math/9904006.}

\bibitem{musiker2013bases}
{G}regg {M}usiker, {R}alf {S}chiffler, and {L}auren {W}illiams.
\newblock {B}ases for cluster algebras from surfaces.
\newblock {\em {C}ompositio {M}athematica}, 149(02):217--263, 2013.

\bibitem{neeman2014triangulated}
{A}mnon {N}eeman. 
\newblock {\em {T}riangulated {C}ategories. ({AM}-148)}, volume 148.
\newblock {P}rinceton {U}niversity {P}ress, 2014.

\bibitem{intematrix}
M. Newman.
\newblock {\em Integral Matrices}.
\newblock{ Academic Press, New York}, 1972



\bibitem{palu2008cluster}
{Y}ann {P}alu.
\newblock {C}luster characters for 2-{C}alabi--{Y}au triangulated categories.
\newblock In {\em {A}nnales de l'institut {F}ourier}, volume~58, pages
  2221--2248, 2008.

\bibitem{palu2009grothendieck}
Yann Palu.
\newblock {G}rothendieck group and generalized mutation rule for
  2--{C}alabi--{Y}au triangulated categories.
\newblock {\em Journal of Pure and Applied Algebra}, 213(7):1438--1449, 2009.

\bibitem{plamondon2011cluster}
{P}ierre-{G}uy {P}lamondon.
\newblock {C}luster algebras via cluster categories with infinite-dimensional
  morphism spaces.
\newblock {\em {C}ompositio {M}athematica}, 147(06):1921--1954, 2011.

\bibitem{qiu2016decorated}
{Y}u {Q}iu.
\newblock {D}ecorated marked surfaces: spherical twists versus braid twists.
\newblock {\em {M}athematische {A}nnalen}, 365(1-2):595--633, 2016.

\bibitem{qiu2013cluster}
{Y}u {Q}iu and {Y}u {Z}hou.
\newblock {C}luster categories for marked surfaces: punctured case.
\newblock {\em arXiv preprint arXiv:1311.0010}, 2013.

\bibitem{qiu2014decorated}
{Y}u {Q}iu and {Y}u {Z}hou.
\newblock {D}ecorated marked surfaces {II}: {I}ntersection numbers and
  dimensions of {H}oms.
\newblock {\em arXiv preprint arXiv:1411.4003}, 2014.

\bibitem{reiten2010tilting}
{I}dun {R}eiten.
\newblock {T}ilting theory and cluster algebras.
\newblock {\em arXiv preprint arXiv:1012.6014}, 2010.

\bibitem{segal2016all}
{E}d {S}egal.
\newblock {A}ll autoequivalences are spherical twists.
\newblock {\em arXiv preprint arXiv:1603.06717}, 2016.

\bibitem{seiberg1995electric}
Nathan Seiberg.
\newblock {E}lectric-magnetic duality in supersymmetric non-{A}belian gauge
  theories.
\newblock {\em Nuclear Physics B}, 435(1-2):129--146, 1995.

\bibitem{seiberg1994monopoles}
{N}athan {S}eiberg and {E}dward {W}itten.
\newblock {M}onopoles, duality and chiral symmetry breaking in {N}=2
  supersymmetric {QCD}.
\newblock {\em {N}uclear {P}hysics {B}}, 431(3):484--550, 1994.

\bibitem{seidel2001braid}
{P}aul {S}eidel and {R}ichard {T}homas.
\newblock {B}raid group actions on derived categories of coherent sheaves.
\newblock {\em {D}uke {M}athematical {J}ournal}, 108(1):37--108, 2001.

\bibitem{shapere1999bps}
Alfred~D. Shapere and Cumrun Vafa.
\newblock {B}{P}{S} structure of {A}rgyres-{D}ouglas superconformal theories.
\newblock {\em arXiv preprint hep-th/9910182}, 1999.

\bibitem{terash}
   Y.~Terashima and M.~Yamazaki,
\newblock SL(2,R) Chern-Simons, Liouville, and Gauge Theory on Duality Walls,
\newblock{\em JHEP {\bf 1108}, 135},2011, 
  [arXiv:1103.5748 [hep-th]].
  
\bibitem{Witten:1979ey} 
  E.~Witten,
\newblock Dyons of Charge $e \theta/2 \pi$,
\newblock{\em  Phys.\ Lett.\  {\bf 86B}, 283 }, 1979.

\bibitem{xie2012network}
{D}an {X}ie.
\newblock {N}etwork, {C}luster coordinates and {N}=2 theory {I}.
\newblock {\em arXiv preprint arXiv:1203.4573}, 2012.

\end{thebibliography}

\end{document}